\documentclass[11pt]{article}
\usepackage[english]{babel}
\usepackage{amsmath,amsfonts,amssymb,amsthm}
\usepackage[left=3cm,top=2.5cm,right=3cm,bottom=2.5cm,headheight=1.25cm]{geometry}
\usepackage{setspace}
\usepackage{graphicx}
\usepackage{natbib}
\usepackage{bm}
\usepackage{multirow}
\usepackage{multicol}
\usepackage{indentfirst}
\usepackage{array}
\usepackage{float}
\usepackage[pdftex,breaklinks,plainpages=true,pdfpagelabels,colorlinks=true,citecolor=blue,linkcolor=blue,urlcolor=blue,filecolor=blue,bookmarksopen=true,pagebackref=false]{hyperref} 
\usepackage{comment}
\usepackage{enumitem}
\newcolumntype{L}{>{$}l<{$}} 
\newcolumntype{C}{>{$}c<{$}} 
\newcolumntype{R}{>{$}r<{$}} 
\graphicspath{{./figures/}}  
\newtheorem{theorem}{Theorem}
\newtheorem{proposition}{Proposition}
\newtheorem{corollary}{Corollary}
\newtheorem{definition}{Definition}
\newtheorem{lemma}{Lemma}
\AtBeginDocument{} 

\title{Bivariate Simplex Distribution}
\author{
  Emerson A. Alves\textsuperscript{+}, 
  Lucas S. Vieira\textsuperscript{+}, 
  Lizandra C. Fabio\textsuperscript{+}, \\ 
  Vanessa Barros\textsuperscript{+}, and 
  Jalmar M. F. Carrasco\textsuperscript{+}\thanks{CONTACT Author. Email: carrascojalmar@gmail.com} \\
  \textsuperscript{+}Federal University of Bahia, Institute of Mathematics and Statistics, \\ 
  Bahia, Brazil
}
\date{}

\begin{document}

\maketitle

\begin{abstract}
This article proposes a bivariate Simplex distribution for modeling continuous outcomes constrained to the interval $(0,1)$, which can represent proportions, rates, or indices. We derive analytical expressions to calculate the dependence between the variables and examine its relationship with the association parameter. Parameters are estimated using the maximum likelihood method, and their performance is assessed through Monte Carlo simulations. The simulations explore various aspects of the bivariate distribution, including different surfaces and contour graphs. To illustrate the proposed model's methodology and properties, we present an application in the Jurimetric area. 
\\
\vspace{0.5cm} 
\end{abstract}

\section{Introduction}
\label{sec:int}

Relevant scientific studies have provided data describing intrinsic phenomena regarding rates, fractions, proportions, or indices. For instance, the distribution supported on the interval (0,1) plays a crucial role in research and application in finance \citet{gomez2014log} and \citet{biswas2021estimation}, public health \cite{mazucheli2019one} and \cite{biswas2019r} and demographics \citet{andreopoulos2019mortality}. In this context, the Beta and Simplex distributions are particularly prominent, with their density functions capable of assuming different shapes depending on parameter values. 

Bivariate distributions are essential in practice because they allow simultaneous analysis and decision-making regarding two related or dependent variables. Methods for constructing joint distributions for random variables can be found in \citet{lai2009continuous} and \citet{kotz2004continuous}. Specifically, for the Beta distribution, \citet{barros2015estimaccao} proposed parameter estimation methods for the bivariate Beta distribution introduced by \citet{nadarajah2005some}. Further studies \cite{arnold2011flexible}, who explored the bivariate Beta distributions for correlated data, and \cite{gupta1985three}, who examined two bivariate Beta distributions derived from the Morgenstern curves and the bivariate Dirichlet distribution, respectively. Other notable contributions include \cite{sarabia2006bivariate}, who studied various bivariate extensions under a Bayesian framework, and \cite{olkin2003bivariate} who demonstrated a positively dependent bivariate Beta distribution via the likelihood ratio. The extension of the Beta distribution to the multivariate case $(0,1)^s$ was investigated by \cite{jones2002multivariate}, while  \cite{moschen2023bivariate} analyzed the Beta distribution proposed by \cite{olkin2015constructions} using both classic and Bayesian approach. 

Despite these advances, the bivariate Simplex distribution remains relatively unexplored for modeling the distribution of two proportions, such as the proportion of budget allocated to different sectors. Bivariate distributions are often constructed using copula functions, which allow the analysis of dependence structures between two random variables independently of their marginal distributions. This approach offers flexibility in combining different types of marginal distributions. Therefore, this article proposes deriving the bivariate Simplex distribution via copulas\footnote{A copula describes a joint distribution function in terms of its marginals and is widely used in empirical analysis across various fields, including survival analysis, actuarial sciences, marketing, medical statistics, and econometrics.} as an alternative method for analyzing bivariate data constrained to the standard unit interval. This contribution is significant for the Simplex distribution framework. 

The article is structured as follows: Section \ref{sec:pre} reviews the properties and inferential processes associated with the Simplex distribution. Section \ref{sec:bSimplex} introduces the bivariate Simplex distribution via copulas, develops analytical expressions for calculating dependence between variables, and provides estimators using the maximum likelihood method. A Monte Carlo simulation study is conducted to investigate the asymptotic behavior of these estimators. Section \ref{sec:apli} applies the methodology to a real dataset from the Jurimetry field to validate the results. Finally, Section \ref{sec:conclu} summarizes the conclusions. 

\section{Preliminaries}
\label{sec:pre}
The Simplex distribution was proposed by \citet{barndorff1991some}; later introduced into a class of dispersion models by \citet{jorgensen1997theory}, which extended the generalized linear models (GLMs) \citep{nelder1972generalized}. The Simplex distribution is very convenient and flexible regarding data restricted to the continuous unit interval (0,1), which can be interpreted as proportions, rates, or indices. Let $y$ be a random variable that follows a Simplex distribution, with parameters $\mu \in (0,1)$ and $\sigma^2>0$. The probability density function (pdf) of this distribution is given by 
\begin{eqnarray}
\label{fdensid2}
f(y;\mu;\sigma^2)=\{2\pi\sigma^2[y(1-y)]^3\}^{-1/2}\exp\Big\{-\frac{1}{2\sigma^2}d(y;\mu)\Big\}, 
\end{eqnarray}
where $0<y<1$ and $d(y;\mu)=(y-\mu)^2/y(1-y)\mu^2(1-\mu)^2$ is the unit deviation. The expected value and variance of $Y$ are given by ${\rm E}(Y)=\mu$  and 
\begin{eqnarray*}
{\rm Var}(Y)=\mu(1-\mu)-\sqrt{\frac{1}{2\sigma^2}}\exp\Big\{\frac{1}{2\sigma^2\mu^2(1-\mu)^2}\Big\}\Gamma\Big\{\frac{1}{2},\frac{1}{2\sigma^2\mu^2(1-\mu)^2}\Big\},
\end{eqnarray*}
where $\Gamma(a,b)=\int_b^\infty x^{a-1} e^{-x} dx$ is the incomplete gamma function. In addition, the variance function is given by $V(\mu)=\mu^3(1-\mu)^3$. 
 
The Simplex distribution can take several shapes by assuming different values for the parameters ($\mu$, $\sigma^2$), as shown in Figure~(\ref{fig:dSimplex}).
\begin{figure}[h!]
    \centering
    \begin{minipage}[!]{0.32\linewidth}
        \includegraphics[width = 5cm]{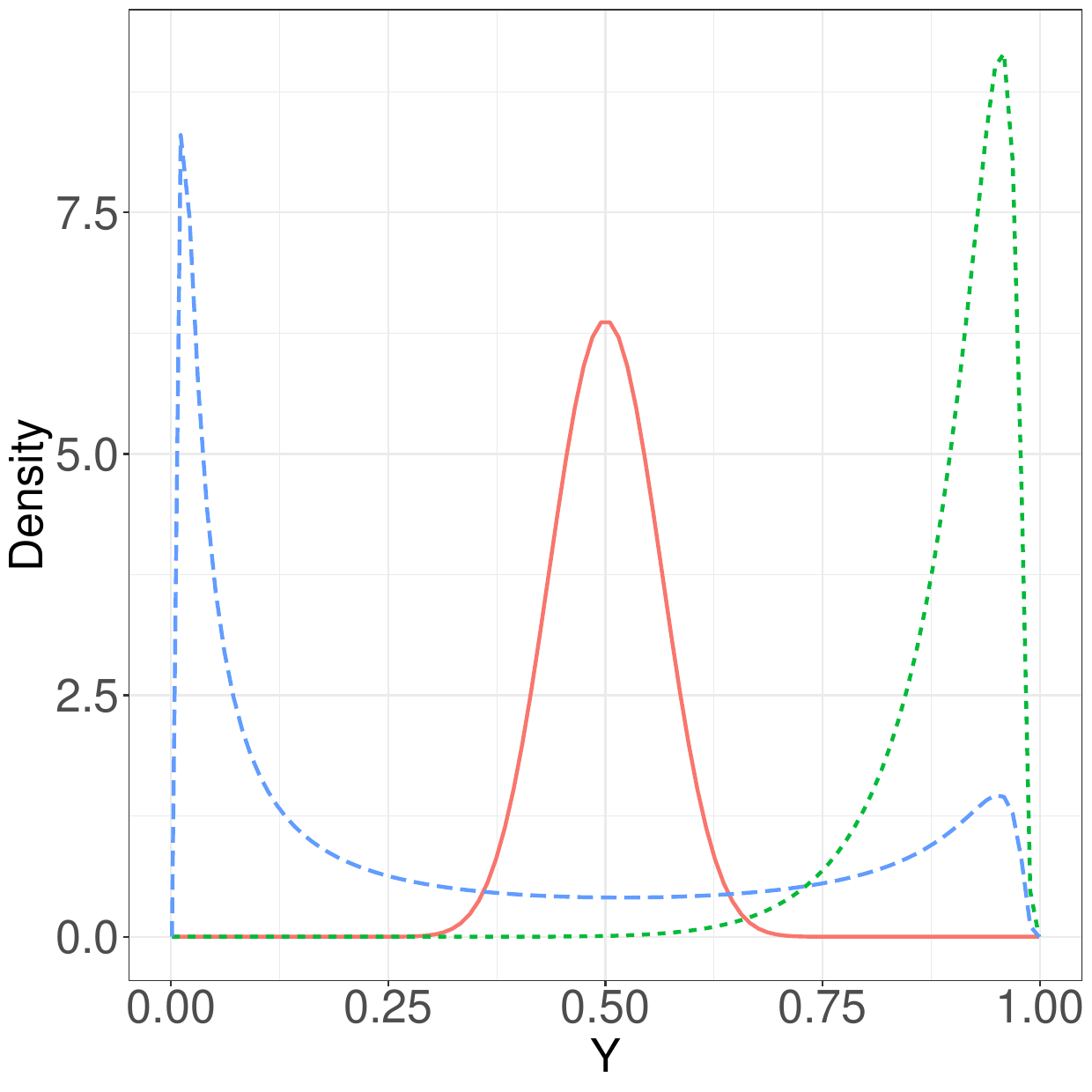}
    \end{minipage}
    \begin{minipage}[!]{0.32\linewidth}
        \includegraphics[width = 5cm]{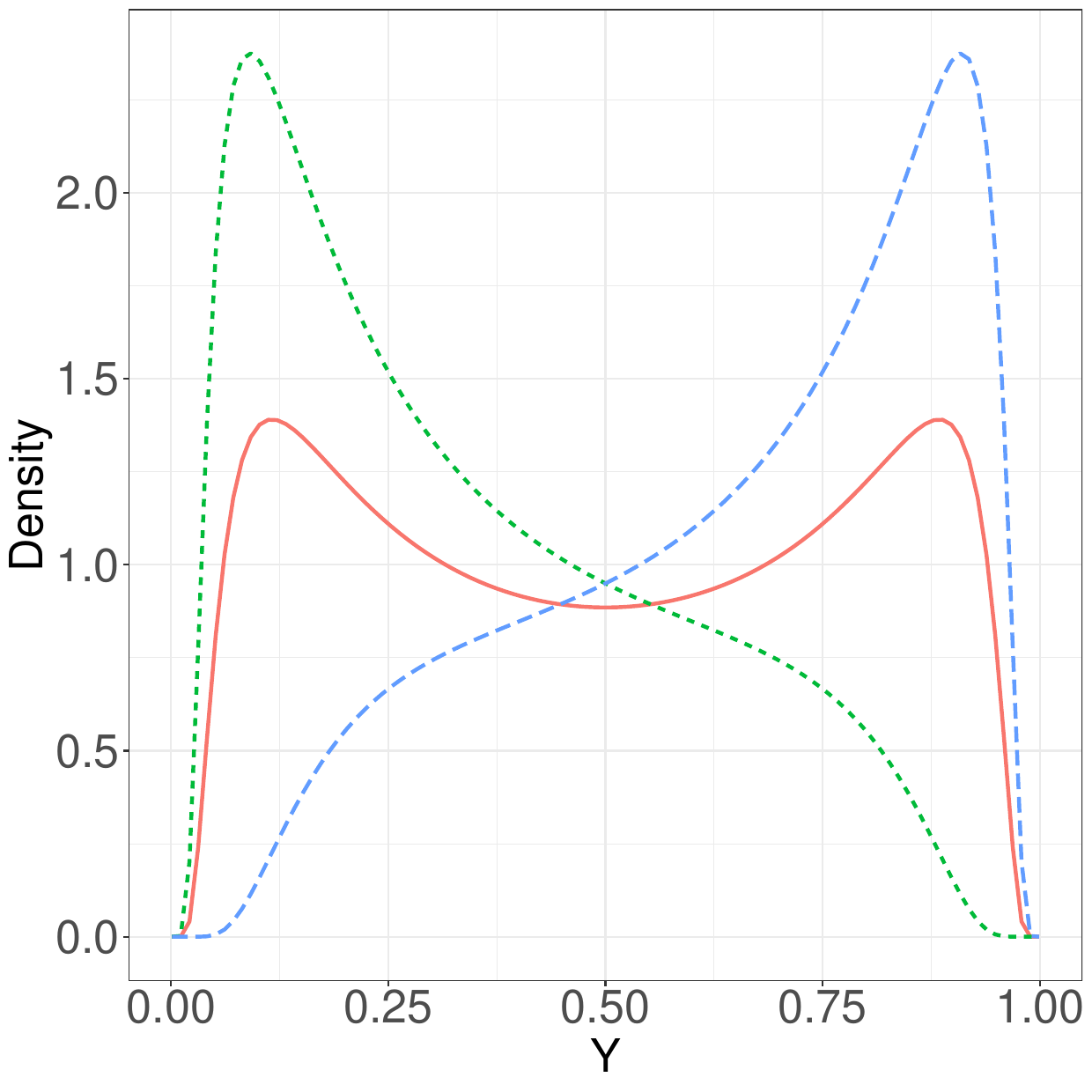}
    \end{minipage}
    \begin{minipage}[!]{0.32\linewidth}
        \includegraphics[width = 5cm]{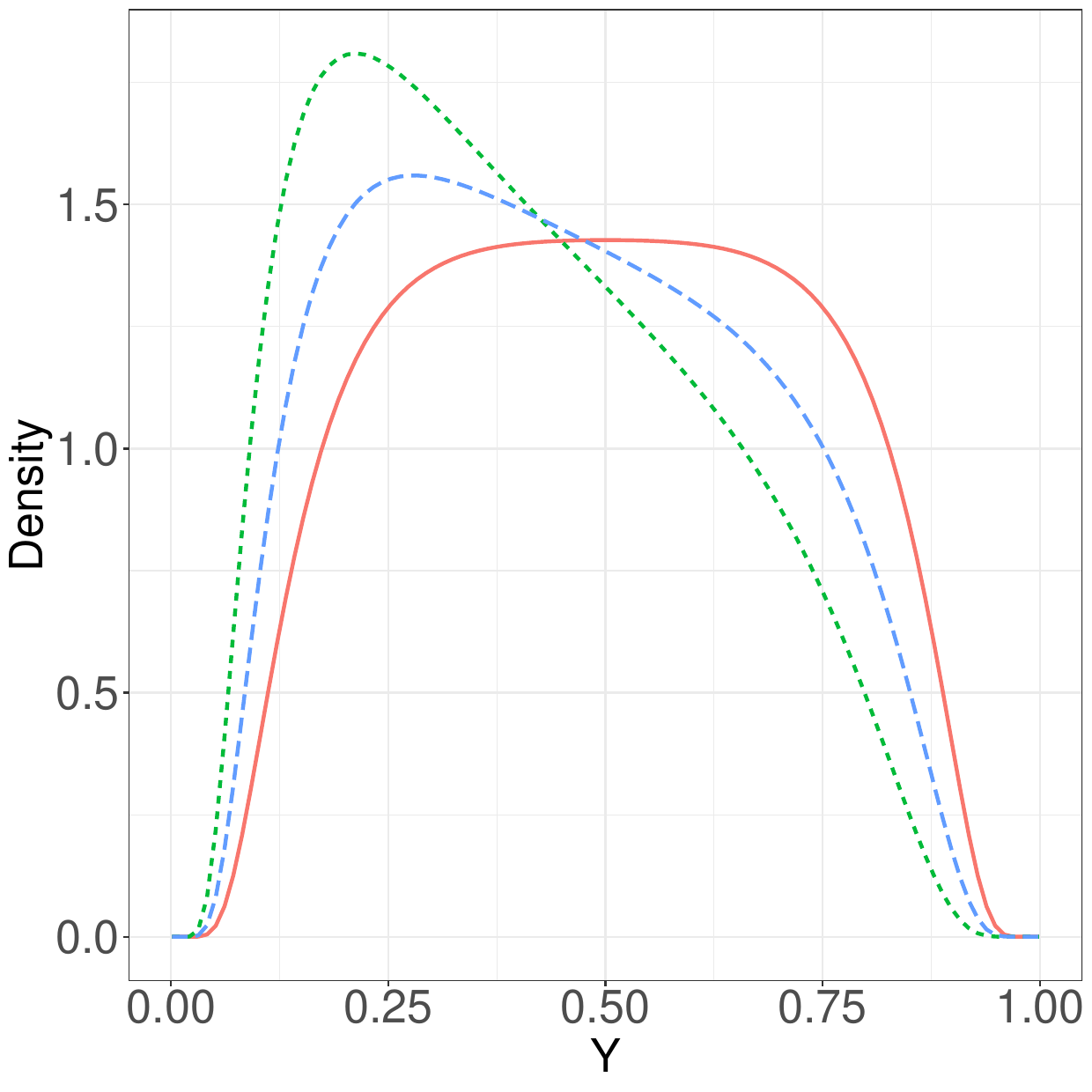}
    \end{minipage}
    \caption{Density graph of the Simplex Distribution for different parameter values.}
    \label{fig:dSimplex}
\end{figure}
If $y$ is a random variable that follows the Simplex distribution with mean $\mu$ and dispersion parameter $\sigma^2$, it frequently is denoted by $y \sim S(\mu,\sigma^2)$. It was shown that $(i)~{\rm E}[d'(y;\mu)]=0$, $(ii)~{\rm Var}[d(y;\mu)]=2(\sigma^2)^2$, $(iii)~{\rm E}[d(y;\mu)]=\sigma^2$,
$(iv)~{\rm E}[(y-\mu)d(y;\mu)]=0$, $(v)~{\rm E}[(y-\mu)d^2(y;\mu)]=0$, $(vi)~{\rm E}[(y-\mu)d''(y;\mu)]=-2\sigma^2$ and
$(vii)~\frac{1}{2}{\rm E}[(d''(y;\mu)]=\frac{3\sigma^2}{\mu(1-\mu)}+\frac{1}{\mu^3(1-\mu)^3}$,
where $d'(y;\mu)=\partial d(y;\mu)/\partial\mu$ e $d''(y;\mu)=\partial^2d(y;\mu)/\partial\mu^2$; for more details see \citet{song2000marginal}. 

Let $y_1, y_2, \ldots, y_n$ be a random sample, such that $y_i$ follows the Simplex distribution, given in (\ref{fdensid2}) for all $i=1,\ldots,n$. The likelihood function for the independent observations is defined as $\text{L}(\bm{\theta};\bm{y}) = \prod_{i=1}^{n}f(y_i;\bm{\theta}),$ where $\bm{\theta}=(\mu,\sigma^2)^{\top}$. The logarithm of the likelihood function is expressed in the form $\ell(\bm{\theta};\bm{y}) = \sum_{i=1}^n \ell_{i} (\bm{\theta};y_i)$, where
$\ell_{i}(\bm{\theta}, y_i)= -\log(2\pi)/2-\log(\sigma^2)/2-3\log[y_i(1-y_i)]/2-d(y_i;\mu)/2\sigma^2$. 
The maximum likelihood estimators for $\mu$ and $\sigma^2$ are found by simultaneously solving the estimation equations, i.e., $\partial \ell(\bm{\theta};\bm{y})/\partial \mu = 0$ e $\partial \ell(\bm{\theta};\bm{y})/\partial \sigma^2 = 0$, where 
\begin{align*}
\frac{\partial\ell(\bm{\theta}; \bm{y})}{\partial\mu} &= 
-\frac{1}{2\sigma^2} \sum^n_{i=1} 
\Bigg(
-\frac{2(y_i - \mu)}{\mu(1-\mu)}
\Big[
d(y_i, \mu) + \frac{1}{\mu^2(1-\mu)^2}
\Big]
\Bigg), \\
\frac{\partial\ell(\bm{\theta}; \bm{y})}{\partial\sigma^2} &= 
-\frac{n}{2\sigma^2} 
+ \sum^n_{i=1} \frac{d(y_i, \mu)}{2(\sigma^2)^2}.
\end{align*}   
It is easy to find that $\widehat{\sigma}^2=\sum_{i=1}^n d(y_i,\mu)/n$.
The second derivatives of $\ell(\bm{\theta};\bm y)$ related to the parameter vector  
are given by $\partial^2 \ell(\bm{\theta};\bm y)/\partial \mu^2=\sum_{i=1}^n -d''(y_i,\mu)/2\sigma^2$,
$\partial^2\ell(\bm{\theta};\bm y)/\partial(\sigma^2)^2=n/2\sigma^4-\sum_{i=1}^n d(y_i,\mu)/(\sigma^2)^{3}$ and
$\partial^2\ell(\bm{\theta};\bm y)/\partial\mu\partial\sigma^2=- \sum_{i=1}^n d'(y_i,\mu)/(\sigma^2)^2$ 
where $d^{\prime \prime}(y_i,\mu)= \partial^2 d(y_i,\mu)/\partial \mu^2$. 
Using properties $(i)$, $(iii)$, and $(vii)$ above, the Fisher’s information matrix is given by
\begin{eqnarray*}
  K(\bm{\theta})=K(\mu,\sigma^2)=  
  \left(
  \begin{tabular}{cc}
         $K_{\mu\mu}$ & $K_{\mu\sigma^2}$ \\
         $K_{\sigma^2 \mu}$ & $K_{\sigma^2\sigma^2}$\\
  \end{tabular}
  \right),
\end{eqnarray*}
where $K_{\mu \mu}=-{\rm E}[\partial^2 \ell(\bm{\theta};\bm y)/\partial \mu^2]=3n/\mu(1-\mu)+n/\sigma^2\mu^3(1-\mu)^3$, 
$K_{\mu\sigma^2}=K_{\sigma^2\mu}=-{\rm E}[\partial^2\ell(\bm{\theta};\bm y)/\partial\mu\partial\sigma^2]=0$ and
$K_{\sigma^{2}\sigma^{2}}=-{\rm E}[\partial^2\ell(\bm{\theta};\bm y)/\partial (\sigma^2)^2]=-n/2\sigma^4$
respectively. Under general regularity conditions, the maximum likelihood estimators are consistent, and the asymptotic distribution of $\sqrt{n}(\widehat{\bm{\theta}}-\bm{\theta})$  is normal with mean zero and covariance matrix $K^{-1}(\bm{\theta})$. 

\section{Bivariate Simplex distribution}
\label{sec:bSimplex}

In this section, we deduced the bivariate Simplex distribution via copulas because it is known that Bivariate distributions using copulas are powerful tools for modeling complex dependence structures between two variables. By separating the marginal distributions from the dependency, copulas offer more flexibility and accuracy, especially in fields where non-linear relationships and tail dependencies are essential. Expression analytics are deduced for the expected dependence between variables. 

The copula was introduced by \citet{sklar1959fonctions}, although similar ideas and results can be found in \cite{hoeffding1940masstabinvariante}. The copula function is one of several ways of generating multivariate distributions. A copula is defined as a joint distribution function of the form
\begin{eqnarray}
  C(u_1, u_2, \ldots, u_k)=P(U_1 \leq u_1, U_2 \leq u_2, \ldots, U_k \leq u_k ), 
\end{eqnarray}
where $0 \leq u_j \leq 1$, $U_j \sim U(0,1)$, for every $j = 1, 2,\ldots,k$. Suppose $H(\cdot)$ is a $k$-dimensional cumulative distribution function with marginals $F_1(\cdot),\ldots,F_k(\cdot)$. Then, according to \cite{sklar1959fonctions}, there is a $k$-dimensional copula $C$ such that, for every $(y_1,\ldots, y_k)\in (-\infty,\infty)^k$, $H(y_1,\ldots,y_k)=C(F_1(y_1),\ldots,F_k(y_k))$. In this sense, $C$ is unique if  $F_1,\ldots, F_k$ are all continuous. The above result guarantees that we can find the joint distribution of $k$ random variables $y_1,\ldots,y_k$, i.e., given a set of continuous random variables with marginal distribution functions $F_1(y_1),\ldots, F_k(y_k)$ and a distribution function $H(y_1,\ldots,y_k)$; the joint density function is given by
\begin{eqnarray*}
h(y_1, ..., y_k)&=&\frac{\partial^k H(y_1,...,y_k)}{\partial y_1...\partial y_k}
 = \frac{\partial^k C(F_1(y_1),...,F_k(y_k))}{\partial F_1(y_1)...\partial F_k(y_k)}\frac{\partial F_1(y_1)}{\partial y_1}\times ...\times \frac{\partial F_k(y_k)}{\partial y_k}
 \\ &=& c(F_1(y_1)...F_k(y_k))\prod^k_{i=1} f_i(y_i),
\end{eqnarray*}
where
$c(F_1(y_1),...,F_k(y_k))= \dfrac{\partial^k C(F_1(y_1),...,F_k(y_k))}{\partial F_1(y_1) \ldots \partial F_k(y_k)}$ and $f_i(y_i)=\dfrac{\partial F_i(y_i)}{\partial y_i}$. In particular, if $y_1$ and $y_2$ are two random variables with a joint distribution function $F(y_1,y_2)$ and continuous marginal distribution functions $F_1(y_1)$ and $F_2(y_2)$, respectively, then there exists a single copula $C: [0,1]^2 \rightarrow [0,1]$ such that for all $(y_1,y_2) \in R^2$ and $F(y_1,y_2)=C(F_1(y_1),F_2(y_2))$, expressing a bivariate density function with marginals $F_1$ and $F_2$. 

The FGM copula (Farlie-Gumbel-Morgenstern copula) is a relatively simple and classical copula used to model the dependence between two random variables. While it is not as flexible as other copulas, it provides a basic structure to capture weak positive or negative dependence. The FGM copula is defined as
\begin{eqnarray*}
    C(F_1(y_1),F_2(y_2); \lambda) = F_1(y_1)F_2(y_2) + \lambda F_1(y_1)F_2(y_2)(1 - F_1(y_1))(1 - F_2(y_2)),
\end{eqnarray*}
where $\lambda \in [-1, 1]$ is the dependence parameter. If $\lambda=0$, the copula reduces to the independence copula, meaning  $C(F_1(y_1), F_2(y_2)) = F_1(y_1) \times F_2(y_2)$, which indicates that the two random variables are independent; for $\lambda>0$ (Positive dependence), meaning the two variables tend to increase together (though weakly); for $\lambda<0$ (Negative dependence), meaning that when one variable increases, the other tends to decrease (weak negative dependence).

The Farlie-Gumbel-Morgenstern (FGM) copula is a valuable tool in the realm of dependence modeling due to its simplicity and interpretability. Its structure offers an intuitive way to capture weak dependence between variables, making it particularly useful in scenarios where more complex copulas may be unnecessary, or overfitting may be a concern. The FGM copula is symmetric, which allows it to treat both variables equally, and its dependence parameter $\lambda$ is easy to interpret, making it accessible for analysts and practitioners alike. While it does not capture extreme tail dependence, this feature can be advantageous in applications where only moderate associations are needed, such as in educational settings, initial exploratory analyses, or simple bivariate models. The FGM copula’s tractability and low computational demands make it a practical choice when the goal is to understand basic dependence structures without the overhead of more complex copulas, allowing it to strike a balance between simplicity and effectiveness. Thus, in this particular paper, we assume that the variables $y_1$ and $y_2$ follow univariate Simplex distributions, and the density function $f(y_1,y_2)$, is as follows:
\begin{eqnarray}
    \label{eq:bsimplex}
    f(y_1, y_2, \bm{\theta})=\{2\pi\sigma^2_{1}[y_{1}(1-y_{1})]^3\}^{-1/2}\exp\left\{-\frac{1}{2\sigma^2_{1}}d(y_{1};\mu_{1})\right\}\times\nonumber \\ 
    \{2\pi\sigma^2_{2}[y_{2}(1-y_{2})]^3\}^{-1/2}\exp\left\{-\frac{1}{2\sigma^2_{2}}d(y_{2};\mu_{2})\right\}\times\nonumber \\
    \{1+\lambda[1-2F_1(y_{1})][1-2F_2(y_{2})]\},
\end{eqnarray}
where $F_1(y_1)$ and $F_2(y_2)$ are the distribution functions of $y_1$ and $y_2$, respectively, and $\bm{\theta} = (\mu_1, \mu_2, \sigma^2_1, \sigma^2_2, \lambda)^{\top}$ is the vector of unknown parameters, were $\mu_1$, $\mu_2$  and $\sigma_1$, $\sigma_2$ represent the location and dispersion parameters, respectively and $\lambda$ is the dependency parameter. We denote $\bm{y} \sim S^{2}(\bm{\mu},\bm{\sigma}^2,\lambda)$ the random variable vector $\bm{y}=(y_1,y_2)^{\top}$  which following the bivariate Simplex distribution with  $\bm{\mu}=(\mu_1,\mu_2)^{\top}$  and $\bm{\sigma}=(\sigma_1,\sigma_2)^{\top}$ parameters.

\subsection{Moments}
\label{subsec:mom}

In this sub-section, we derive the joint expectation ${\rm E}(y_1y_2)$ in the theorem below. For better algebraic manipulation, we will consider the variables $y_1$ and $y_2$ as $x$ and $y$, respectively.  

\begin{theorem}
Let \( (x,y)^{\top} \) be a random vector following the bivariate Simplex distribution, where \( (\mu_x, \mu_y, \sigma^2_x, \sigma^2_y, \lambda)^{\top} \) is the vector of parameters. The joint moment of \( x \) and \( y \) is given by:
\begin{equation}
\label{EXY1}
{\rm E}[xy]=\mu_x\mu_y+\lambda
\Big[\frac{r_x^2\pi}{2}\Big(\frac{1}{a_x\xi_x}+\frac{1}{a_x}+A_x\Big)-\mu_x\Big]
\Big[\frac{r_y^2\pi}{2}\Big(\frac{1}{a_y\xi_y}+\frac{1}{a_y}+A_y\Big)-\mu_y\Big],
\end{equation}
where $A_m=1-2\Big(K_{0}(2a_m)L_{-1}(2a_m)+K_{1}(2a_m)L_{0}(2a_m)\Big)$, $a_m=(\xi_m+1)^2/\sigma_m^2\xi_m$, $\xi_m=(1/\mu_m)-1$ and $r_m=1/(\sigma_m\sqrt{2\pi})$ for $m=\{x,y\}$,  $K_{\nu}(\cdot)$ represents the modified Bessel function of the second kind and $L_{\nu}(\cdot)$  denotes the modified Struve function.
\end{theorem}
\begin{proof}

Since ${\rm E}[xy]=\int_0^1 \int_0^1 xyf(x,y)dxdy$, where $f(x,y)$ is given in Equation \eqref{eq:bsimplex}, it is clear that to calculate ${\rm E}[xy]$, it is sufficient to evaluate the following integral:
\begin{eqnarray*}
  I=\int_0^1 xg(x;\mu,\sigma^2)G(x;\mu,\sigma^2)dx,  
\end{eqnarray*}
where $g(\cdot; \mu,\sigma^2)$ denotes the density function of the univariate Simplex distribution with parameters $\mu$ and $\sigma^2$, and $G(\cdot,\mu,\sigma^{2})$ represents its cumulative distribution function. Then
\begin{eqnarray*}
I=\int_0^1 x\{2\pi \sigma^2 [x(1-x)]^3\}^{-1/2}e^{-\frac{1}{2\sigma^2}d(x,\mu)}
\Big(\int_0^x \{2\pi \sigma^2 [t(1-t)]^3\}^{-1/2}e^{-\frac{1}{2\sigma^2}d(t,\mu)}dt\Big)dx.   
\end{eqnarray*}
Letting $r=1/\sigma\sqrt{2\pi}$, the previous expression can be simplified by
\begin{eqnarray*}
I=r^2\int_0^1 \frac{1}{x^{1/2}(1-x)^{3/2}}e^{-\frac{1}{2\sigma^2}d(x,\mu)}
\int_0^x \frac{1}{t^{3/2}(1-t)^{3/2}}e^{-\frac{1}{2\sigma^2}d(t,\mu)}dt dx.   
\end{eqnarray*}
Applying the change of variable $t=1/(z+1)$ we find
\begin{eqnarray*}
I=r^2\int_0^1 \frac{1}{x^{1/2}(1-x)^{3/2}}e^{-\frac{1}{2\sigma^2}d(x,\mu)}
\int_{\frac{1}{x}-1}^\infty \frac{z+1}{z^{3/2}}e^{-\frac{1}{2\sigma^2}\frac{(\xi-z)^2(\xi+1)^2}{z\xi^2}}dzdx.    
\end{eqnarray*}
Setting \( \mu = 1/(\xi+1) \) and performing the change of variable \( x = 1/(\omega+1) \) in the outer integral results in:
\begin{eqnarray*}
I=r^2\int_0^\infty \frac{1}{\omega^{3/2}}e^{-\frac{1}{2\sigma^2}\frac{(\xi-\omega)^2(\xi+1)^2}{\omega\xi^2}}
\int_{\omega}^\infty \frac{z+1}{z^{3/2}}e^{-\frac{1}{2\sigma^2}\frac{(\xi-z)^2(\xi+1)^2}{z\xi^2}}dz d\omega.    
\end{eqnarray*}
We observe that setting \( \beta = (\xi+1)^2 / \sigma^2 \) and \( \gamma = (\xi+1)^2 / (\sigma^2 \xi^2) \), the exponential term in \( z \) becomes \( \exp\{-(\beta z^{-1} + \gamma z)/2\}\exp\{a\} \), where \( a = (\xi+1)^2 / (\sigma^2 \xi) \). Similarly, for the exponential term in \( \omega \). We then obtain the following simplified expression:
\begin{eqnarray*}
I=r^2 e^{2a} \int_0^\infty \frac{1}{\omega^{3/2}} e^{-\frac{1}{2}(\beta \omega^{-1} + \gamma \omega)}
\int_{\omega}^\infty \frac{z+1}{z^{3/2}} e^{-\frac{1}{2}(\beta z^{-1} + \gamma z)} dz d\omega.
\end{eqnarray*}
Once again, we use a change of variable in the inner integral. By letting $p=z/\omega$, then changing the order of the integrals and simplifying the terms, we get:
\begin{eqnarray}
\label{expressao}
I=r^2e^{2a}\int_{1}^\infty \int_0^\infty 
\frac{1}{p^{3/2}}\left(\frac{1}{\omega^2}+\frac{p}{\omega}\right)e^{-\frac{1}{2}((\beta+\beta/p) \omega^{-1}+(\gamma+\gamma p) \omega)}d\omega dp.
\end{eqnarray}
We intend to apply Proposition \ref{teobessel} to equation \ref{expressao}. To do that, we observe that the above integral can be separated into a sum of two integrals $I=I_1+I_2$, where
\begin{eqnarray*}
I_1=r^2e^{2a}\int_{1}^\infty \frac{1}{p^{3/2}}\int_0^\infty 
\omega^{\nu_1 -1}e^{-\frac{1}{2}((\beta+\beta/p) \omega^{-1}+(\gamma+\gamma p) \omega)}d\omega dp, \text{ where } \nu_1=-1,   
\end{eqnarray*}
and
\begin{eqnarray*}
I_2=r^2e^{2a}\int_{1}^\infty \frac{1}{p^{1/2}}\int_0^\infty 
\omega^{\nu_2 -1}e^{-\frac{1}{2}((\beta+\beta/p) \omega^{-1}+(\gamma+\gamma p) \omega)}d\omega dp,\text{ where } \nu_2=0.     
\end{eqnarray*}
Now, applying Proposition \ref{teobessel} in each of the integrals, the expression \eqref{expressao} reduces to the following sum of integrals of Bessel functions:
\begin{equation*}
I=2r^2e^{2a} \Bigg\{ \frac{1}{\xi} \int_{1}^\infty \frac{1}{p} K_1\Big(a \frac{p+1}{p^{1/2}}\Big) dp +\int_{1}^\infty \frac{1}{p^{3/2}} K_0\Big(a \frac{p+1}{p^{1/2}}\Big) \Bigg\}.
\end{equation*}
Applying a change of variables, let $q=\sqrt{p}$. We then obtain:
\begin{equation}
\label{expressao3}
I=r^2e^{2a} \{ \frac{1}{\xi} \int_{1}^\infty \frac{1}{q} K_1\Big(a \frac{q^2+1}{q}\Big)dq +\int_{1}^\infty K_0\Big(a \frac{q^2+1}{q}\Big) dq\}.
\end{equation}
For simplicity, we denote $J_0$ and $J_1$ by the following integrals
\begin{eqnarray*}
J_0=\int_{1}^\infty K_0\Big(a \frac{q^2+1}{q}\Big)dq, \quad {\rm and} \quad J_1=\int_{1}^\infty \frac{1}{q} K_1\Big(a \frac{q^2+1}{q}\Big)dq. 
\end{eqnarray*}
We recall that the expression \eqref{expressao3} represents our integral \( \int_0^1 x g(x;\mu,\sigma^2) G(x;\mu,\sigma^2) \, dx \) from the beginning of this proof. Therefore, substituting \( J_0 \) and \( J_1 \) into the expression \eqref{expressao3}, we obtain:
\begin{equation}
\label{integralgG}
I=\int_0^1 xg(x;\mu,\sigma^2)G(x;\mu,\sigma^2)dx=r^2e^{2a}\{J_1/\xi+J_0\}.
\end{equation}
Note that the integrals \( J_0 \) and \( J_1 \) were calculated in Lemma \ref{lema0} and Corollary \ref{lema1} respectively:

\begin{equation}
\label{J1}
    J_1=\frac{\pi}{4a}e^{-2a},
\end{equation}
and

\begin{equation}\label{J0}
    J_0=\frac{\pi}{4a}e^{-2a}+\Big[\frac{\pi}{4}-\frac{\pi}{2}(K_{0}(2a)L_{-1}(2a)+K_{1}(2a)L_{0}(2a))\Big].
\end{equation}

To compute ${\rm E}[xy]$, we use identity  \eqref{integralgG} in the following expression
\begin{eqnarray*}
{\rm E}[xy]=\mu_x\mu_y +&\lambda&\Bigg[2\int_0^1 xg(x;\mu_x,\sigma_{x}^{2})G(x;\mu_x,\sigma_{x}^{2})dx-\mu_x\Bigg] \times \\ &&\Bigg[2\int_0^1 yh(y;\mu_y,\sigma_{y}^{2})H(y;\mu_y,\sigma_{y}^{2})dy-\mu_y\Bigg],    
\end{eqnarray*}
to obtain
\begin{equation}
\label{EXY}
{\rm E}[xy]=\mu_x\mu_y+\lambda\Big[2r_x^{2}e^{2a_x}(J_1^{x}/\xi_{x}+J_0^x)-\mu_1\Big]\Big[2r_y^{2}e^{2a_y}(J_1^{y}/\xi_{y}+J_0^{y})-\mu_2\Big],
\end{equation}
where $r_m=1/\sigma_m\sqrt{2\pi}$, $a_m=(\xi_x+1)^2/\sigma_m^2\xi_m$ , $\xi_m=(1/\mu_m)-1$, $J_0^m=\int_{1}^\infty K_0(a_m \frac{m^2+1}{m})dm$, $J_1^m=\int_{1}^\infty \frac{1}{m} K_1(a_m \frac{m^2+1}{m})dm$, $m\in \{x,y\}.$

Finally, substituting the identities \eqref{J1} and \eqref{J0} in \eqref{EXY} yields:

\begin{equation}
\label{EXY1}
{\rm E}[xy]=\mu_x\mu_y+\lambda
\Big[r_x^2\frac{\pi}{2}\Big(\frac{1}{a_x\xi_x}+\frac{1}{a_x}+A_x\Big)-\mu_x\Big]
\Big[r_y^2\frac{\pi}{2}\Big(\frac{1}{a_y\xi_y}+\frac{1}{a_y}+A_y\Big)-\mu_y\Big],
\end{equation}
where $A_x=1-2(K_{0}(2a_x)L_{-1}(2a_x)+K_{1}(2a_x)L_{0}(2a_x))$ and $A_y=1-2(K_{0}(2a_y)L_{-1}(2a_y)+K_{1}(2a_y)L_{0}(2a_y))$.
\end{proof}

\subsection{Maximum likelihood estimation}
\label{subsec:Esti}

Let $\bm{y}=(\bm{y}_1, \bm{y}_2, \ldots, \bm{y}_n)^{\top}$ represent a vector, where $\bm{y}_i=(y_{i1}, y_{i2})^{\top}$ follow the bivariate Simplex distribution given in (\ref{eq:bsimplex}) , for $i=1,\ldots, n$. The likelihood function of  all pairs of observations is defined as $\text{L}(\bm{\theta}, \bm{y})=\prod_{i=1}^{n}f(\bm{y}_{i};\bm{\theta} ),$  where $\bm{\theta} = (\mu_1, \mu_2, \sigma^2_1, \sigma^2_2, \lambda)^{\top}$ and the logarithm of the likelihood function is given by $\ell(\bm{\theta};\bm{y}) = \sum_{i=1}^n \ell_{i} (\bm{\theta};\bm{y}_i)$, where

\begin{eqnarray}
\label{logVero}
\ell_i(\bm{\theta};\bm{y}_i)= -\log(2\pi)-\frac{1}{2}\log(\sigma^2_{1})-\frac{3}{2}\log[y_{i1}(1-y_{i1})]-\frac{1}{2\sigma^2_{1}}d(y_{i1};\mu_{1})\nonumber \\ -\frac{1}{2}\log(\sigma^2_{2})-\frac{3}{2}\log[y_{i2}(1-y_{i2})]-\frac{1}{2\sigma^2_{2}}d(y_{i2};\mu_{2})\nonumber \\ + \log\{1+\lambda[1-2F_1(y_{i1})][1-2F_2(y_{i2})]\}.
\end{eqnarray}

By partially deriving the logarithm of the likelihood function with respect to the vector of parameters, the elements of the score vector $\bm{U}(\bm{\theta})=(U_{\mu_1}, U_{\mu_2}, U_{\sigma^{2}_{1}}, U_{\sigma_{2}^{2}}, U_{\lambda})^{\top}$, are obtained from the expressions:
\begin{eqnarray*}
U_{\mu_1}=\frac{\partial \ell(\bm{\theta};\bm y)}{\partial \mu_1} &=& \sum^n_{i=1}\frac{d(y_{i1}; \mu_1)}{\sigma^2_1 \mu_1 (1-\mu_1)} \left[ (y_{i1}-\mu_1) + \frac{y_{i1}(1y_{i1})}{y_{i1}-\mu_1} \right] + \\
& & \frac{2\lambda\{2F_2(y_{i2})-1\}\dot{F_1}_{\mu_1}}{1+\lambda\{ 2F_1(y_{i1})-1 \} \{2F_2(y_{i2})-1\}}, \\ \\
U_{\mu_2}=\frac{\partial \ell(\bm{\theta};\bm y)}{\partial \mu_2} &=& \sum^n_{i=1}\frac{d(y_{i2}; \mu_2)}{\sigma^2_2 \mu_2 (1-\mu_2)} \left[ (y_{i2}-\mu_2) + \frac{y_{i2}(1-y_{i2})}{y_{i2}-\mu_2} \right] + \\ 
& & \frac{2\lambda\{2F_1(y_{i1})-1\}\dot{F_2}_{\mu_2}}{1+\lambda\{ 2F_2(y_{i2})-1 \} \
2F_1(y_{i1})-1\}}, \\ \\
U_{\sigma^{2}_{1}}=\frac{\partial \ell(\bm{\theta};\bm y)}{\partial \sigma^2_1} &=& \sum^n_{i=1}\frac{1}{2\sigma^2_1} \left[ \frac{d(y_{i1};\mu_1)}{\sigma^2_1}-1 \right] + \frac{2\lambda\{2F_2(y_{i2})-1\}\dot{F_1}_{\sigma^2_1}}{1+\lambda\{2F_1(y_{i1})-1\}\{2F_2(y_{i2})-1\}}, \\ \\
U_{\sigma_{2}^{2}}=\frac{\partial \ell(\bm{\theta};\bm y)}{\partial \sigma^2_2} &=& \sum^n_{i=1}\frac{1}{2\sigma^2_2} \left[ \frac{d(y_{i2};\mu_2)}{\sigma^2_2}-1 \right] + \frac{2\lambda\{2F_1(y_{i1})-1\}\dot{F_2}_{\sigma^2_2}}{1+\lambda\{2F_2(y_{i2})-1\}\{2F_1(y_{i1})-1\}}, \\ \\
U_{\lambda}=\frac{\partial \ell(\bm{\theta};\bm y)}{\partial \lambda} &=& \sum^n_{i=1}\frac{\{2F_1(y_{i1})-1\}\{2F_2(y_{i2})-1\}}{1+\lambda\{2F_1(y_{i1})-1\}\{2F_2(y_{i2})-1\}},
\end{eqnarray*}
where $\dot{F_1}_{\mu_1} = \partial F_1(y_1)/\partial \mu_1$, $\dot{F_2}_{\mu_2} = \partial F_2(y_2)/\partial \mu_2$, $\dot{F_1}_{\sigma^2_1}  = \partial F_1(y_1)/\partial \sigma^2_1$, e $\dot{F_2}_{\sigma^2_2} = \partial F_2(y_2)/\partial \sigma^2_2$. Similarly to the univariate case, it is possible to find the observed information matrix, $J(\bm{\theta})=\partial\ell(\bm{\theta};\bm y)/\partial\bm{\theta}\partial\bm{\theta}^{\top}$, 
whose partial derivatives are given by

\begin{align*}
J_{\mu_1\mu_1}(\bm{\theta})=& \frac{\partial^2 \ell(\bm{\theta})}{\partial \mu_1^2} =
\sum^n_{i=1}-\frac{1}{\sigma^2_1 \mu_1^2 (1-\mu_1)^2} \Big[ \frac{3y_{i1}(1-y_{i1})}{\mu_1^2 (1-\mu_1)^2} - \frac{2}{\mu_1(1-\mu_1)} + 3d(y_{i1}; \mu_1) \big[ (y_{i1}-\mu_1)^2+ \\
& 2y_{i1}(1-y_{i1}) \big] \Big] \sum^n_{i=1}+\frac{2\lambda\{2F_2(y_{i2})-1\}}{G_{y_{i1}y_{i2}}^2} \left[ \ddot{F_1}_{\mu_1}G_{y_{i1}y_{i2}} - 2\lambda\{F_2(y_{i2})-1\} \dot{F_1}_{\mu_1} \right],
\end{align*}
\begin{align*}
J_{\mu_1\sigma_1^{2}}(\bm{\theta})=&\frac{\partial^2 \ell(\bm{\theta})}{\partial \mu_1 \partial \sigma_1^2}=
\sum^n_{i=1} -\frac{d(y_{i1};\mu_1)}{\sigma_1^4 \mu_1(1-\mu_1)} \Big[ y_{i1}-\mu_1+\frac{y_{i1}(1-y_{i1})}{y_{i1}-\mu_1} \Big] + \\
& \frac{2\lambda\{2F_2(y_{i2})-1\}}{G_{y_{i1}y_{i2}}^2} \Big[ \ddot{F_1}_{\mu_1\sigma_1^2}G_{y_{i1}y_{i2}}-2\lambda\dot{F_1}_{\mu_1}\dot{F_1}_{\sigma_1^2} \{2F_2(y_{i2})-1\} \Big],
\end{align*}
\begin{align*}
J_{\mu_1\mu_2}(\bm{\theta})= \frac{\partial^2 \ell(\bm{\theta})}{\partial \mu_1 \partial \mu_2} =&
\sum^n_{i=1} \frac{4\lambda\dot{F_1}_{\mu_1}\dot{F_2}_{\mu_2}}{G_{y_{i1}y_{i2}}^2},
\end{align*}
\begin{align*}
J_{\mu_1\sigma_2^{2}}(\bm{\theta})= \frac{\partial^2 \ell(\bm{\theta})}{\partial \mu_1 \partial \sigma_2^2} =&
\sum^n_{i=1} \frac{4\lambda\dot{F_1}_{\mu_1}\dot{F_2}_{\sigma_2^2}}{G_{y_{i1}y_{i2}}^2},
\end{align*}
\begin{align*}
J_{\mu_1\lambda}(\bm{\theta})= \frac{\partial^2 \ell(\bm{\theta})}{\partial \mu_1 \partial \lambda} =&
\sum^n_{i=1} \frac{2 \{ 2F_2(y_{i2})-1 \} \dot{F_1}_{\mu_1} }{G_{y_{i1}y_{i2}}^2},
\end{align*}
\begin{align*}
J_{\mu_2\mu_2}(\bm{\theta})=& \frac{\partial^2 \ell(\bm{\theta})}{\partial \mu_2^2} =
\sum^n_{i=1}-\frac{1}{\sigma^2_2 \mu_2^2 (1-\mu_2)^2} \Big[ \frac{3y_{i2}(1-y_{i2})}{\mu_2^2 (1-\mu_2)^2} - \frac{2}{\mu_2(1-\mu_2)} + 3d(y_{i2}; \mu_2) \big[ (y_{i2}-\mu_1)^2+\\
& 2y_{i2}(1-y_{i2}) \big] \Big]\sum^n_{i=1}+\frac{2\lambda\{2F_1(y_{i1})-1\}}{G_{y_{i1}y_{i2}}^2} \left[ \ddot{F_2}_{\mu_2}G_{y_{i1}y_{i2}} - 2\lambda\{F_1(y_{i1})-1\} \dot{F_2}_{\mu_2} \right],
\end{align*}
\begin{align*}
J_{\mu_2\sigma_1^{2}}(\bm{\theta})= \frac{\partial^2 \ell(\bm{\theta})}{\partial \mu_2 \partial \sigma_1^2} =&
\sum^n_{i=1} \frac{4\lambda\dot{F_2}_{\mu_2}\dot{F_1}_{\sigma_1^2}}{G_{y_{i1}y_{i2}}^2},
\end{align*}
\begin{align*}
J_{\mu_2\sigma_2^{2}}(\bm{\theta})= & \frac{\partial^2 \ell(\bm{\theta})}{\partial \mu_2 \partial \sigma_2^2} =
\sum^n_{i=1} -\frac{d(y_{i2};\mu_2)}{\sigma_2^4 \mu_2(1-\mu_2)} \Big[ y_{i2}-\mu_2+\frac{y_{i2}(1-y_{i2})}{y_{i2}-\mu_2} \Big] + \\
& \frac{2\lambda\{2F_1(y_{i1})-1\}}{G_{y_{i1}y_{i2}}^2} \Big[ \ddot{F_2}_{\mu_2\sigma_2^2}G_{y_{i1}y_{i2}}-2\lambda\dot{F_2}_{\mu_2}\dot{F_2}_{\sigma_2^2} \{2F_1(y_{i1})-1\} \Big],
\end{align*}
\begin{align*}
J_{\mu_2\lambda}(\bm{\theta})= \frac{\partial^2 \ell(\bm{\theta})}{\partial \mu_2 \partial \lambda} =&
\sum^n_{i=1} \frac{2 \{ 2F_1(y_{i1})-1 \} \dot{F_2}_{\mu_2} }{G_{y_{i1}y_{i2}}^2},
\end{align*}
\begin{align*}
J_{\sigma_1^{2}\sigma_1^{2}}(\bm{\theta})= \frac{\partial^2 \ell(\bm{\theta})}{\partial (\sigma^2_1)^2} =&
\sum^n_{i=1}\frac{1}{2\sigma_1^4} - \frac{d(y_{i1};\mu_1)}{\sigma_1^6} + \frac{2\lambda\{2F_2(y_{i2})-1\}}{G_{y_{i1}y_{i2}}} \left[ \ddot{F_1}_{\sigma_1^2} - \frac{2\lambda\{2F_2(y_{i2})-1\} \dot{F_1}_{\sigma_1^2}}{G^2_{y_{i1}y_{i2}}} \right],
\end{align*}
\begin{align*}
J_{\sigma_1^{2}\sigma_2^{2}}(\bm{\theta})= \frac{\partial^2 \ell(\bm{\theta})}{\partial \sigma_1^2 \partial \sigma_2^2} =&
\sum^n_{i=1} \frac{4\lambda\dot{F_1}_{\sigma_1^2}\dot{F_2}_{\sigma_2^2}}{G_{y_{i1}y_{i2}}^2},
\end{align*}
\begin{align*}
J_{\sigma_1^{2}\lambda}(\bm{\theta})= \frac{\partial^2 \ell(\bm{\theta})}{\partial \sigma_1^2 \partial \lambda} =&
\sum^n_{i=1} \frac{2 \{ 2F_2(y_{i2})-1 \} \dot{F_1}_{\sigma_1^2} }{G_{y_{i1}y_{i2}}^2},
\end{align*}
\begin{align*}
J_{\sigma_2^{2}\sigma_2^{2}}(\bm{\theta})= \frac{\partial^2 \ell(\bm{\theta})}{\partial (\sigma^2_2)^2} =&
\sum^n_{i=1}\frac{1}{2\sigma_2^4} - \frac{d(y_{i2};\mu_2)}{\sigma_2^6} + \frac{2\lambda\{2F_1(y_{i1})-1\}}{G_{y_{i1}y_{i2}}} \left[ \ddot{F_2}_{\sigma_2^2} - \frac{2\lambda\{2F_1(y_{i1})-1\} \dot{F_2}_{\sigma_2^2}}{G^2_{y_{i1}y_{i2}}} \right],
\end{align*}
\begin{align*}
J_{\sigma_2^{2}\lambda}(\bm{\theta})= \frac{\partial^2 \ell(\bm{\theta})}{\partial \sigma_2^2 \partial \lambda} =&
\sum^n_{i=1} \frac{2 \{ 2F_2(y_{i1})-1 \} \dot{F_2}_{\sigma_2^2} }{G_{y_{i1}y_{i2}}^2},
\end{align*}
\begin{align*}
J_{\lambda\lambda}(\bm{\theta})=\frac{\partial^2 \ell(\bm{\theta})}{\partial \lambda^2} =&
\sum^n_{i=1}-\left[ \frac{\{2F_1(y_{i1})-1\}\{2F_2(y_{i2})-1\}}{G_{y_{i1}y_{i2}}} \right]^2,
\end{align*}
where
$\ddot{F_1}_{\mu_1} = \partial \dot{F_1}_{\mu_1} / \partial \mu_1$, 
$\ddot{F_2}_{\mu_2} = \partial \dot{F_2}_{\mu_2} / \partial \mu_2$, 
$\ddot{F_1}_{\sigma^2_1} = \partial \dot{F_1}_{\sigma^2_1} / \partial \sigma^2_1$,
$\ddot{F_2}_{\sigma^2_2} = \partial \dot{F_2}_{\sigma^2_2} / \partial \sigma^2_2$,
$\ddot{F_1}_{\mu_1\sigma_1^2} = \partial F_1(y_{i1}) / \partial \mu_1 \partial \sigma_1^2$,
$\ddot{F_2}_{\mu_2\sigma_2^2} = \partial F_2(y_{i2}) / \partial \mu_2 \partial \sigma_2^2$ e
$G_{y_{i1}y_{i2}} = 1+\lambda\{ 2F_1(y_{i1})-1 \} \{2F_2(y_{i2})-1\}$. Under certain regularity conditions, the maximum likelihood estimator $\widehat{\bm{\theta}}$ of $\bm{\theta}$ approximates a Normal distribution with zero mean and variance and covariance matrix $J^{-1}(\bm{\theta})$; allowing confidence intervals to be found, hypotheses to be tested and predictions to be made.
\section{Simulation study}
\label{subsec:simu}
In this section, we perform a Monte Carlo simulation study to assess the asymptotic behavior of the maximum likelihood estimators for the bivariate Simplex distribution. The numerical results are derived on $R=1,000$ Monte Carlo replications, with sample sizes of $n=50$, $100$, $150$, $200$, and $1,000$ observations. The random response vector $\bm{y}=(\bm{y}_1,\ldots,\bm{y}_n)^{\top},$ where $\bm{y}_i=(y_{1i},y_{2i})^{\top}$ 
is generated by using the algorithm described in \cite{Johnson1987n_aleatorios}. The algorithm involves the following steps: $(i)$ Generate two independent random variables $u_1$ and $v$ with uniform distributions, $U(0, 1)$; $(ii)$ Compute: $A = \lambda(2u_1-1)-1$, $B = [1-\lambda(2u_1-1)]^2 + 4v\lambda(2u_1-1)$ and $u_2 = 2v/(\sqrt{B}-A)$; $(iii)$ Apply the inverse transformation method to obtain $y_1 = F_1^{-1}(u_1)$ and $y_2 = F_2^{-1}(u_2)$, where $F_1(\cdot)$ and $F_2(\cdot)$ are the cumulative distribution function of $y_1$ and $y_2$, respectively.
The mean, bias, root mean square error (RMSE), and the 95\% confidence interval coverage probability are computed based on the following expressions:
\begin{equation}
    \overline{\theta}_j = \frac{1}{R} \sum_{i=1}^R \widehat{\theta}_j^{(i)}, \hspace{1cm}
    \text{Bias}(\theta_j) = \overline{\theta}_j - \theta_j \hspace{1cm} \text{and} \hspace{1cm}
    \text{RMSE}(\theta_j) = \sqrt{\frac{1}{R} \sum_{i=1}^N (\widehat{\theta}_j^{(i)} - \theta_j)^2}, \nonumber
\end{equation}
\\
  where  $\bm{\theta} = (\theta_1, \theta_2, \theta_3, \theta_4, \theta_5)^{\top}$ = $(\mu_1, \mu_2, \sigma^2_1, \sigma^2_2, \lambda)^{\top}$. We considered three scenarios, in which the parameter $\lambda$ takes on values of $-1$, $0$ and $1$ 
  to perform the behavior of $\widehat{\bm{\theta}}$.

\subsection{Scenario 1}
\label{subsec:cena1}

In this scenario, the following vectors are taken as the true values of the parameters $\bm{\theta}_1 = (0.5, 0.5, 2, 2, 1)^{\top}$, $\bm{\theta}_2 = (0.5, 0.5, 5, 5, 1)^{\top}$ and $\bm{\theta}_3 = (0.9, 0.9, \sqrt{11}, \sqrt{11}, 1)^{\top}$. Figure~\ref{fig:superficie-contono1} (see Appendix) illustrates the surface and contour plots of the generated samples. The joint moments ${\rm E}(y_1y_2)$ for $\bm{\theta_1}$, $\bm{\theta_2}$ and $\bm{\theta_3}$ are 0.36, 0.40 and 1.24, respectively. For $\bm{\theta}_1$, the generated samples are concentrated in the interval $(0.25;0.75)$, indicating a bimodal behavior. Similarly, for $\bm{\theta}_2$ and $\bm{\theta}_3$, the samples are concentrated near zero and one simultaneously and near one, respectively.
Figure \ref{fig:fig1} displays the simulation results for this scenario. The parameter vectors $\bm{\theta}_1$, $\bm{\theta}_2$, and $\bm{\theta}_3$ are represented by the colors red, green, and blue, respectively. Solid and dashed lines in the figure correspond to the parameters associated with $y_1$ and $y_2$, respectively.
\begin{figure}[h!]
    \begin{minipage}[t]{\linewidth}
        \centering
        \begin{minipage}[t]{0.32\linewidth}
            \centering
            \includegraphics[width=5cm]{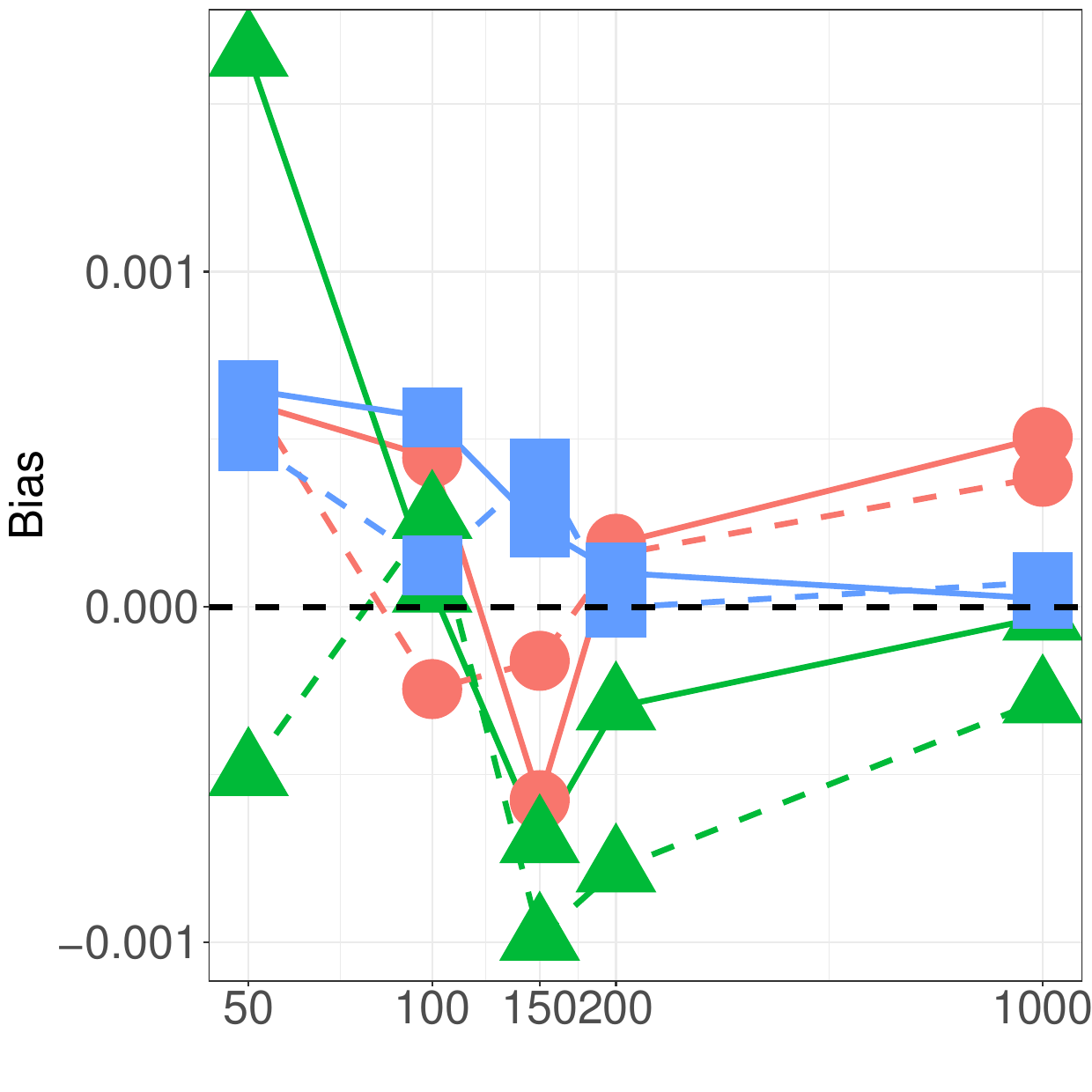}
        \end{minipage}%
        \begin{minipage}[t]{0.32\linewidth}
            \centering
            \includegraphics[width=5cm]{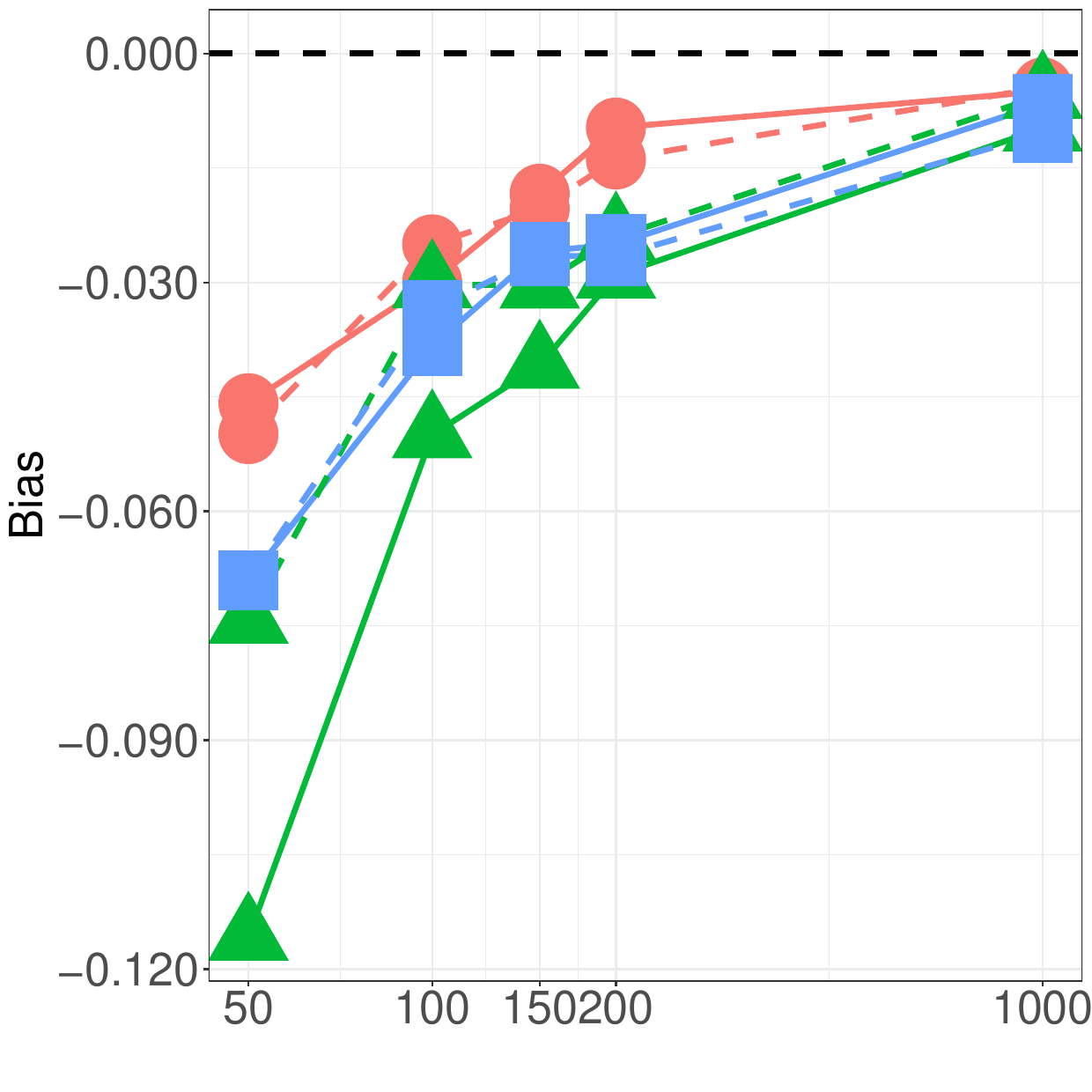}
        \end{minipage}%
        \begin{minipage}[t]{0.32\linewidth}
            \centering
            \includegraphics[width=5cm]{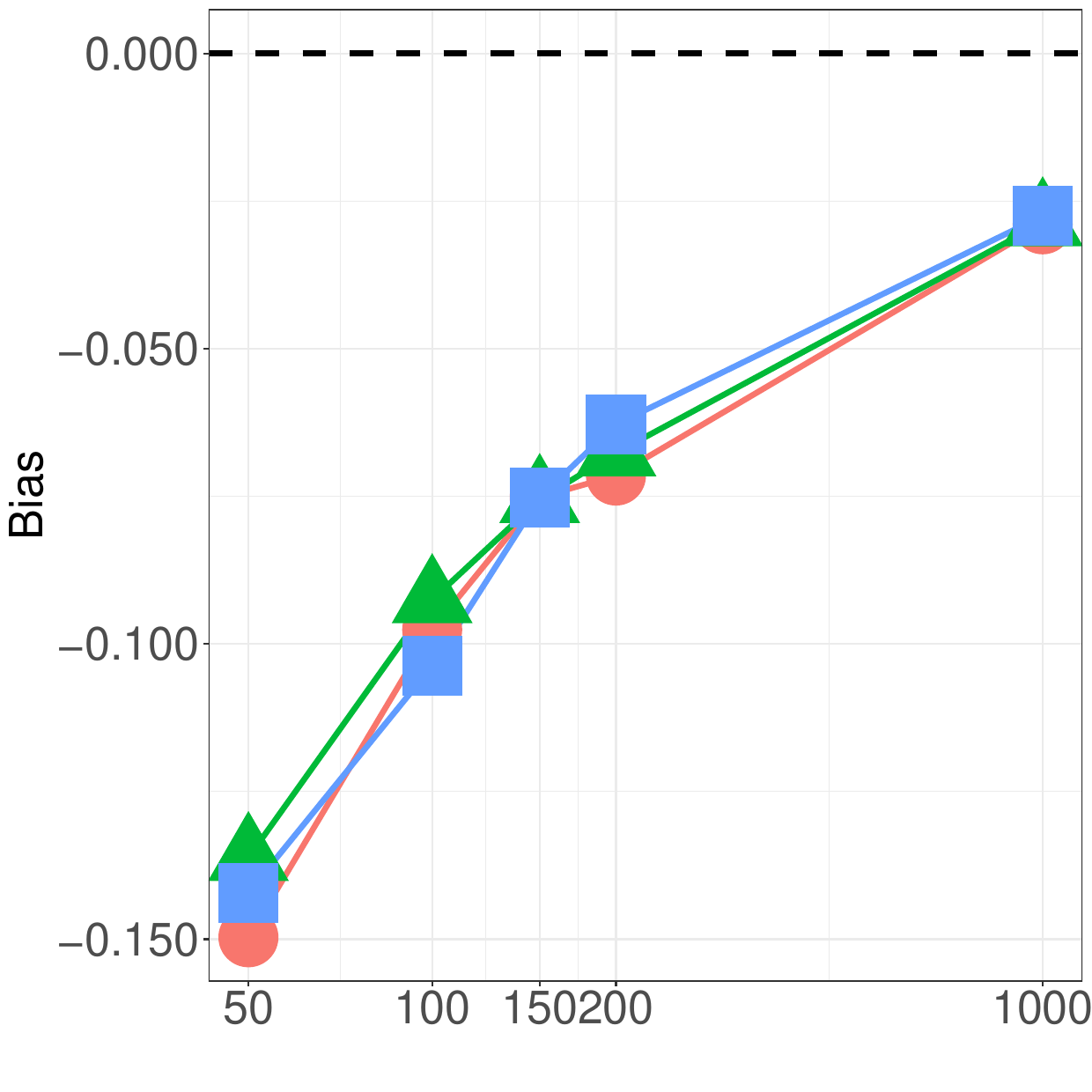}
        \end{minipage}
    \end{minipage}
    
    \begin{minipage}[t]{\linewidth}
        \centering
        \begin{minipage}[t]{0.32\linewidth}
            \centering
            \includegraphics[width=5cm]{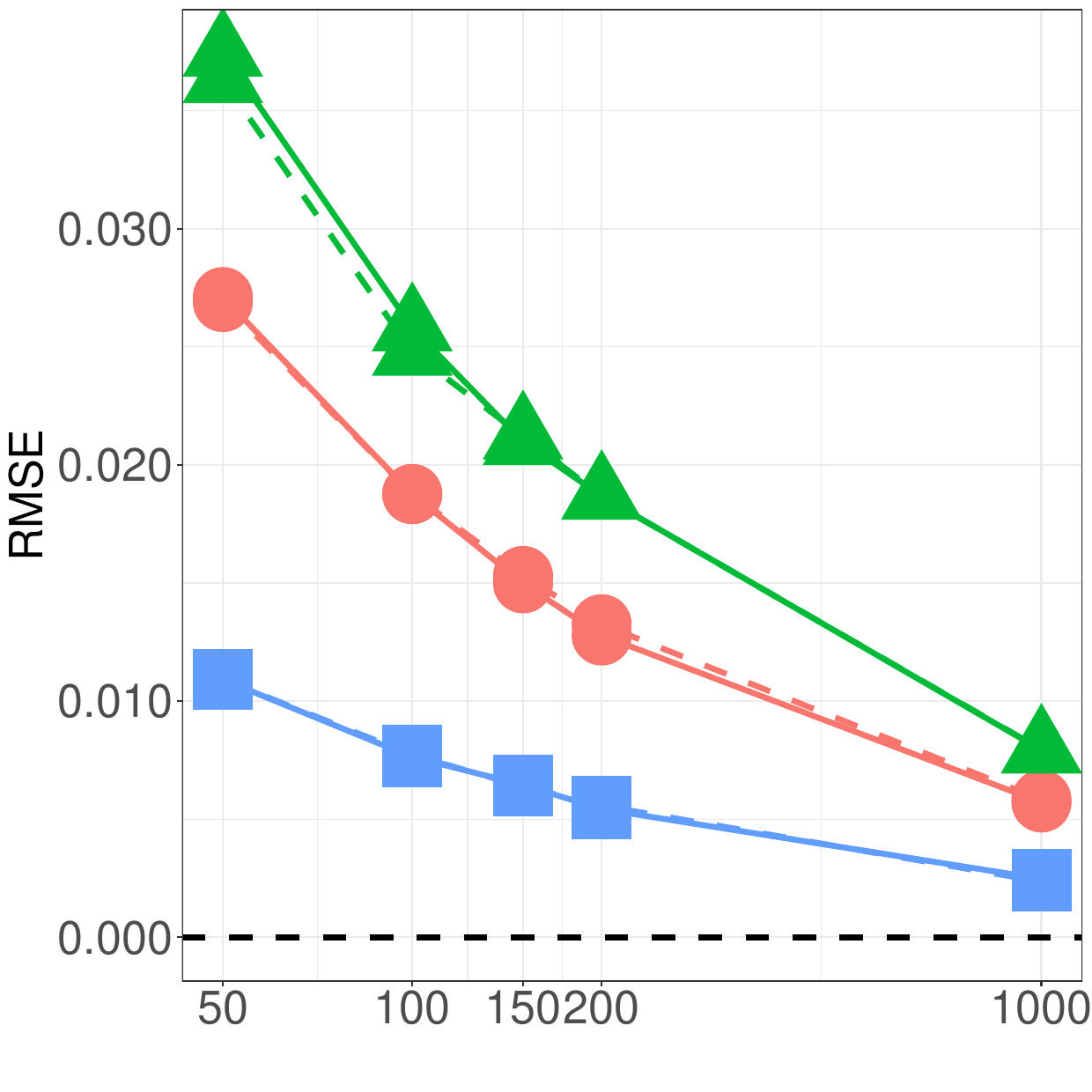}
        \end{minipage}%
        \begin{minipage}[t]{0.32\linewidth}
            \centering
            \includegraphics[width=5cm]{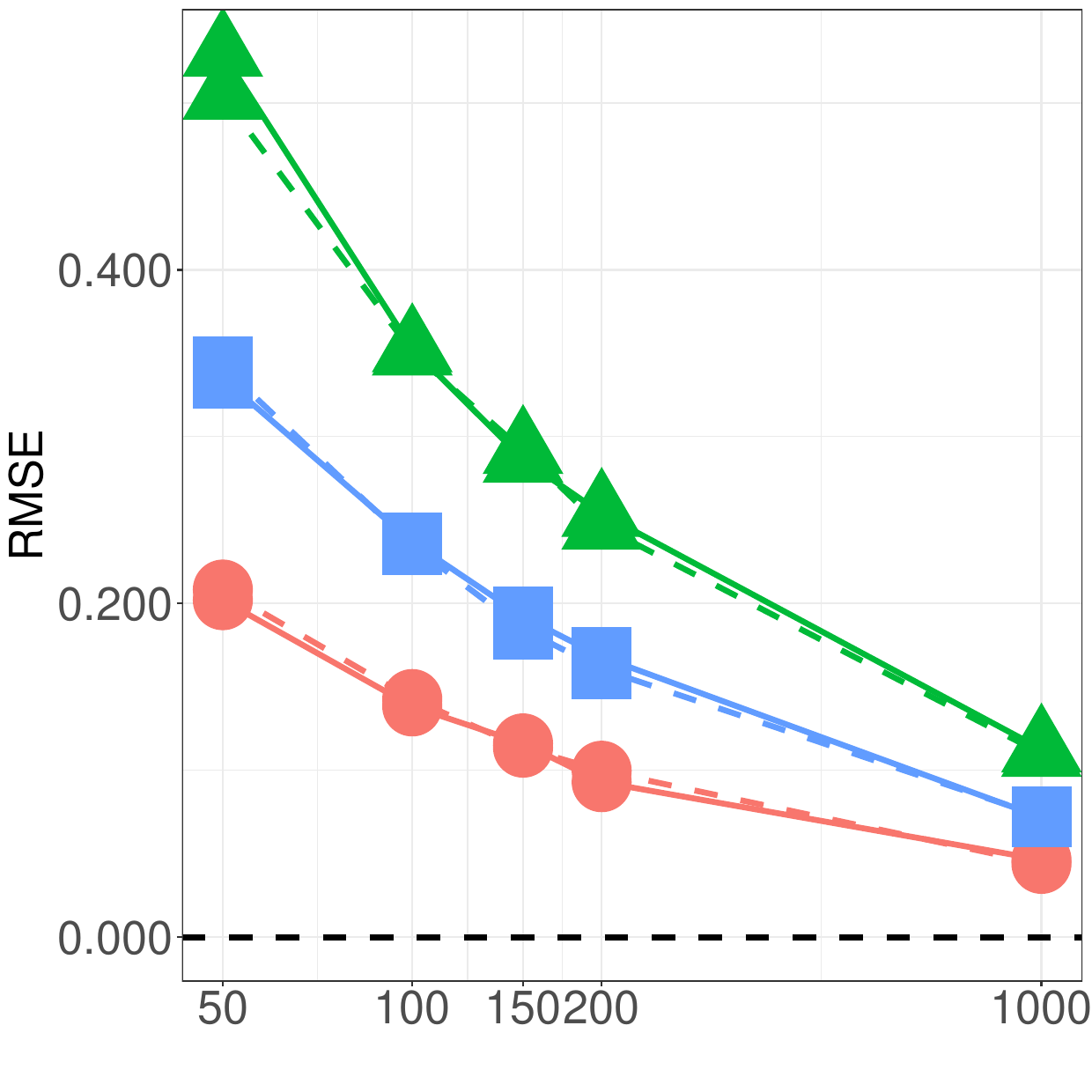}
        \end{minipage}%
        \begin{minipage}[t]{0.32\linewidth}
            \centering
            \includegraphics[width=5cm]{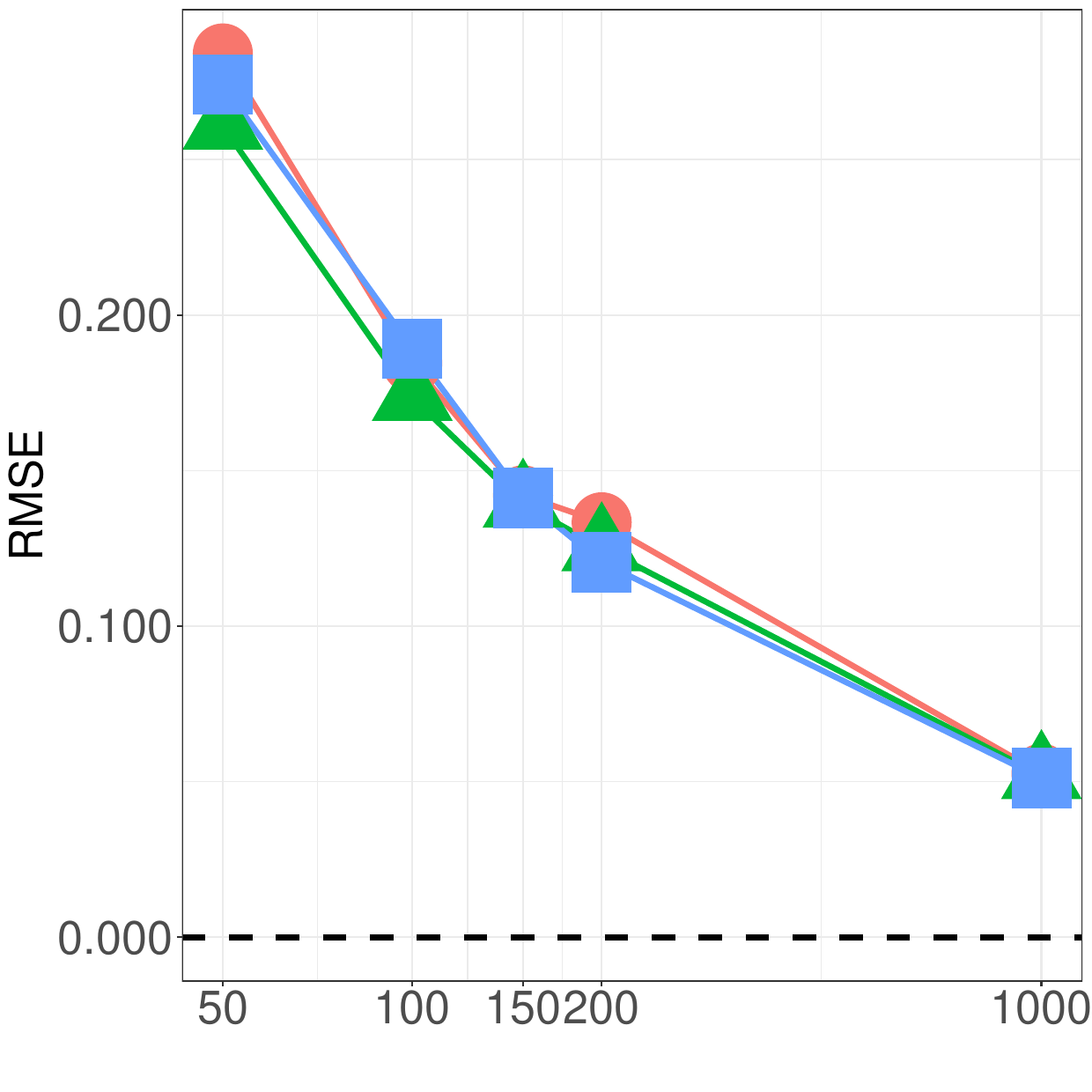}
        \end{minipage}
    \end{minipage}
    
    \begin{minipage}[t]{\linewidth}
        \centering
        \begin{minipage}[t]{0.32\linewidth}
            \centering
            \includegraphics[width=5cm]{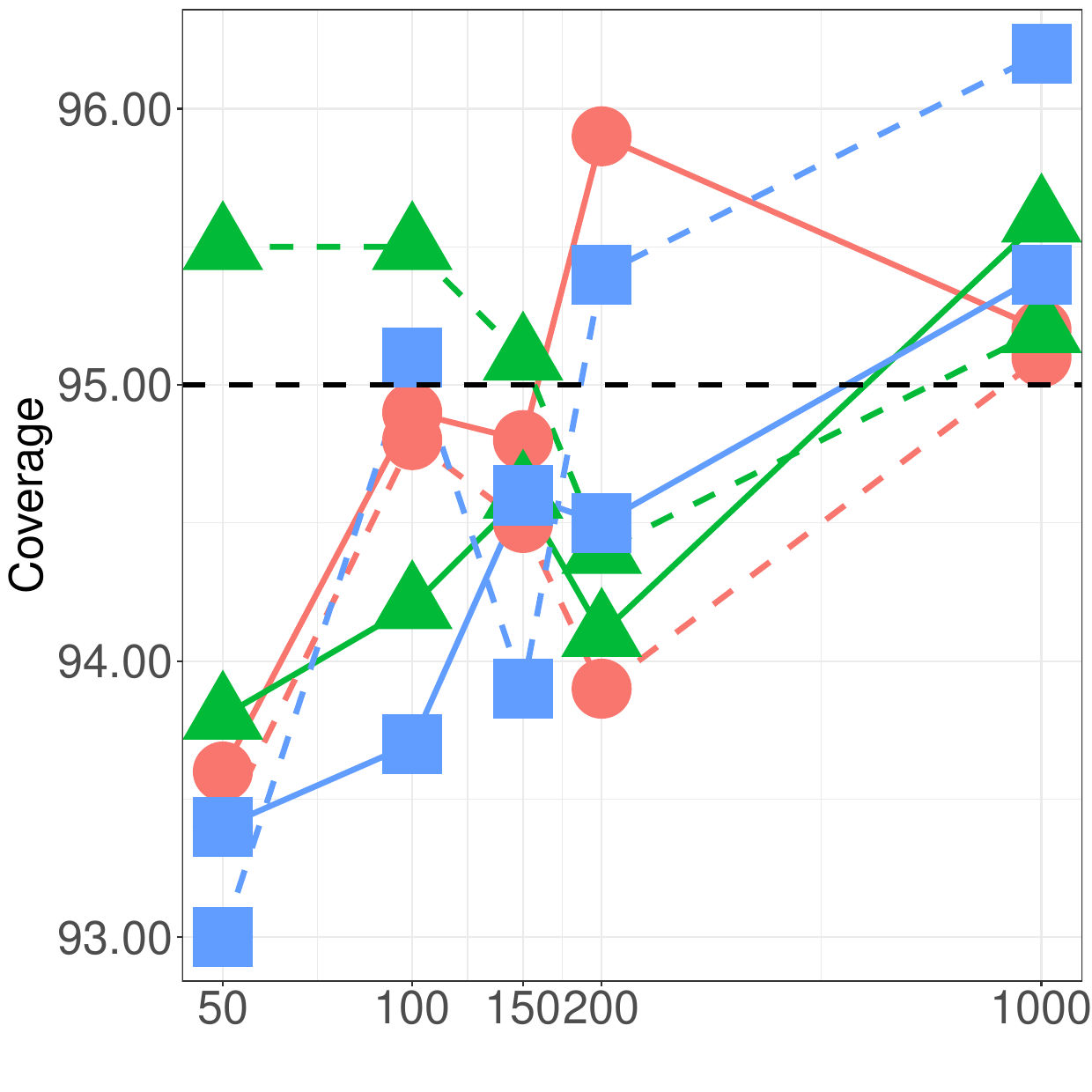}
        \end{minipage}%
        \begin{minipage}[t]{0.32\linewidth}
            \centering
            \includegraphics[width=5cm]{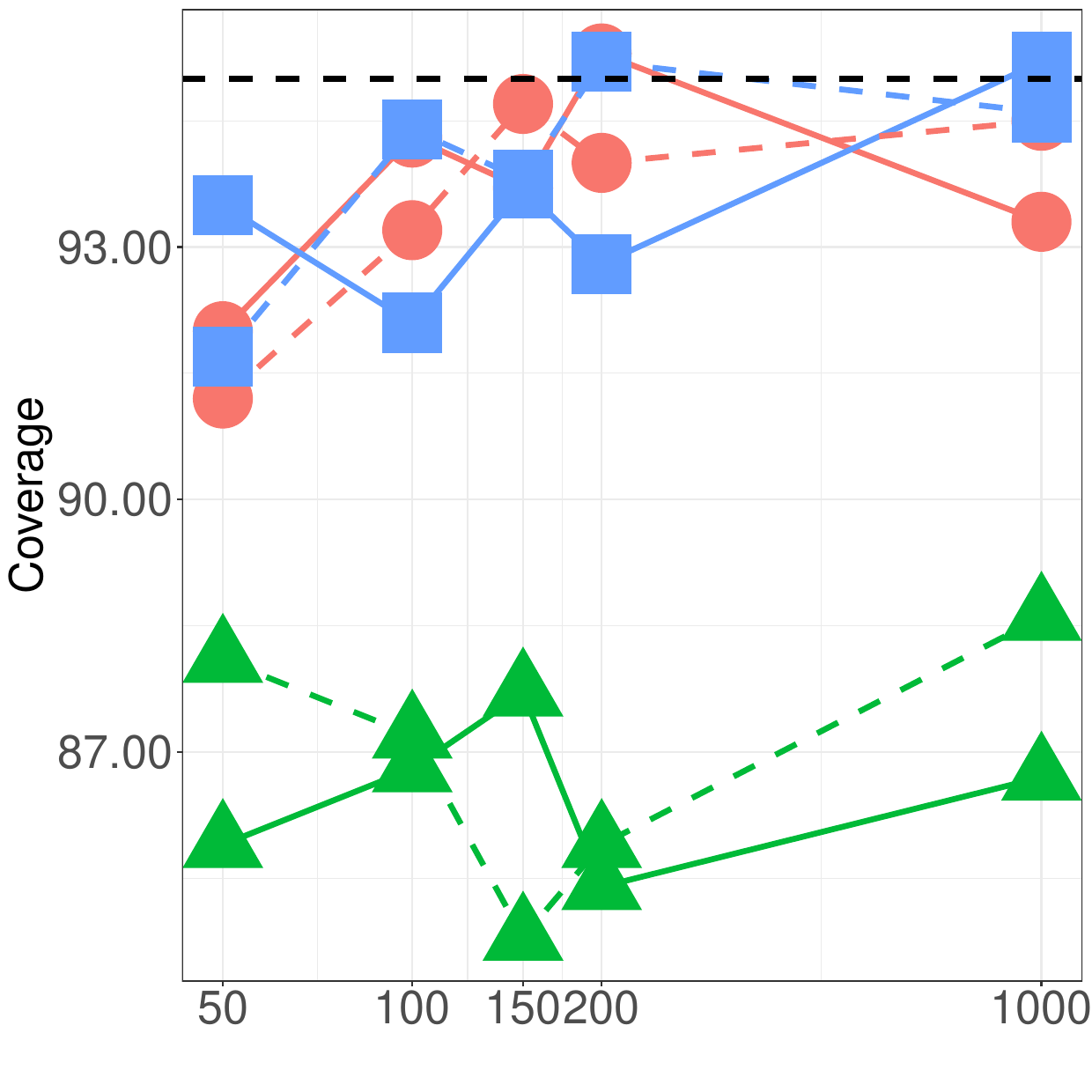}
        \end{minipage}%
        \begin{minipage}[t]{0.32\linewidth}
            \centering
            \includegraphics[width=5cm]{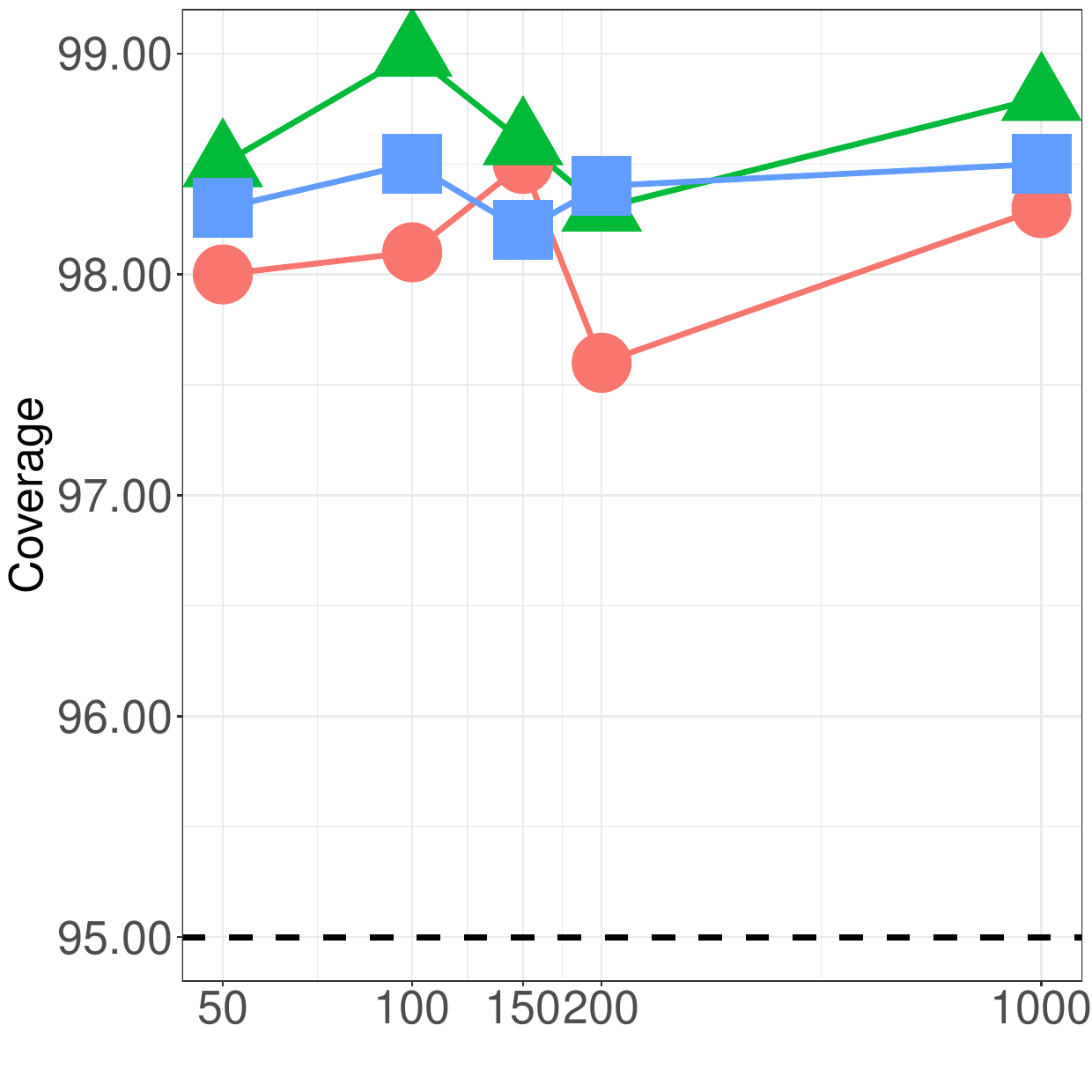}
        \end{minipage}
    \end{minipage}

    \caption{Bias (row 1), RMSE (row 2) and Coverage (row 3) of the parameters $\mu_1$ and $\mu_2$ (column 1), $\sigma^2_1$ and $\sigma^2_2$ (column 2) and $\lambda$ (column 3) for Scenario 1.}
    \label{fig:fig1}
\end{figure}
As expected, the bias and root mean square error (RMSE) approach zero as the sample size increases, indicating that the maximum likelihood estimators are asymptotically unbiased. The probability of coverage for the parameters $\mu_1$ and $\mu_2$ are close to the nominal 95\% level across different parameter vectors $\bm{\theta}$ and sample sizes. However, for parameters $\sigma_1^2$ and $\sigma_2^2$, the coverage probability is underestimated when samples are generated using the $\bm{\theta}_2$ parameter. 
The coverage of the $\lambda$ parameters is overestimated for different sample sizes. Tables 
\ref{tab:scenario-results-1.1} - \ref{tab:scenario-results-1.3} (in the Appendix) show the results of this scenario.


\subsection{Scenario 2}
\label{subsec:cena2}

In this scenario, the following vectors are taken as the true values of the parameters $\bm{\theta}_1 = (0.5, 0.5, 2, 2, -1)^{\top}$, $\bm{\theta}_2 = (0.5, 0.5, 5, 5, -1)^{\top}$ and $\bm{\theta}_3 = (0.9, 0.9, \sqrt{11}, \sqrt{11}, -1)^{\top}$. Figure~\ref{fig:superficie-e-contono-2} (see Appendix) shows the surface and contour plots of the generated samples.  The joint moments ${\rm E}(y_1y_2)$ for $\bm{\theta_1}$, $\bm{\theta_2}$ and $\bm{\theta_3}$ in this scenario are 0.14, 0.10 and 0.38, respectively. For  $\bm{\theta}_1$, the samples are concentrated within the interval $(0.25;0.75)$, indicating a bimodal behavior that is the inverse of the behavior observed in the first scenario. 
Similarly, for $\bm{\theta}_2$ and $\bm{\theta}_3$, the samples are concentrated near zero and one. 
Figure \ref{fig:fig2} presents the simulation results for this scenario. Again, the colors red, green, and blue for the parameter vectors $\bm{\theta}_1$, $\bm{\theta}_2$, and $\bm{\theta}_3$, respectively. The solid and dashed lines correspond to the parameters associated with the variables $y_1$ and $y_2$, respectively.
\begin{figure}[h!]
    \begin{minipage}[t]{\linewidth}
        \centering
        \begin{minipage}[t]{0.32\linewidth}
            \centering
            \includegraphics[width=5cm]{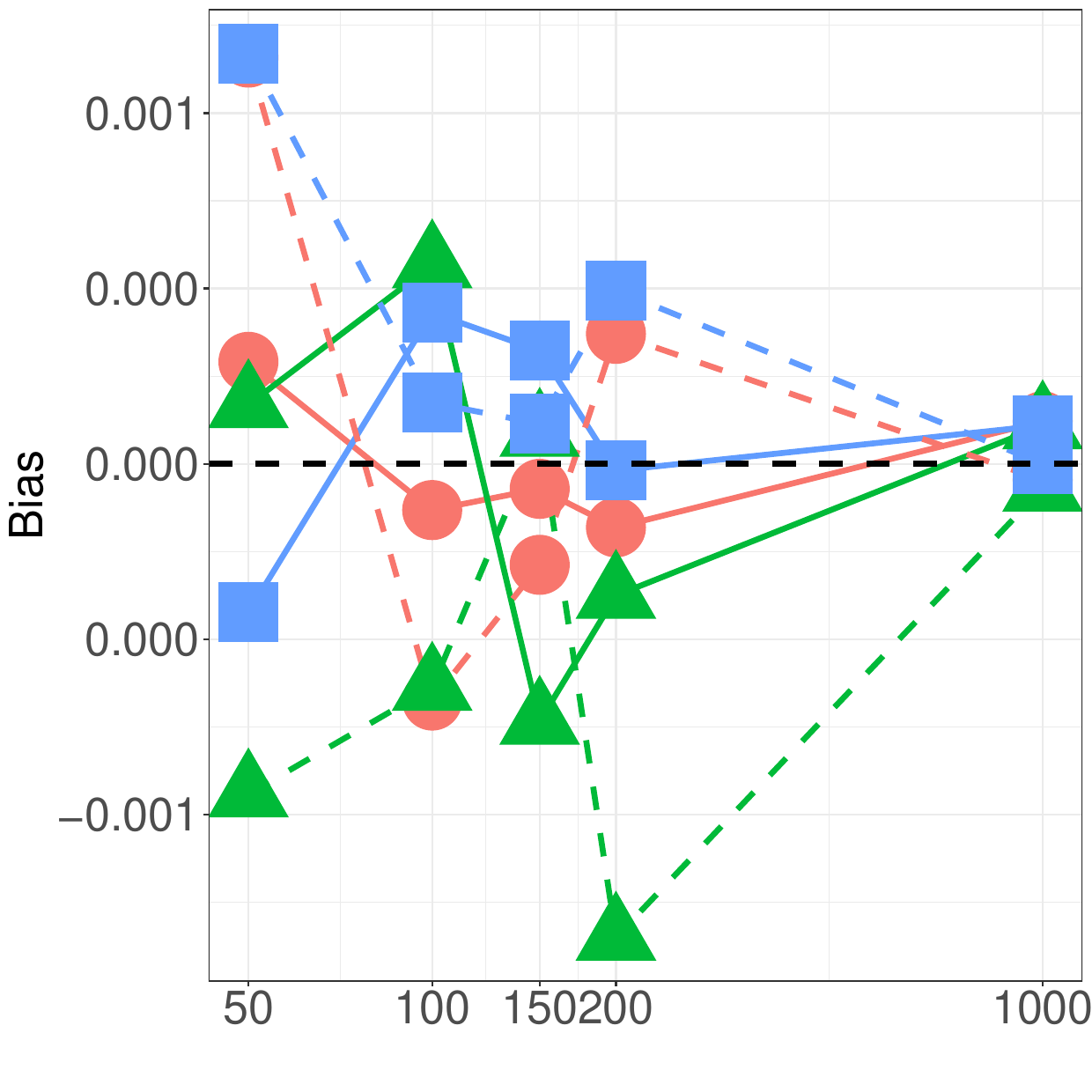}
        \end{minipage}%
        \begin{minipage}[t]{0.32\linewidth}
            \centering
            \includegraphics[width=5cm]{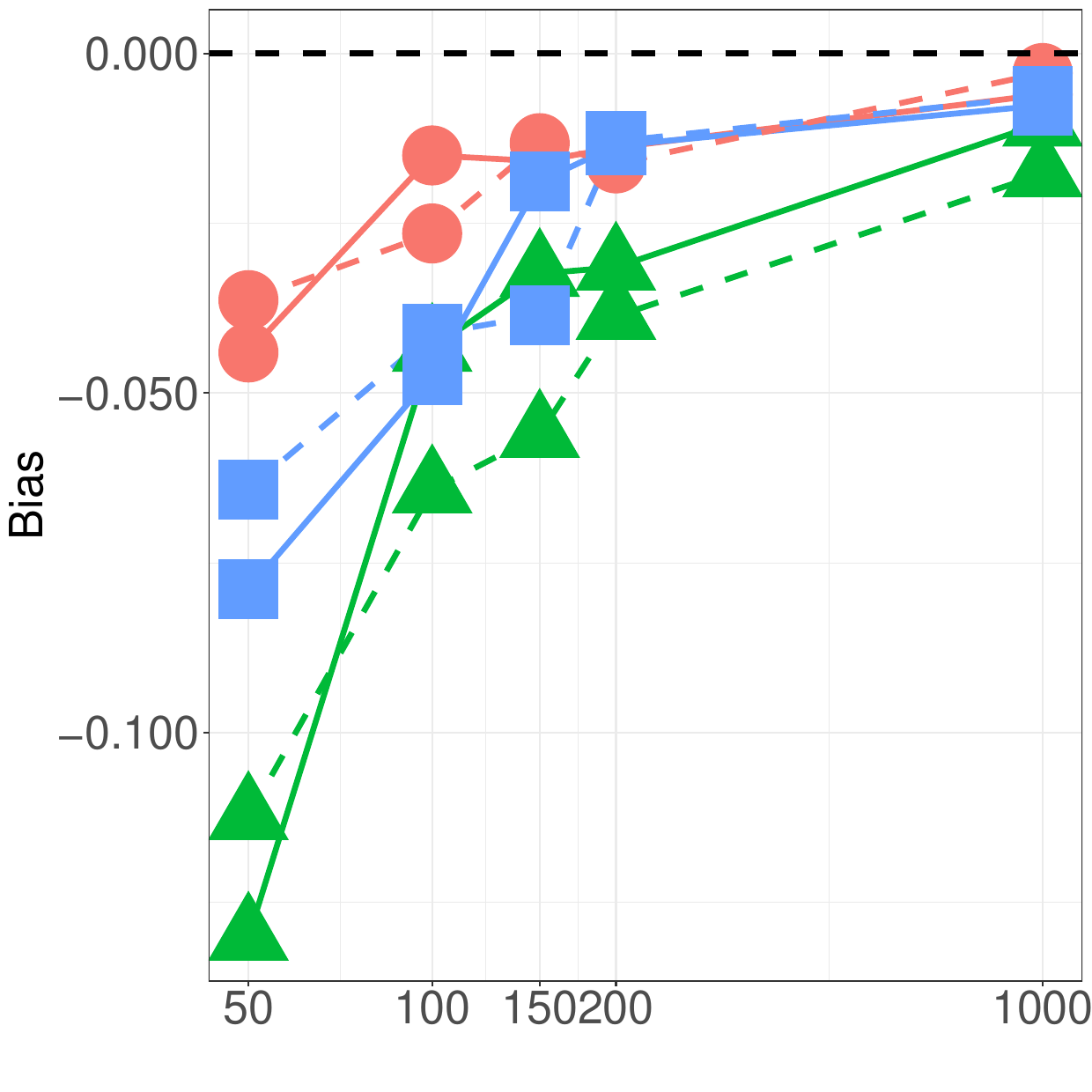}
        \end{minipage}%
        \begin{minipage}[t]{0.32\linewidth}
            \centering
            \includegraphics[width=5cm]{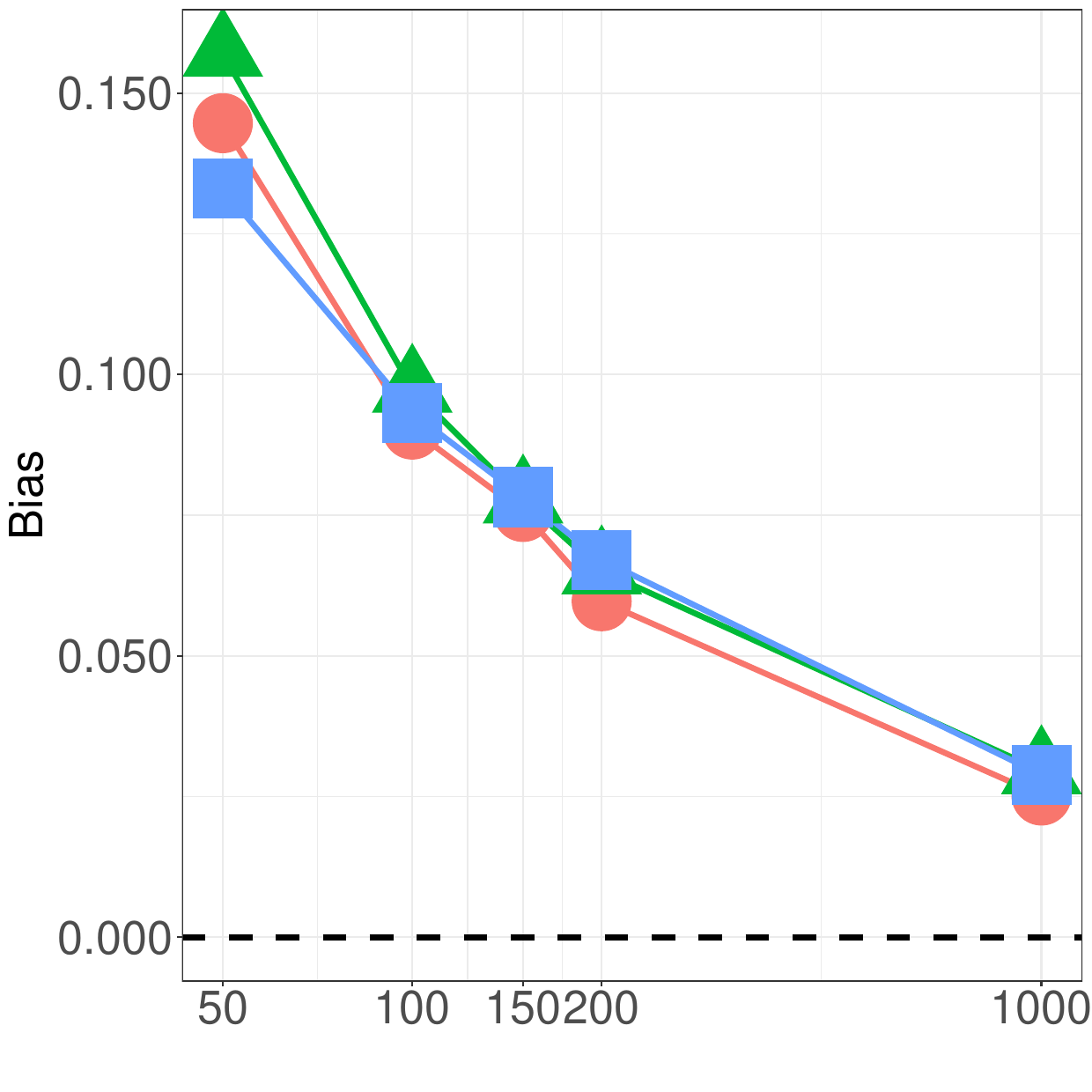}
        \end{minipage}
    \end{minipage}
    
    \begin{minipage}[t]{\linewidth}
        \centering
        \begin{minipage}[t]{0.32\linewidth}
            \centering
            \includegraphics[width=5cm]{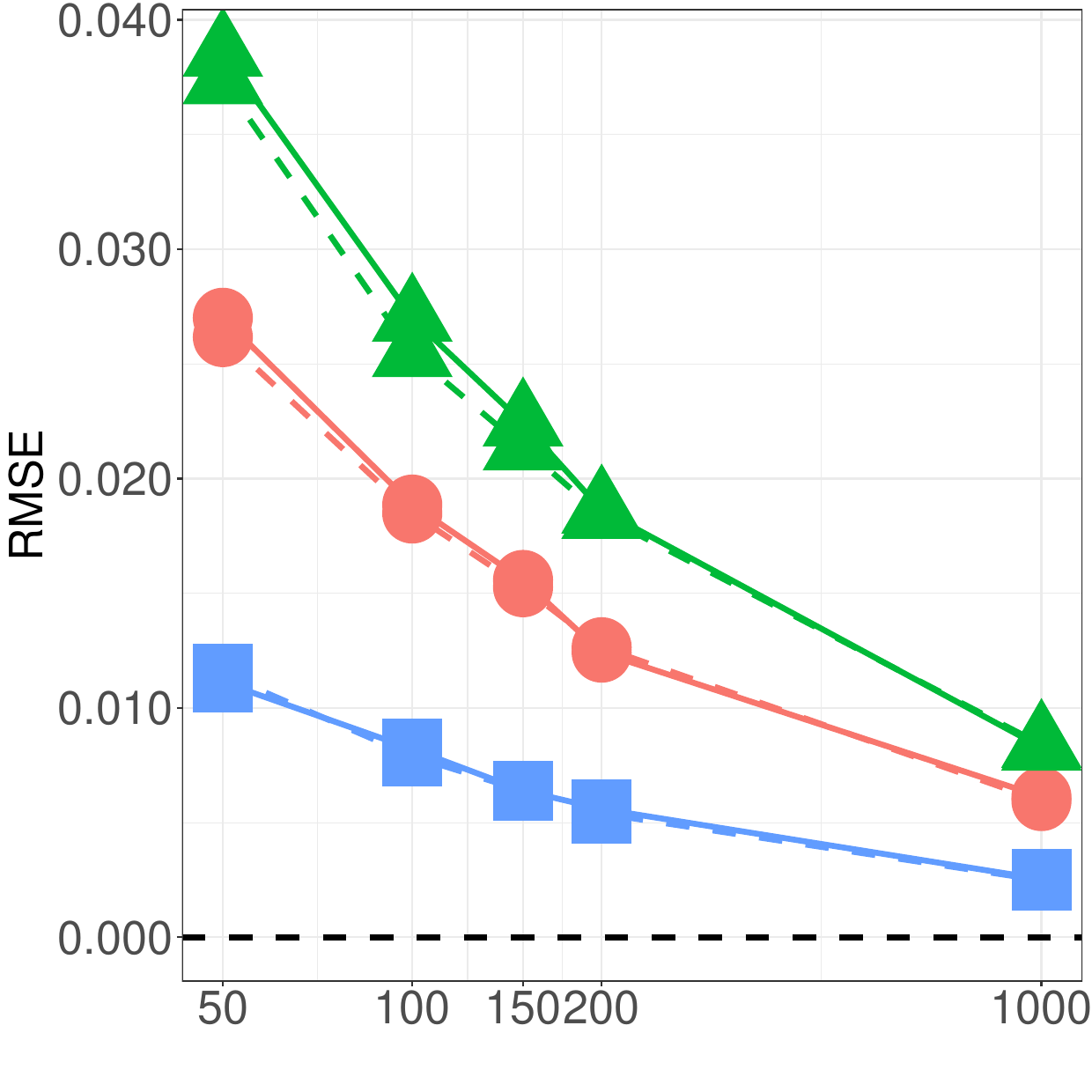}
        \end{minipage}%
        \begin{minipage}[t]{0.32\linewidth}
            \centering
            \includegraphics[width=5cm]{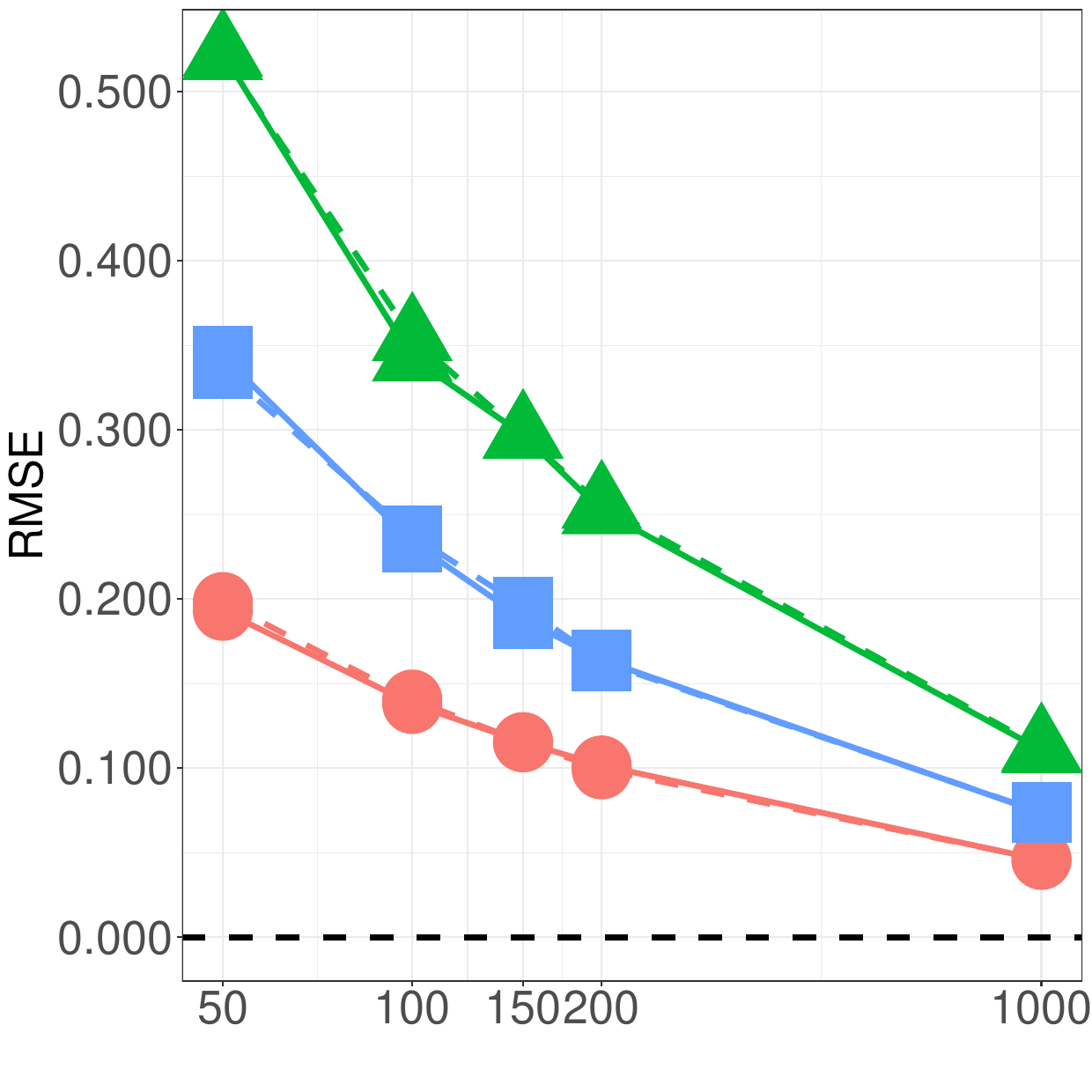}
        \end{minipage}%
        \begin{minipage}[t]{0.32\linewidth}
            \centering
            \includegraphics[width=5cm]{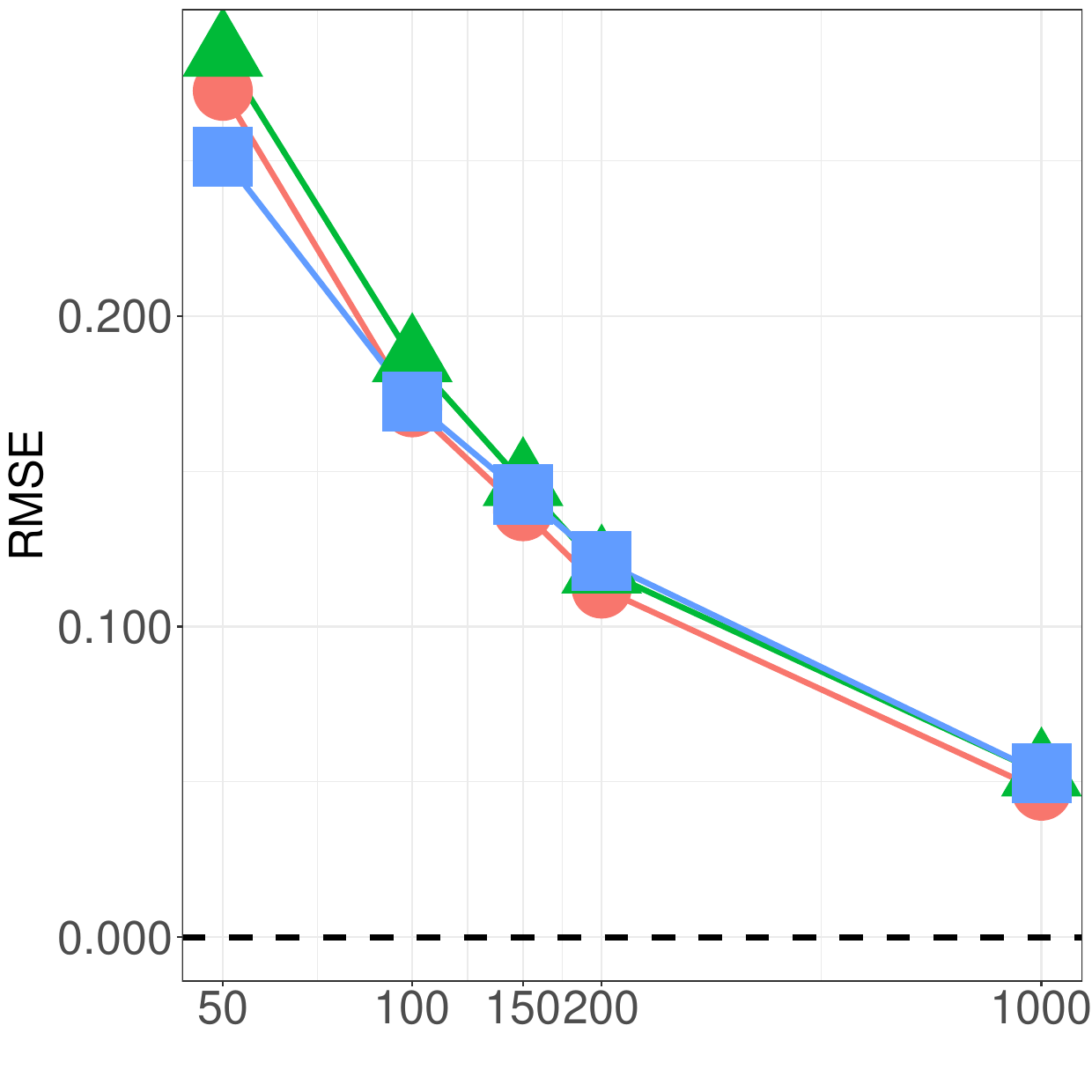}
        \end{minipage}
    \end{minipage}
    
    \begin{minipage}[t]{\linewidth}
        \centering
        \begin{minipage}[t]{0.32\linewidth}
            \centering
            \includegraphics[width=5cm]{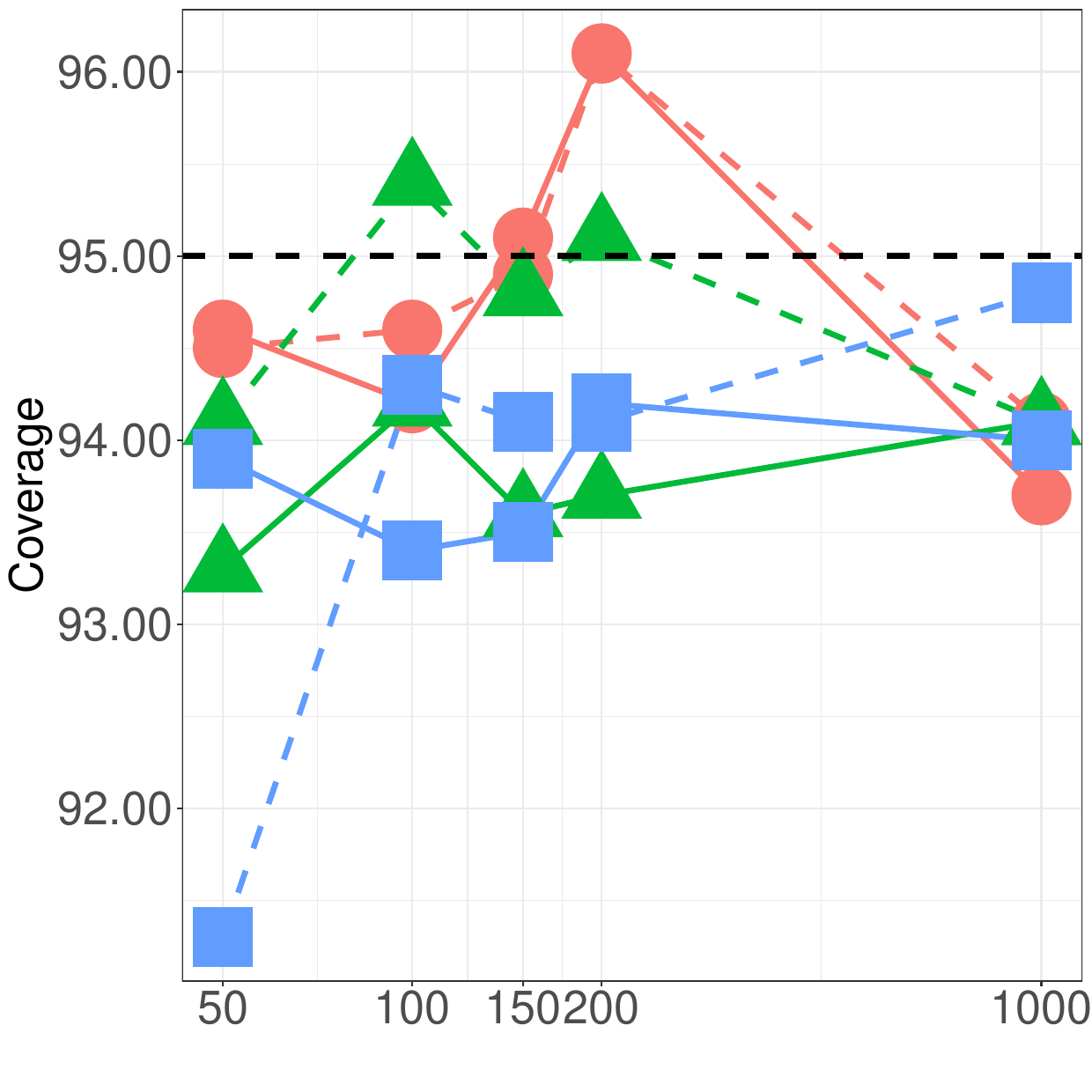}
        \end{minipage}%
        \begin{minipage}[t]{0.32\linewidth}
            \centering
            \includegraphics[width=5cm]{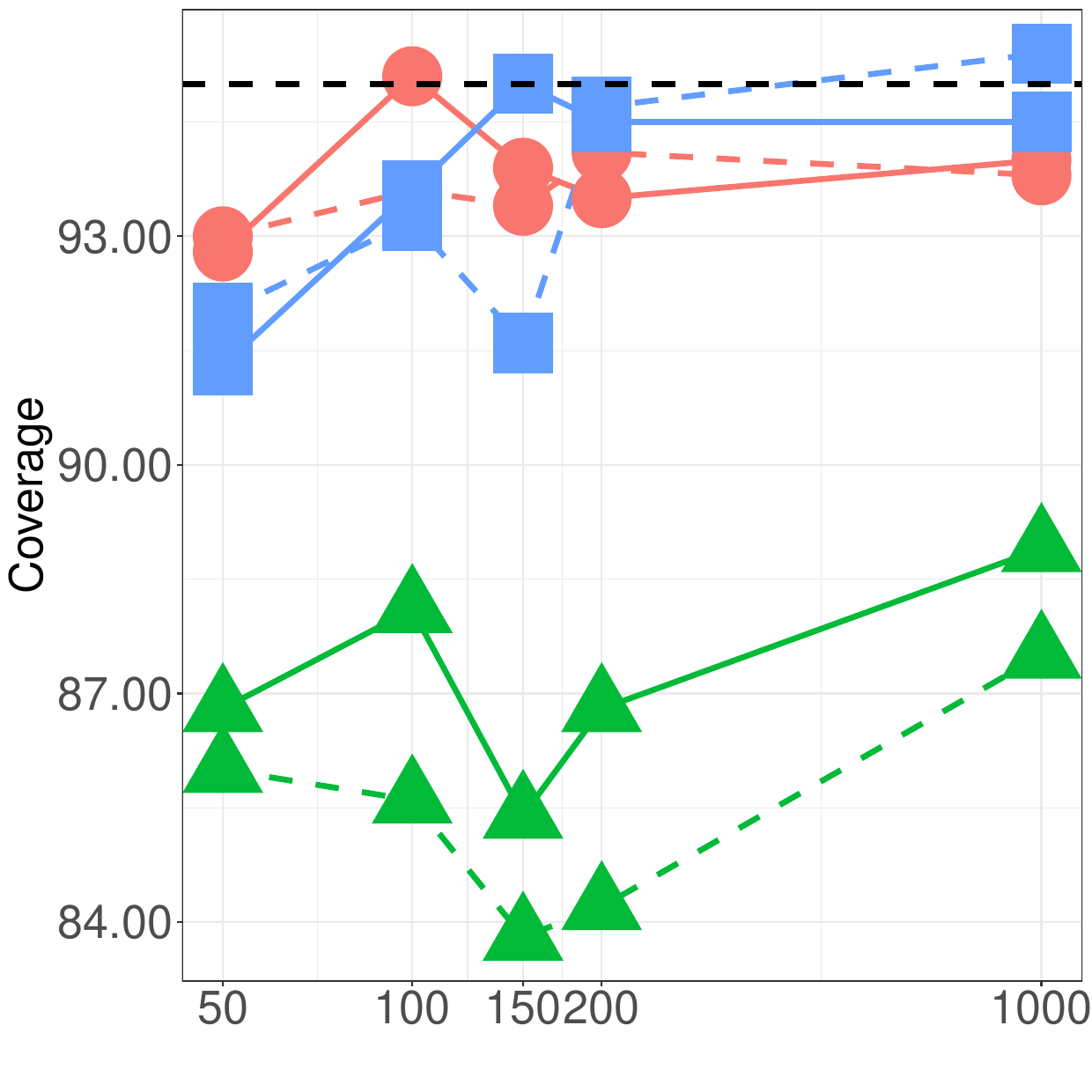}
        \end{minipage}%
        \begin{minipage}[t]{0.32\linewidth}
            \centering
            \includegraphics[width=5cm]{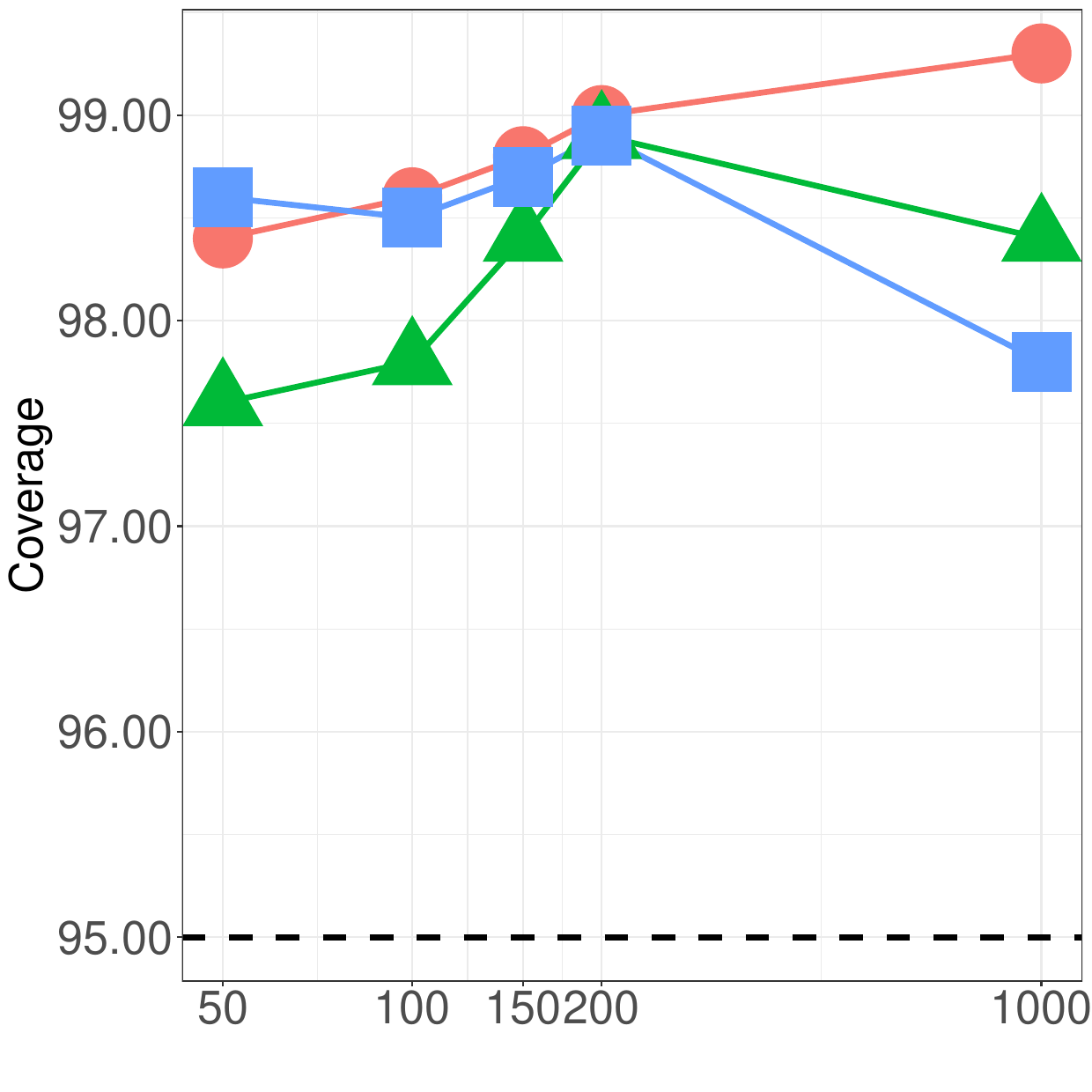}
        \end{minipage}
    \end{minipage}
    
    \caption{Bias (row 1), RMSE (row 2) and Coverage (row 3) of parameters $\mu_1$ and $\mu_2$ (column 1), $\sigma^2_1$ and $\sigma^2_2$ (column 2) and $\lambda$ (column 3) for Scenario 2.}
    \label{fig:fig2}
\end{figure}

In this scenario, as expected, the bias and the root mean square error (RMSE) approach zero as the sample size increases, indicating that the maximum likelihood estimators are asymptotically unbiased. The coverage probability for the parameters $\mu_1$ and $\mu_2$ are close to the nominal 95\% level across different parameter vectors $\bm{\theta}$ and sample sizes. However, for the parameters $\sigma_1^2$ and $\sigma_2^2$, the coverage is underestimated when $\bm{\theta}_2$ is assumed. Conversely, the coverage for the parameter $\lambda$ is overestimated regardless of the parameter vector $\bm{\theta}$ and the various sample sizes. The results can also be found in the Appendix in Tables \ref{tab:scenario-results-2.1} - \ref{tab:scenario-results-2.3}.
\subsection{Scenario 3}
\label{subsec:cena3}

Finally, in this scenario, the following vectors are taken as the true values of the parameters $\bm{\theta}_1 = (0.5, 0.5, 2, 2, 0)^{\top}$, $\bm{\theta}_2 = (0.5, 0.5, 5, 5, 0)^{\top}$ and $\bm{\theta}_3 = (0.9, 0.9, \sqrt{11}, \sqrt{11}, 0)^{\top}$. Figure~\ref{fig:superficie-e-contono-3} (see Appendix) illustrates the surface and contour plots for the generated samples. The joint moments ${\rm E}(y_1y_2)$ for $\bm{\theta_1}$, $\bm{\theta_2}$ and $\bm{\theta_3}$ are 0.25, 0.25 and 0.81, accordingly. For $\bm{\theta}_1$, the generated samples are concentrated around $0,5$, indicating an unimodal behavior. Similarly, for $\bm{\theta}_2$ and $\bm{\theta}_3$, the samples are concentrated near zero and one. 
Figure \ref{fig:fig3} presents the simulation results in this scenario. The parameter vectors $\bm{\theta}_1$, $\bm{\theta}_2$ and $\bm{\theta}_3$ are represented by the red, green, and blue, respectively. Solid and dashed lines correspond to the parameters associated with $y_1$ and $y_2$, respectively.
\begin{figure}[h!]
    
    \begin{minipage}[t]{\linewidth}
        \centering
        \begin{minipage}[t]{0.32\linewidth}
            \centering
            \includegraphics[width=5cm]{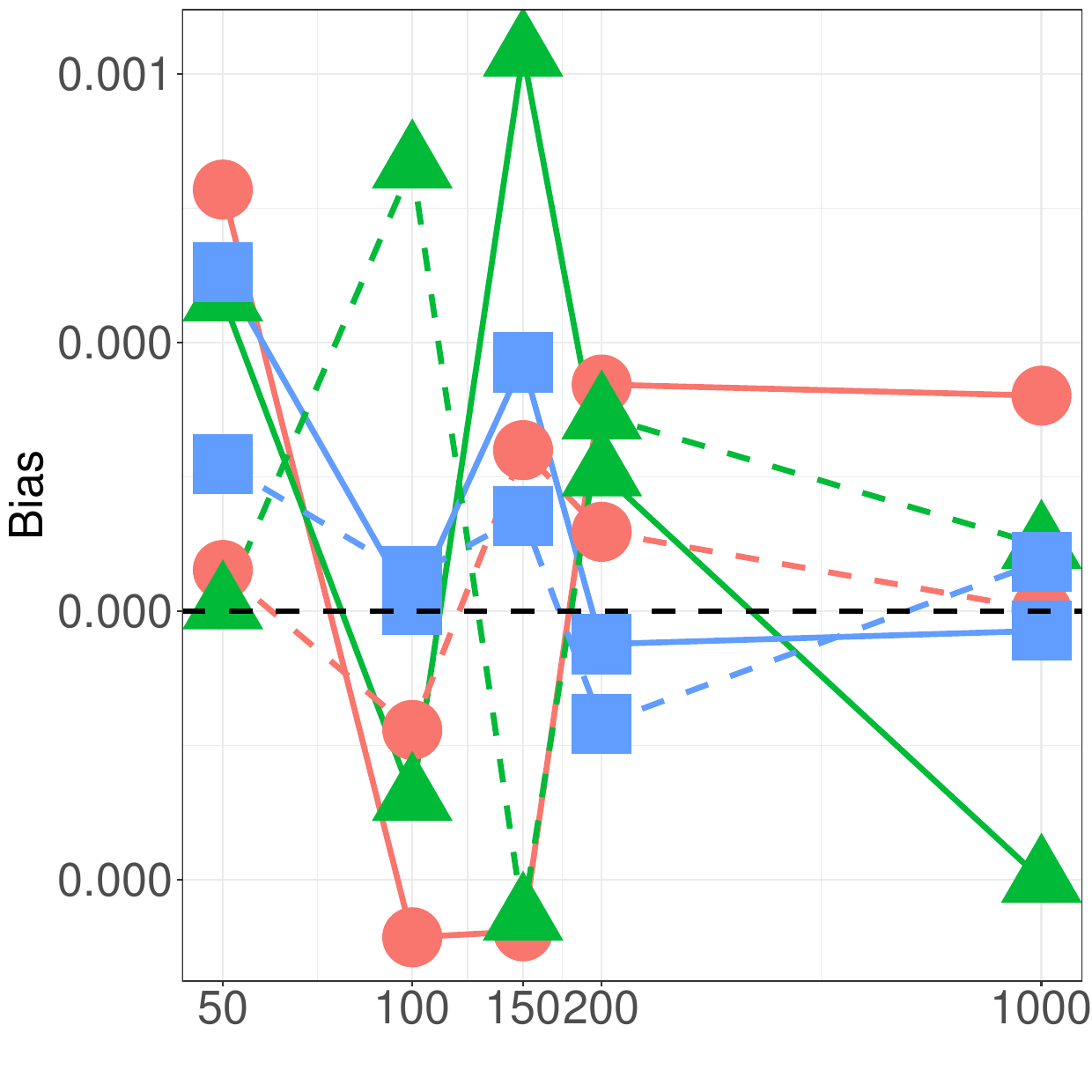}
        \end{minipage}%
        \begin{minipage}[t]{0.32\linewidth}
            \centering
            \includegraphics[width=5cm]{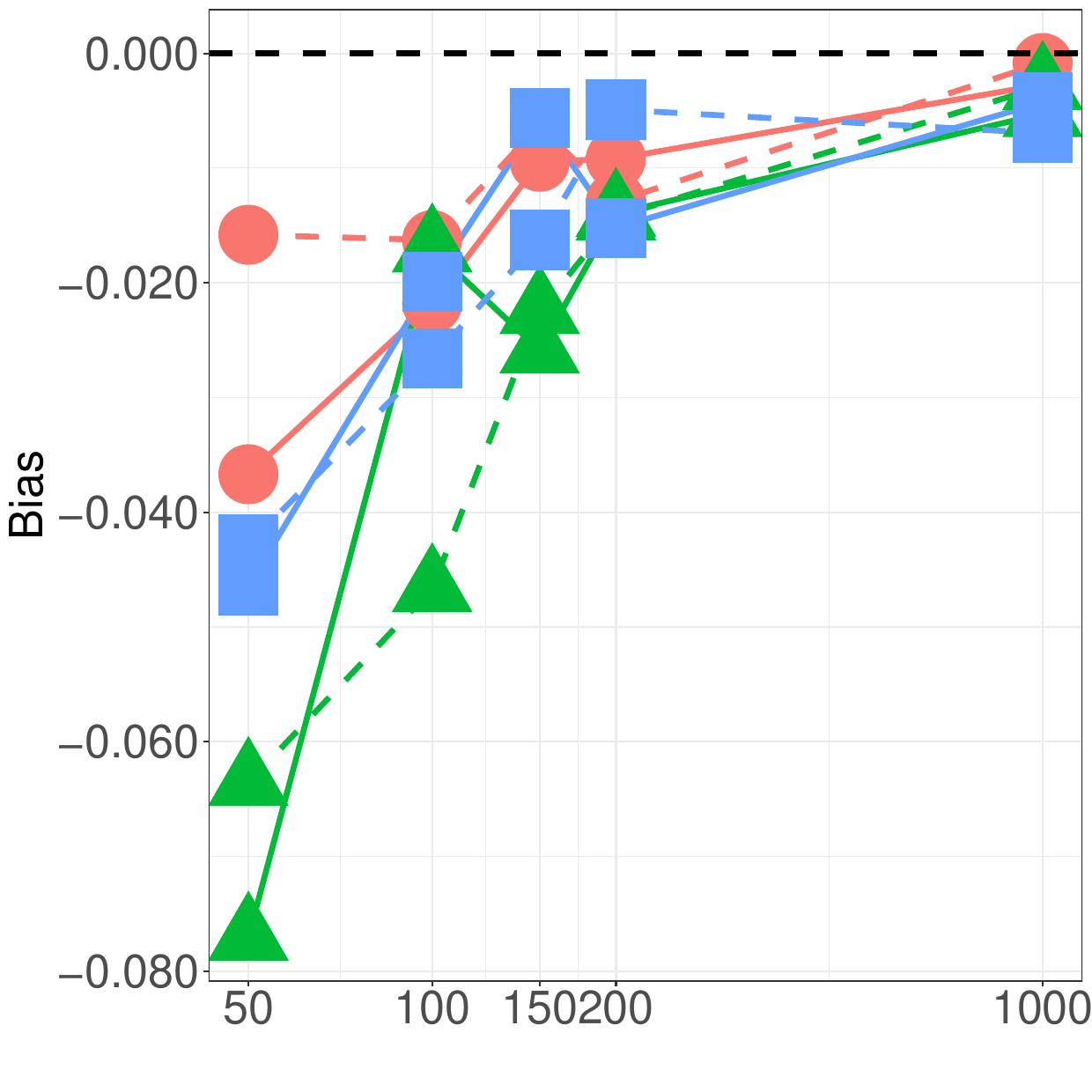}
        \end{minipage}%
        \begin{minipage}[t]{0.32\linewidth}
            \centering
            \includegraphics[width=5cm]{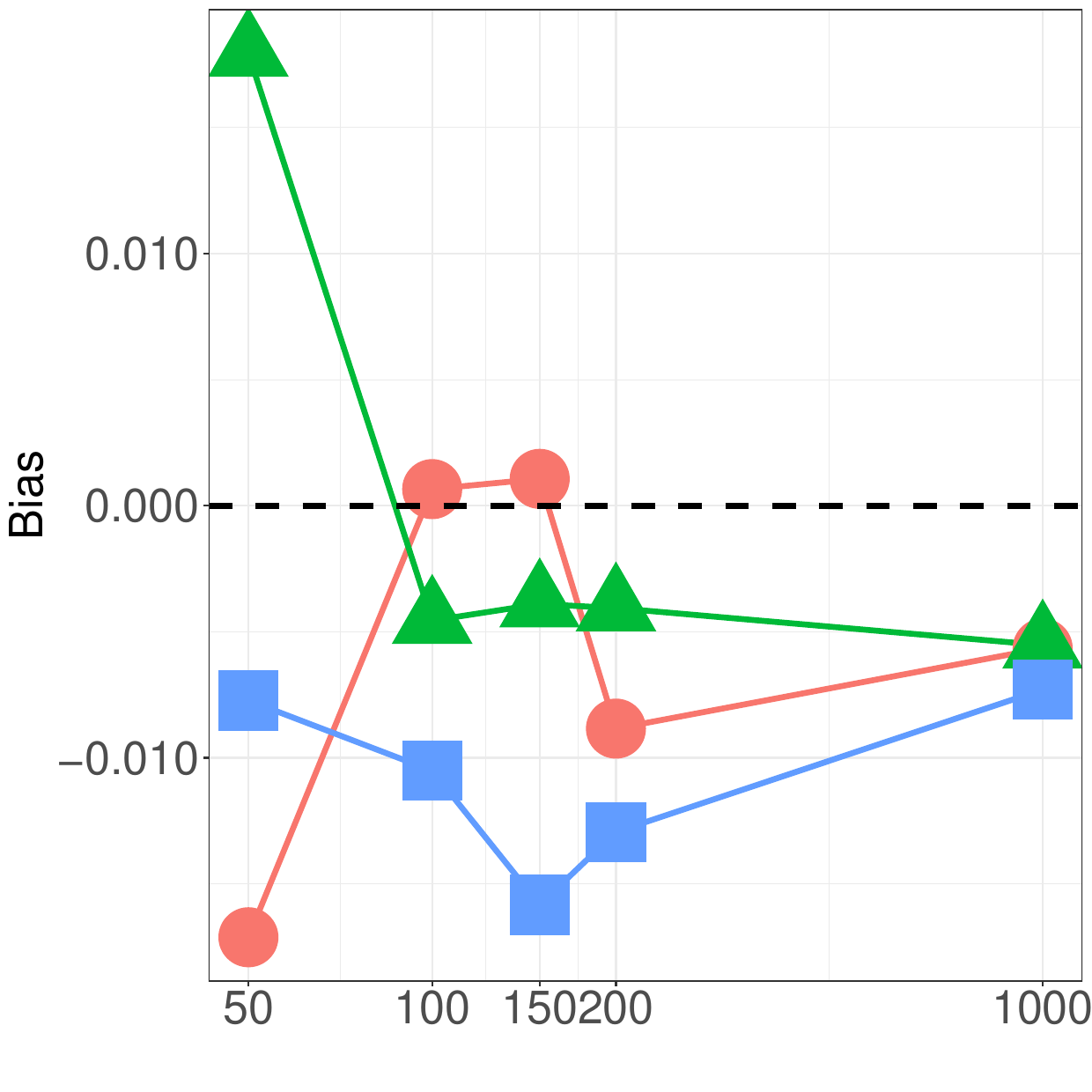}
        \end{minipage}
    \end{minipage}
    
    \begin{minipage}[t]{\linewidth}
        \centering
        \begin{minipage}[t]{0.32\linewidth}
            \centering
            \includegraphics[width=5cm]{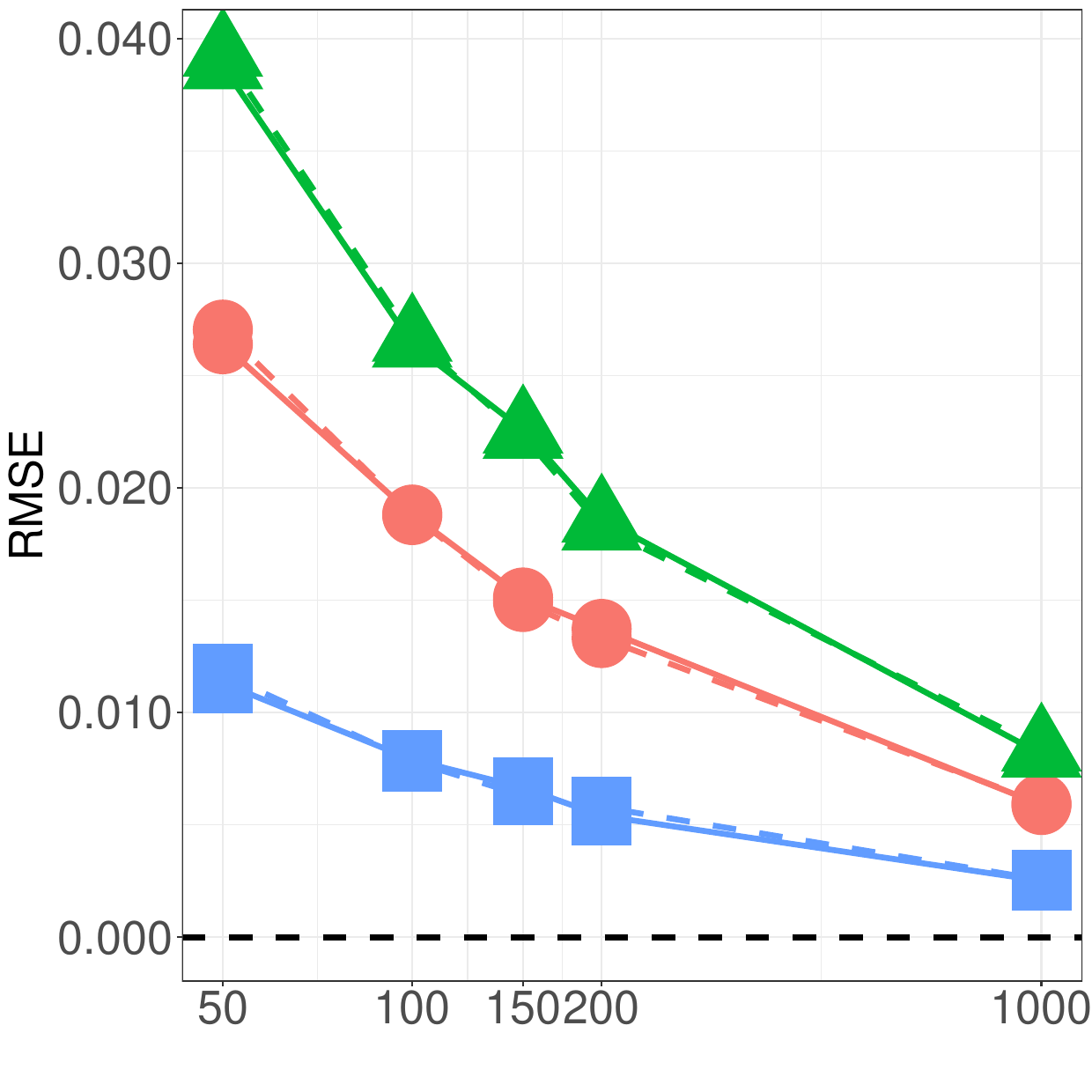}
        \end{minipage}%
        \begin{minipage}[t]{0.32\linewidth}
            \centering
            \includegraphics[width=5cm]{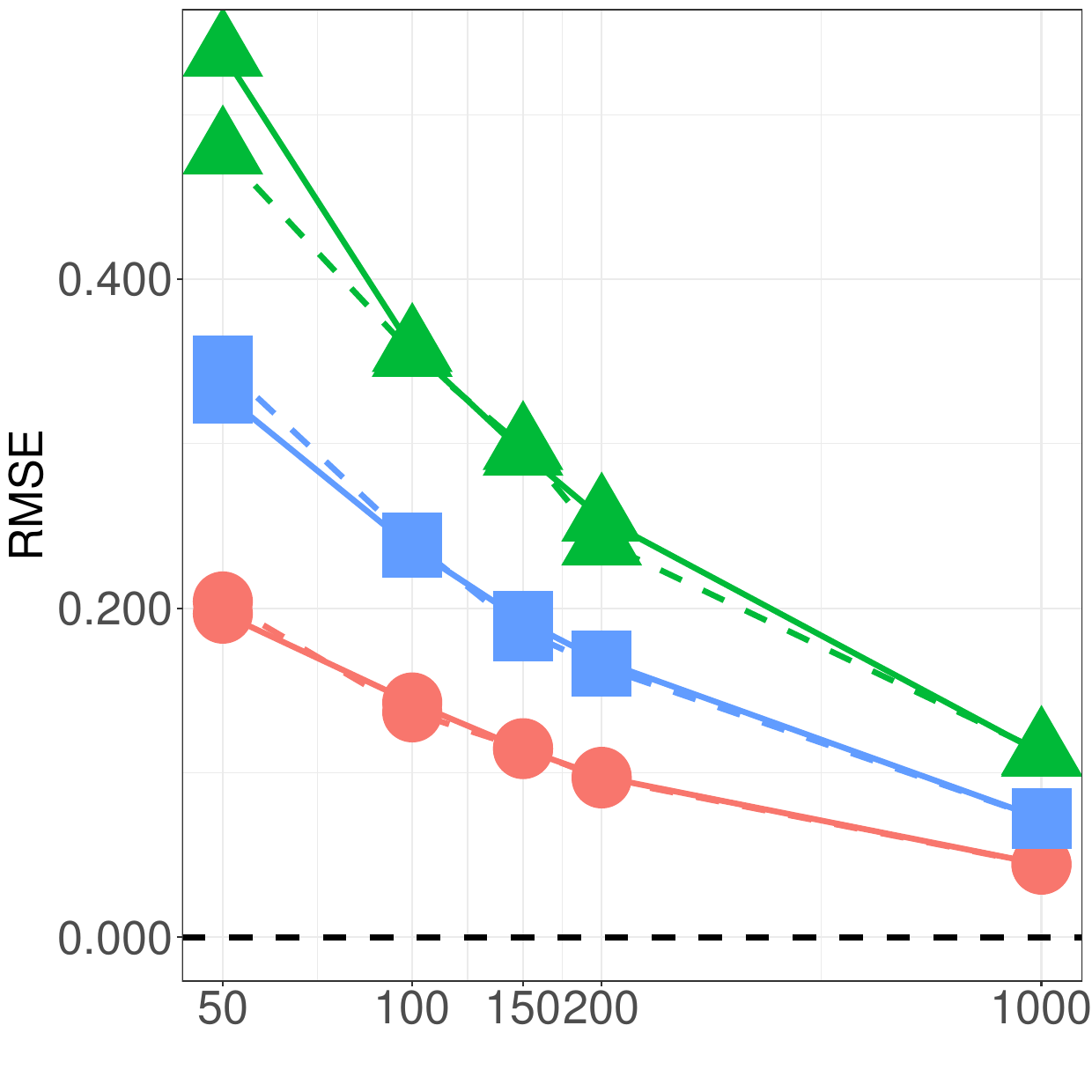}
        \end{minipage}%
        \begin{minipage}[t]{0.32\linewidth}
            \centering
            \includegraphics[width=5cm]{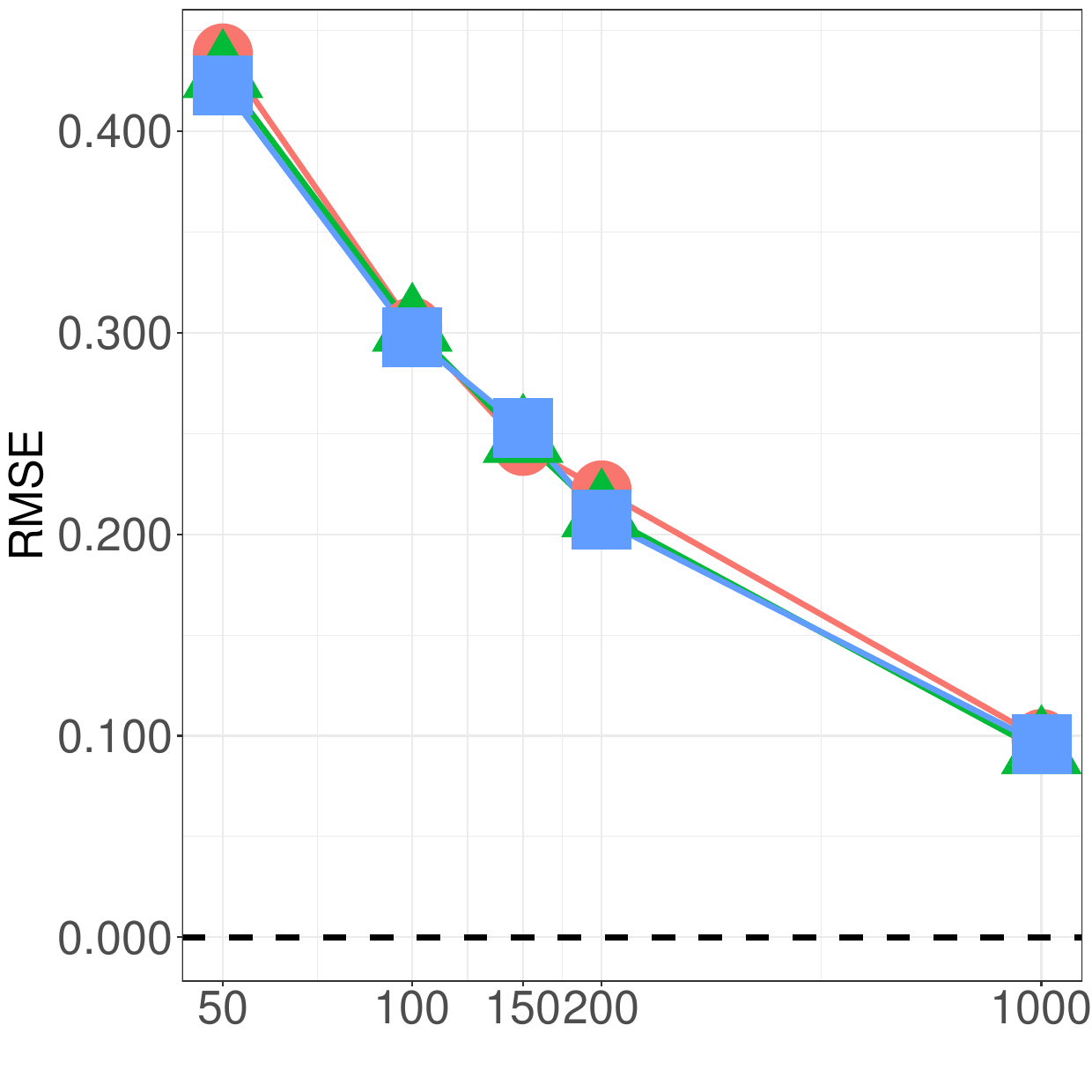}
        \end{minipage}
    \end{minipage}
    
    \begin{minipage}[t]{\linewidth}
        \centering
        \begin{minipage}[t]{0.32\linewidth}
            \centering
            \includegraphics[width=5cm]{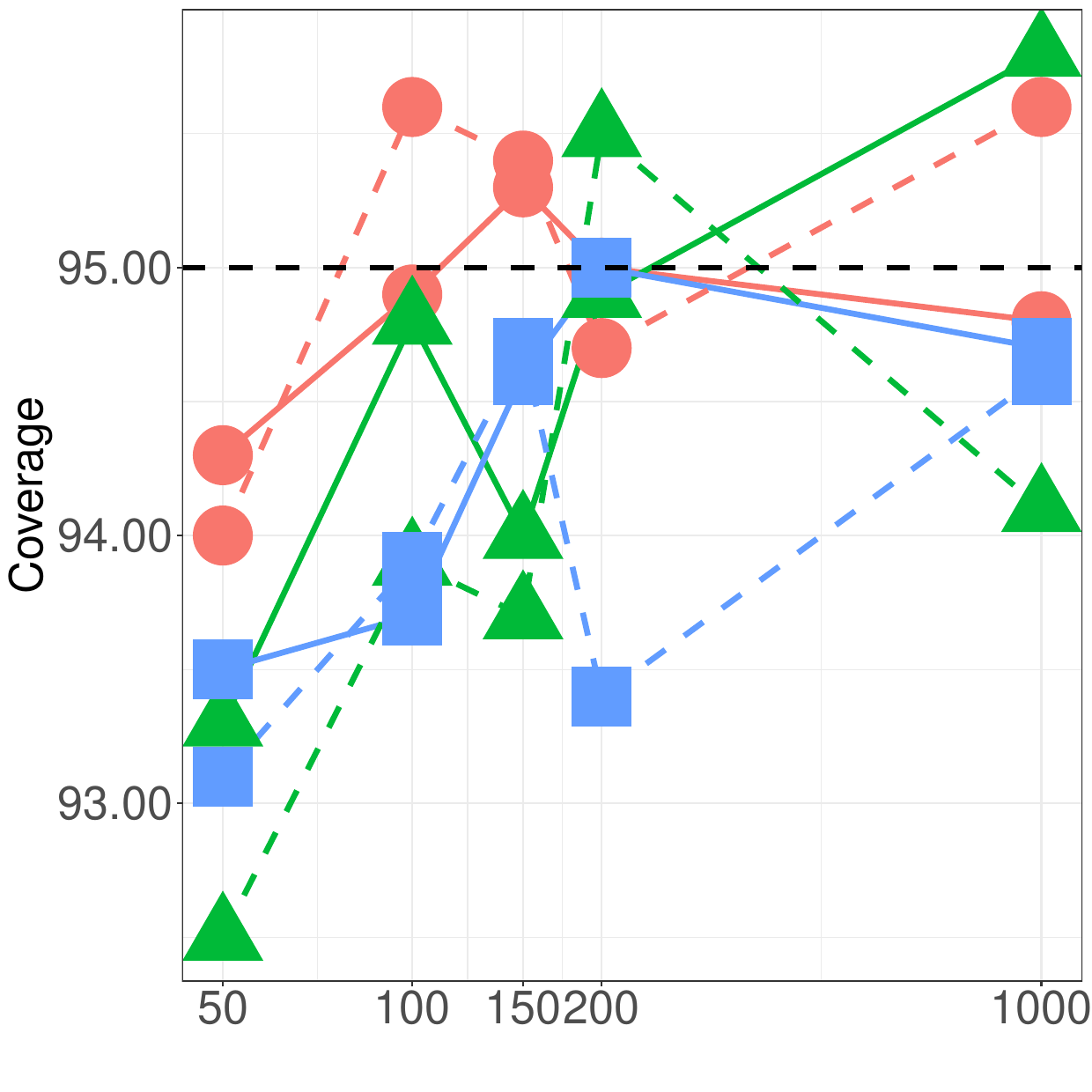}
        \end{minipage}%
        \begin{minipage}[t]{0.32\linewidth}
            \centering
            \includegraphics[width=5cm]{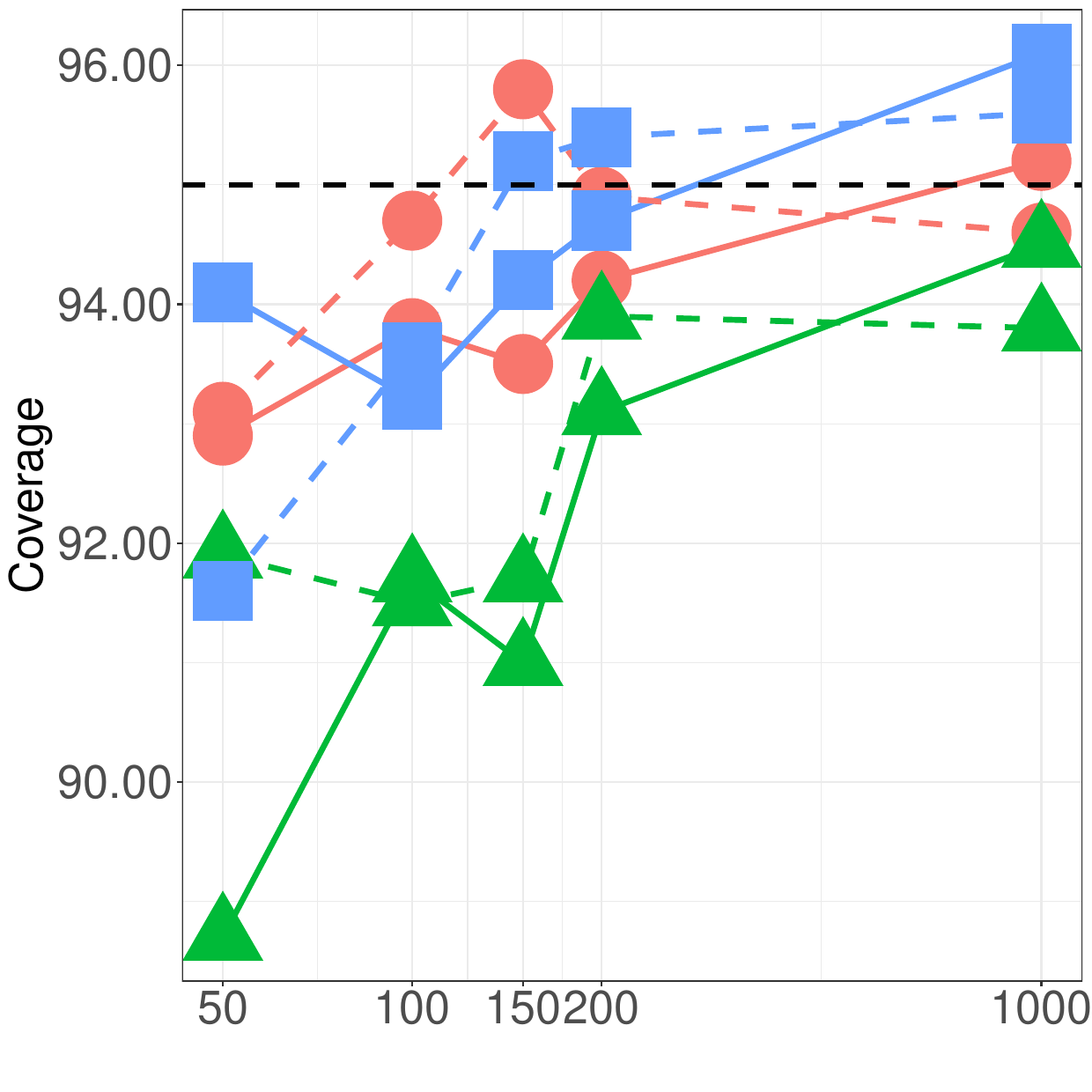}
        \end{minipage}%
        \begin{minipage}[t]{0.32\linewidth}
            \centering
            \includegraphics[width=5cm]{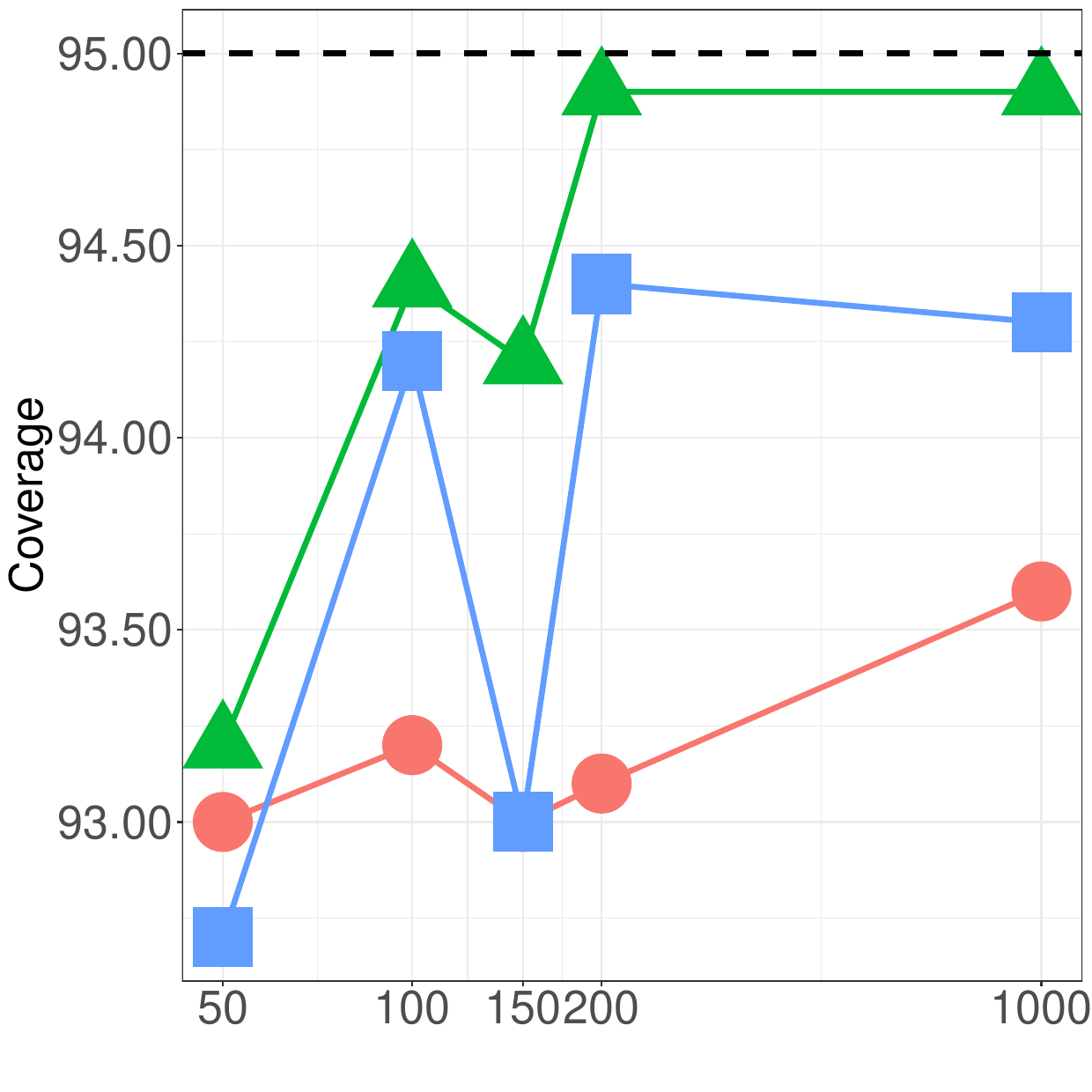}
        \end{minipage}
    \end{minipage}
    
    \caption{Bias (row 1), RMSE (row 2) and Coverage (row 3) of the parameters $\mu_1$ and $\mu_2$ (column 1), $\sigma^2_1$ and $\sigma^2_2$ (column 2) and $\lambda$ (column 3) for the Scenario 3.}
    \label{fig:fig3}
\end{figure}

In this scenario, the bias and root mean square error (RMSE) decreases to zero when the sample size increases. The coverage probability for the parameters $\mu_1$, $\mu_2$, $\sigma_1^2$ and $\sigma_2^2$ are close to the nominal 95\% level across different parameters vector $\bm{\theta}$ and sample sizes. However, for the parameter $\lambda$, only $\bm{\theta}_2$ achieves coverage close to the nominal 95\% level, while for $\bm{\theta}_1$ and $\bm{\theta}_3$, the coverage is underestimated. Again, the results can also be found in Tables \ref{tab:scenario-results-3.1} - \ref{tab:scenario-results-3.3} in the Appendix.

\section{Application}
\label{sec:apli}
The dataset used in this application was obtained from the 2014 Annual Report of the Regional Labor Court of the 5th Region (TRT5) in Bahia, Brazil. The mission of TRT5 is to promote justice in labor relations with efficiency, transparency, and swiftness, contributing to social harmony and strengthening citizenship within Bahia.

For this study, the outcomes are defined as follows: 
$y_1$, the Congestion Rate, represents the proportion of unresolved cases relative to the total cases processed within a year. $y_2$, the Conciliation Index, is calculated as the percentage of sentences and decisions resolved through agreements, relative to the total number of final decisions issued by the 88 Labor Courts in Bahia.
 
Figure \ref{fig:Boxplot} presents the boxplots and histograms for the variables $y_1$ (Congestion Rate) and $y_2$ (Conciliation Index). The graph reveals the presence of outliers in the Conciliation Index. These outliers correspond to the courts located in the cities of Itamaraju (\# 33), Simões Filho (\# 84), Santo Amaro (\# 80), and Candeias (\# 10 and \# 11).
\begin{figure}[h!]
    \centering
    \begin{minipage}[!]{0.32\linewidth}
        \includegraphics[width = 5cm]{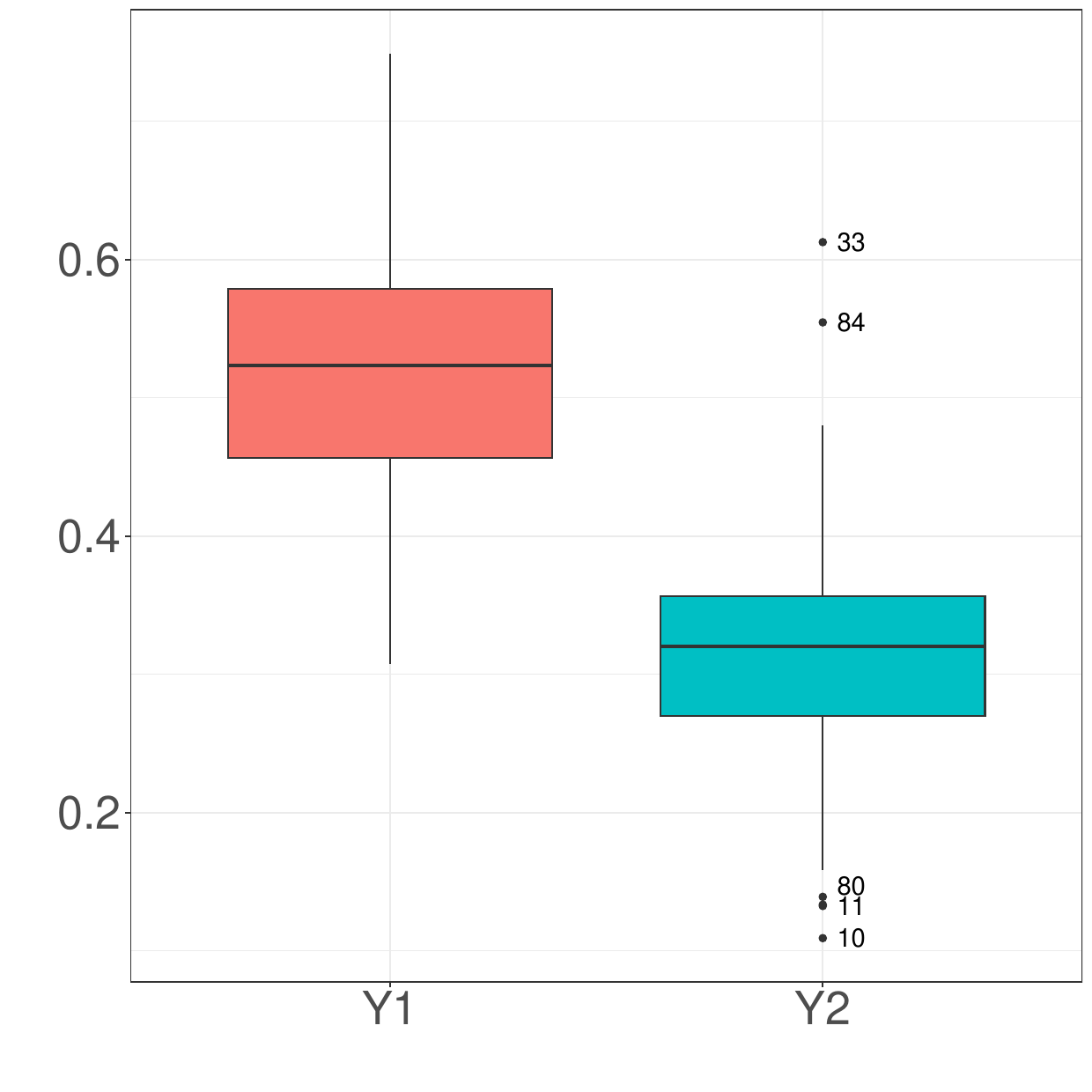}
    \end{minipage}
    \begin{minipage}[!]{0.32\linewidth}
        \includegraphics[width = 5cm]{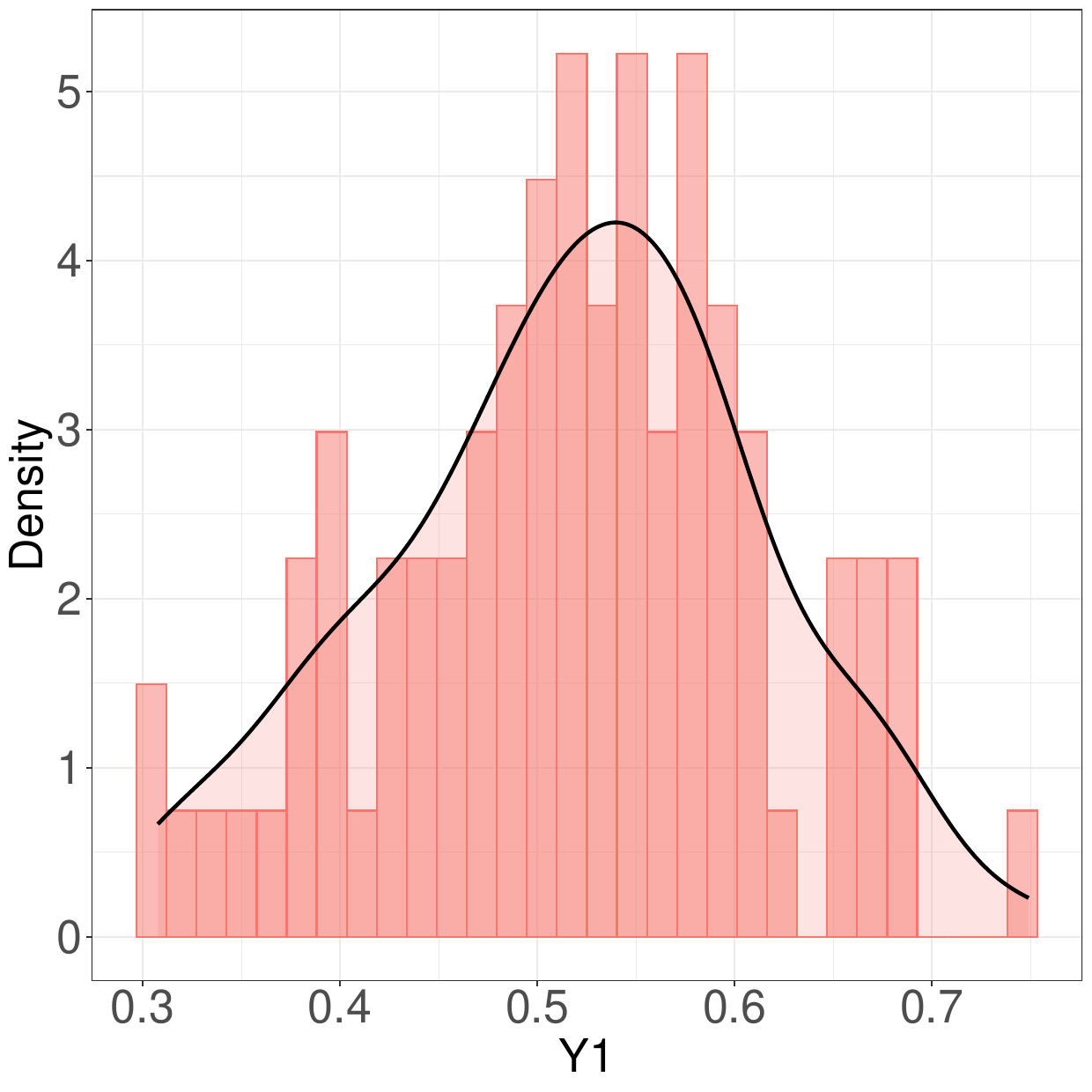}
    \end{minipage}
    \begin{minipage}[!]{0.32\linewidth}
        \includegraphics[width = 5cm]{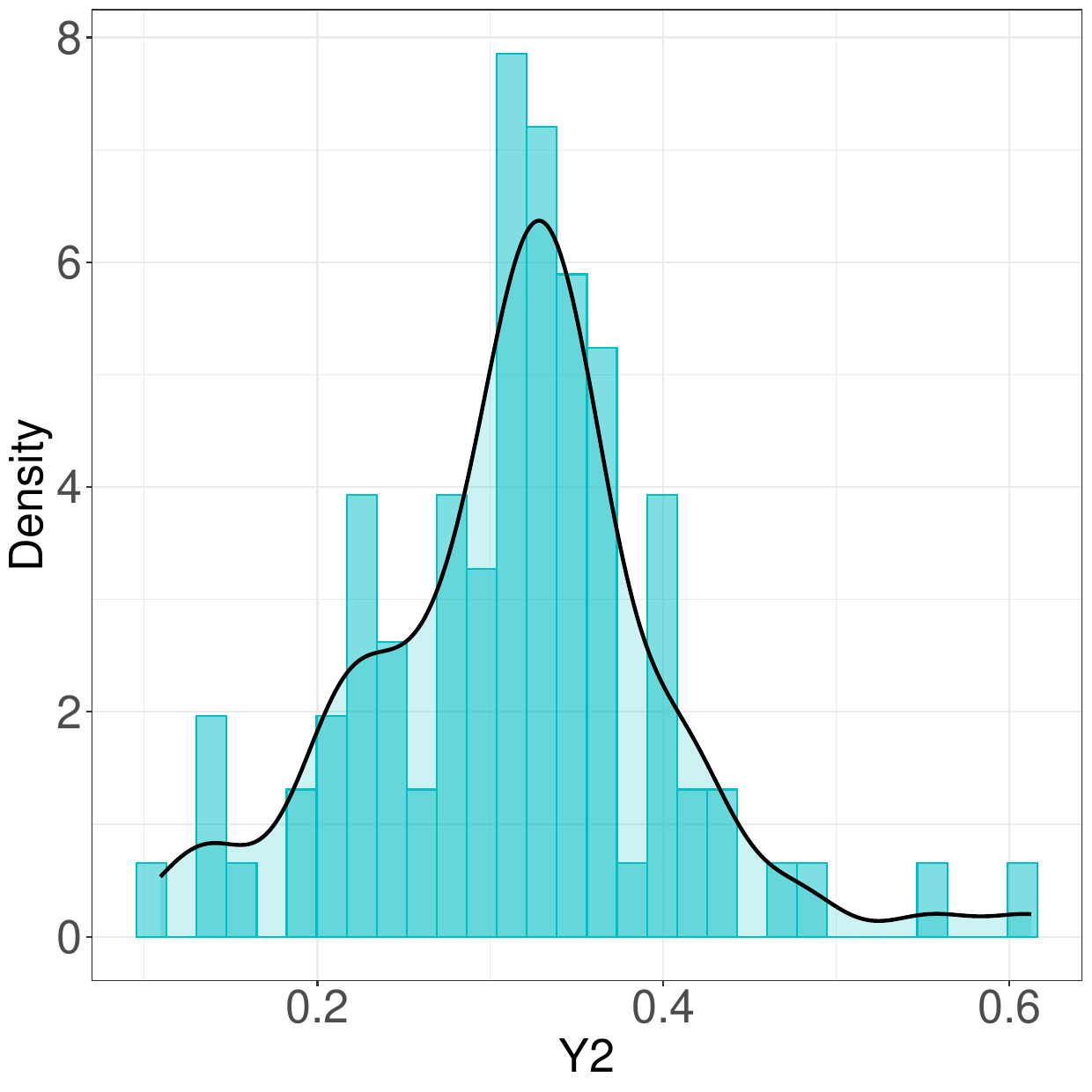}
    \end{minipage}
    \caption{Boxplot and histograms for the $y_1$(Congestion Rate) and $y_2$ (Conciliation Index) variables.}
    \label{fig:Boxplot}
\end{figure}
Table \ref{tab:medidas-descritivas} highlights the central tendency and dispersion measures for the variables. 
\begin{table}[h!]
    \centering
    \caption{Descriptive measures of position, dispersion, asymmetry, kurtosis, and relative position of variables $y_1$ and $y_2$.}
    \vspace{0.5cm}
    \begin{small}
    \begin{tabular}{c|RRRRRRRRRR}
        \hline
          & \text{Min.} & \text{Max.} & Q_1 & Q_3 & \text{Mean} & \text{Median} & \text{\begin{tabular}[c]{@{}l@{}}Standard \\ Errors\end{tabular}} & \text{Asymmetry} & \text{Kurtosis} \\ 
        \hline
        $y_1$ & 0.308 & 0.749 & 0.456 & 0.579 & 0.524 & 0.524 &0.095 & -0.160 & -0.382 \\ 
        $y_2$ & 0.109 & 0.613 & 0.270 & 0.357 & 0.315 & 0.320 &0.086 &  0.288 &  1.273 \\
        \hline
    \end{tabular}
    \end{small}
    \label{tab:medidas-descritivas}
\end{table}
As shown in Table \ref{tab:medidas-descritivas}, the average Conciliation Index is notably lower than the Congestion Rate, despite conciliation being the fastest route to a resolution. Given that both $y_1$ and $y_2$ are constrained to the (0,1) interval, we propose using a model that accounts for this characteristic. Accordingly, we applied the bivariate Simplex distribution defined in Section \ref{sec:bSimplex} to obtain the estimates, standard errors, and confidence intervals, which are reported in Table \ref{tab:estimativas-da-aplicacao}. The estimate for the joint expectation, $\widehat{E}(y_1y_2)$, is 0.162, and the maximum likelihood estimators are obtained when the dependence parameter $\lambda=0.072$ (positive dependence). This suggests that the Congestion Rate and Conciliation Index tend to increase together. 

\begin{table}[!h]
    \centering
    \caption{Estimates, standard errors, and confidence intervals with a confidence coefficient of 95\%.}
    \vspace{0.5cm}
    \begin{small}
    \begin{tabular}{c|RRC}
        \hline
        \text{Parameter} & \text{Estimate} & \text{Standard error} & \text{Confidence Interval (95\%)} \\ 
        \hline
        $\mu_1$         &   0.518   &   0.010   &   (0.498~;~0.537) \\ 
        $\mu_2$         &   0.313   &   0.010   &   (0.294~;~0.332) \\ 
        $\sigma^2_1$    &   0.799   &   0.060   &   (0.681~;~0.917) \\ 
        $\sigma^2_2$    &   0.949   &   0.072   &   (0.809~;~1.089) \\ 
       \hline
    \end{tabular}
    \end{small}
    \label{tab:estimativas-da-aplicacao}
\end{table}

Figure \ref{fig:ajuste} displays the surface and contour plots based on the estimates obtained. The contour graph, in particular, demonstrates a good fit of the model to the data. Additionally, the model applied to the dataset produced graphs that closely resemble those of $\bm{\theta}_1$ in scenario 3 (Figure~\ref{fig:superficie-e-contono-3}).
\begin{figure}[h!]
    \begin{minipage}[!]{0.45\linewidth}
        \includegraphics[width=\linewidth]{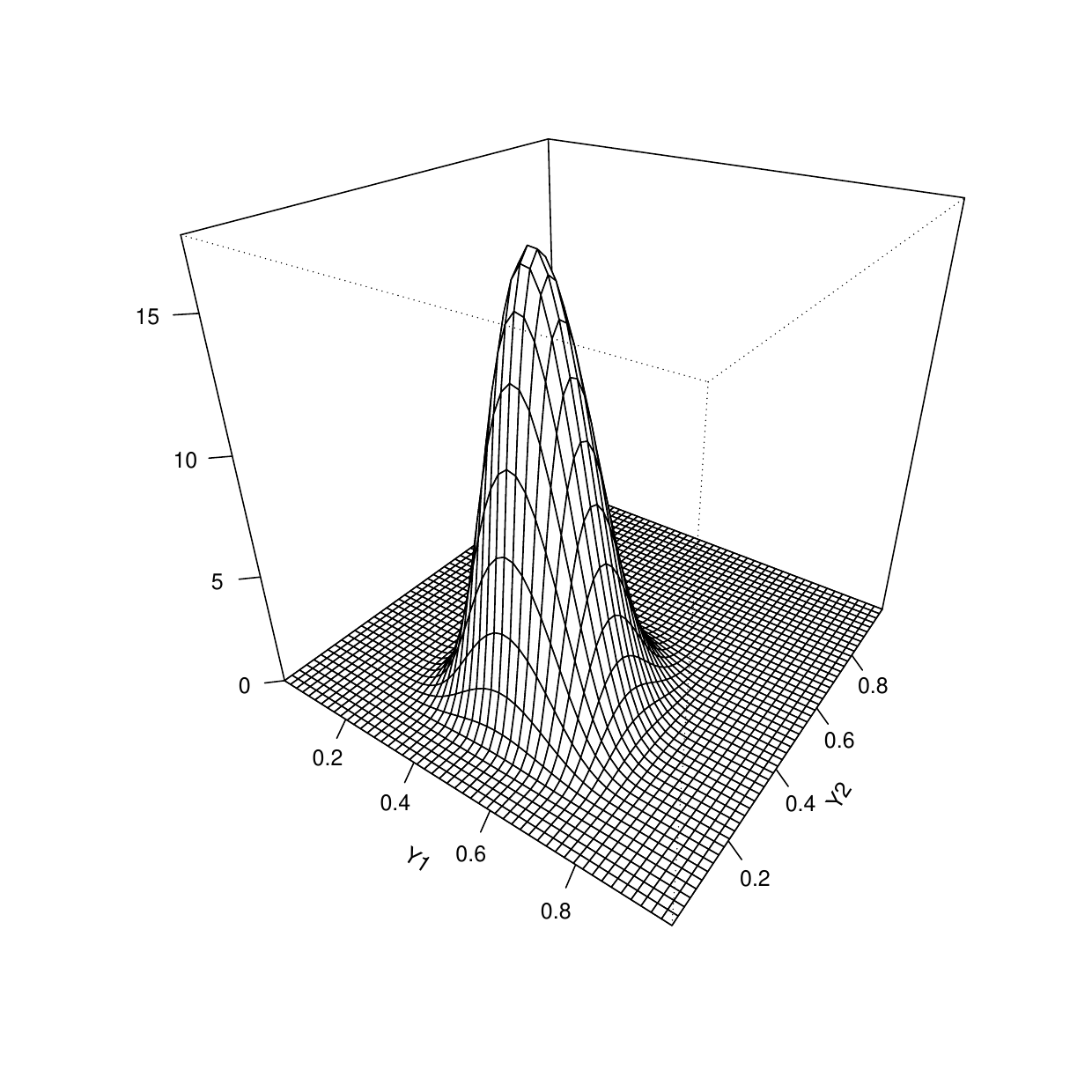}
    \end{minipage}
    \hfill
    \begin{minipage}[!]{0.45\linewidth}
        \includegraphics[width=\linewidth]{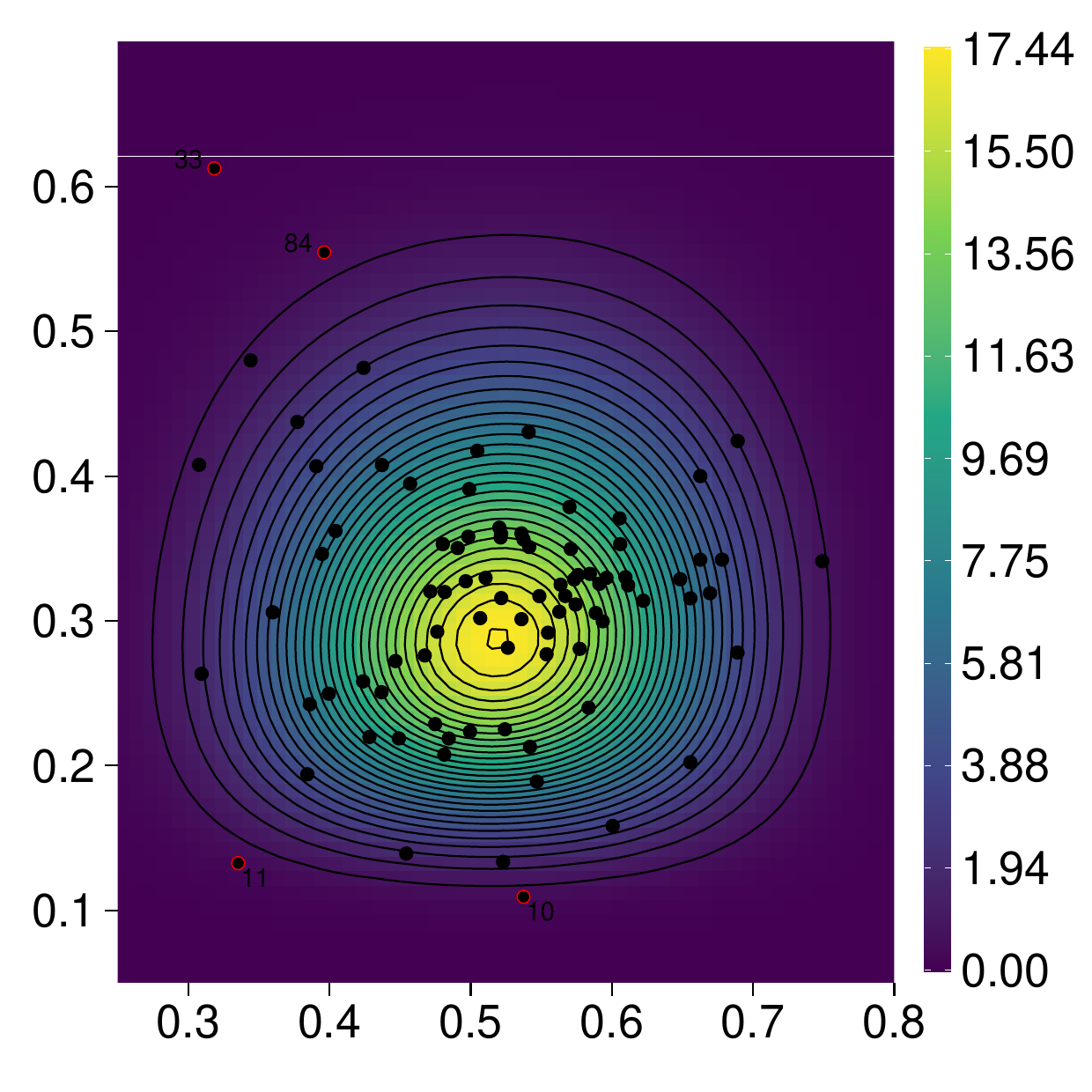}
    \end{minipage}
    \caption{Congestion Rate surface graph (left) and contour graph (right) ($y_1$) and the Conciliation Index ($y_2$).}
    \label{fig:ajuste}
\end{figure}

\section{Conclusions}
\label{sec:conclu}

In this article, we propose the bivariate Simplex distribution obtained via copulas as an alternative to existing distribution to draw inferences on two variables belonging to the standard unit interval (0,1). Additionally, through algebraic manipulations, we derived an analytic form for the joint expectation {\rm E}$(y_1y_2)$. In our Monte Carlo simulation study, we investigated the asymptotic properties of the maximum likelihood estimators. Our findings indicate that these estimators exhibit desirable asymptotic properties across different scenarios, with the dependence parameter $\lambda$ taking values of 1, -1, and 0. Finally, we validated our proposal using a real dataset from the Jurimetry area$-$the Annual Report of the Regional Labor Court 5$^{a}$ region (TRT5) of Bahia, Brazil, to validate our proposal.  



\bibliographystyle{natbib}      
\bibliography{biblio}   

\newpage
\section*{Appendix}

\subsection*{\textbf{Appendix A:}  Mathematical results used in the proof of Theorem 1}
\label{subsec: Ap1}

Let $K_\nu$ be the modified Bessel function of the second kind. From the definition, we know that  for $\nu$ real and $z$ positive, $K_\nu(z)$ is real. Additionally, the following symmetry holds $K_{-\nu}(z)=K_\nu(z).$ Below we state some properties of this function.
\begin{proposition}
\label{teobessel}
 Let $\beta$ and $\gamma$ be positive real numbers, and let $\nu \in \mathbb{R}$. Then we have the following integral representation for the modified Bessel function of the second kind:
$$\int_0^\infty {x^{\nu-1}e^{-\frac{1}{2}(\beta x^{-1}+\gamma x)}} dx =2(\frac{\beta}{\gamma})^{\nu/2}K_{\nu}(\sqrt{\beta\gamma}).$$
\end{proposition}
\begin{proof}
See  \citet{Kropac1982}.
\end{proof}
\begin{proposition}
\label{teobessel1}
Let $z\in \mathbb{C}$ such that the real part of $z$, $\Re(z)$, is positive. Then,
\begin{enumerate}
  \item  $\int_0^\infty K_1(2z\cosh{(t)})dt=\frac{1}{2}K_{1/2}^2(z)$ 
  \item  $\int_0^\infty K_0(2z\cosh{(t)})\cosh(t)dt=\frac{1}{2}K_{1/2}^2(z)$ 
\end{enumerate}
\end{proposition}
\begin{proof}
See expression 10.32.17 (\url{https://dlmf.nist.gov/10.32}).
\end{proof}
\begin{corollary}
\label{lema1}
Let $a\in \mathbb{C}$ be such that $\Re(a)>0.$ Then
\begin{eqnarray*}
\int_1^\infty \frac{1}{x}K_1\Big(a \frac{x^2+1}{x}\Big)dx=\frac{\pi}{4a}e^{-2a}    
\end{eqnarray*}
\end{corollary}
\begin{proof}
Considering $x=e^t$ we have 
\begin{eqnarray*}
\int_1^\infty \frac{1}{x}K_1\Big(a \frac{x^2+1}{x}\Big)dx=\int_0^\infty K_1(2a \cosh(t))dt    
\end{eqnarray*}
Using Proposition \ref{teobessel1}  with $z=a$ and the fact that $K_{1/2}(a)=\sqrt{\frac{\pi}{2a}}e^{-a}$ we obtain the result.
\end{proof}
\noindent Additionally, we introduce the notion of asymptotic equivalence and little-o notation:
\begin{definition}

Two functions \( f \) and \( g \) are asymptotic equivalents as \( x \) approaches \( a \) if 
\[
\lim_{x \to a}\frac{f(x)}{g(x)} = 1.
\]
This relationship is denoted by \( f \sim_a g \).
\end{definition}
We observe that the asymptotic equivalence relation is transitive, that is, if \( f \sim_a g \) and \( g \sim_a h \), then \( f \sim_a h \).
\begin{definition}
A function \( f(x) \) is \( o(g(x)) \) for \( x \to a \) if \( f(x) \) grows slower than \( g(x) \) as \( x \) approaches \( a \). In simpler terms, 
\[
\lim_{x \to a}\frac{f(x)}{g(x)} = 0.
\]
\end{definition}
\noindent From now on, we denote $L_{\nu}$ as the modified Struve function, and $\text{ph}(z)$ denotes the phase of the complex number $z$.
\begin{lemma}
\label{teostruve}
 Let $z\in \mathbb{C}$ such that $|\text{ph}(z)|<\frac{\pi}{2}$, and let $\nu \in \mathbb{R}$ such that $\nu \pm \frac{1}{2} \notin -\mathbb{N}$. Then the following asymptotic expansion holds for $|z| \gg 1$:
\begin{equation}\label{K}
z K_{\nu+1}(z)L_\nu(z) \thicksim c_\nu\sqrt{\frac{z\pi}{2}}z^{\nu-1}e^{-z}+\frac{1}{2}.
\end{equation}
\end{lemma}
\begin{proof}

In this proof, we refer to  \citep[p.249, 252, 288, 293]{TheHanbookMath} for  know results used. First we observe that $L_\nu(z) =M_\nu(z)+I_\nu(z)$. Also, the following asymptotic equivalences hold
$$K_{\nu}(z) \thicksim \sqrt{\frac{\pi}{2z}}e^{-z},$$ $$ I_{\nu}(z) \thicksim \sqrt{\frac{1}{2\pi z}}e^{z},$$ and
$$M_\nu(z)\thicksim \frac{1}{\pi}\sum_{k \geq 0}c^k_\nu z^{\nu-2k-1},$$ 
where  $M_{\nu}$  is the modified Struve function,  \(I_{\nu}\) is the modified Bessel function and $c^k_\nu=\frac{(-1)^{k+1}\Gamma(k+1/2)(1/2)^{\nu-2k-1}}{\Gamma(\nu+1/2-k)}.$
Since $\sum_{k \geq 1}c^k_\nu z^{\nu-2k-1}=o(c_\nu^0 z^{\nu-1})$, then $\sum_{k \geq 0}c^k_\nu z^{\nu-2k-1}\thicksim c_\nu^0 z^{\nu-1}$
and by transitivity
$M_\nu(z)\thicksim \frac{1}{\pi}c_\nu^0 z^{\nu-1}.$
Therefore, 
$$L_\nu(z) =M_\nu(z)+I_\nu(z)\thicksim \frac{1}{\pi}c_\nu^0 z^{\nu-1}+ \sqrt{\frac{1}{2\pi z}}e^{z}.$$
Finally,
$$ z K_{\nu+1}(z)L_\nu(z) \thicksim c_\nu^0 \sqrt{\frac{1}{2\pi}}z^{\nu-1/2}e^{-z}+\frac{1}{2}\thicksim \frac{1}{2}.$$
\end{proof}
\begin{corollary}
\label{lemaproduto}
\begin{eqnarray*}
  \lim_{x->+\infty}x\Big(K_{0}(x)L_{-1}(x)+K_{1}(x)L_{0}(x)\Big)=1.  
\end{eqnarray*}
\end{corollary}
\begin{proof}
This proof is straightforward from Proposition \ref{teostruve}.
\end{proof}
\begin{proposition}
\label{teobessel2}
Let $z\in \mathbb{C}$. Then:
$$\int K_0(z) \, dz = \frac{\pi}{2}z(K_{0}(z)L_{-1}(z)+K_{1}(z)L_{0}(z))+C.$$
\end{proposition}
\begin{proof}
 See expression 10.43.2 (\url{https://dlmf.nist.gov/10.43})
\end{proof}
\begin{lemma}
\label{lema0}
For $a>0$, we have
\begin{eqnarray*}
    \int_1^\infty K_0\Big(a \frac{x^2+1}{x}\Big)dx=\frac{\pi}{4a}e^{-2a}+\frac{\pi}{2}\Big[\frac{1}{2}-K_{0}(2a)L_{-1}(2a)-K_{1}(2a)L_{0}(2a)\Big].
\end{eqnarray*}    
\end{lemma}
\begin{proof}
Let $x=e^t$. Then
\begin{eqnarray*}
\int_1^\infty K_0\Big(a \frac{x^2+1}{x}\Big)dx&=&\int_0^\infty e^t K_0(2a \cosh(t))dt,\\
                                    &=&\int_0^\infty \cosh(t) K_0(2a \cosh(t))dt+\int_0^\infty \sinh(t) K_0(2a \cosh(t))dt.    
\end{eqnarray*}
Using  Proposition \ref{teobessel1} with $z=a$ and the fact that $K_{1/2}(a)=\sqrt{\frac{\pi}{2a}}e^{-a}$ we find that
\begin{eqnarray*}
  \int_0^\infty \cosh(t) K_0(2a \cosh(t))dt=\frac{\pi e^{-2a}}{4a}.  
\end{eqnarray*}    
For the remaining integral, let $y=2a\cosh(t)$, yielding
\begin{eqnarray*}
  \int_0^\infty \sinh(t) K_0\Big(2a \cosh(t)\Big)dt=\frac{1}{2a}\int_{2a}^\infty K_0(y)dy  
\end{eqnarray*}
To compute the last integral, simply apply Proposition \ref{teobessel2} and Corollary \ref{lemaproduto}.
\end{proof}
\begin{proposition}
\label{teoerro}
Let $erf$ denote the error function given by $erf(z)=\frac{2}{\sqrt{\pi}}\int_0^z e^{-t^2}dt.$
Then
\begin{enumerate}[label=(\roman*)]
    \item $erf'(z)=\frac{2}{\sqrt{\pi}} e^{-z^2}$
    \item $lim_{z->+\infty}erf(z)=1$
    \item $\int e^{-a^2x^2-\frac{b^2}{x^2}}dx=\frac{\sqrt{\pi}}{4a}[e^{2ab}erf(ax+\frac{b}{x})+e^{-2ab}erf(ax-\frac{b}{x})]+C$ for $z>0$ and $|\text{ph}(a)|<\pi/4$  
\end{enumerate}
\end{proposition}
\begin{proof}
See \citet[Cap. 7]{TheHanbookMath}.
\end{proof}

\subsection*{\textbf{Appendix B:} Result of the simulation study}
\label{subsec: Ap1}
In this appendix, we show Tables and Figures of the simulation study in the section \ref{subsec:simu}.
\begin{figure}[h!]
\centering
    \includegraphics[width = 6.5cm]{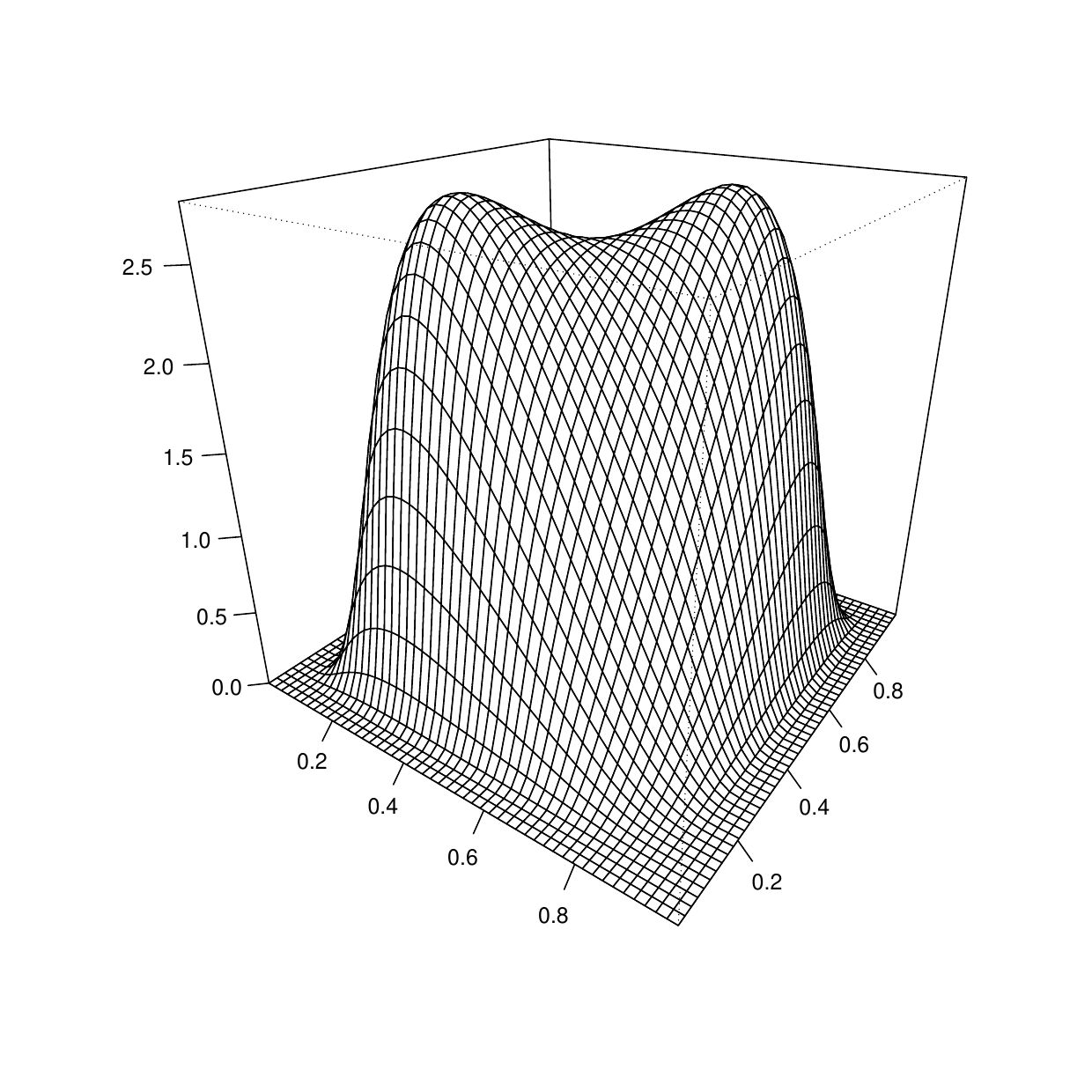}~
    \includegraphics[width = 6.5cm]{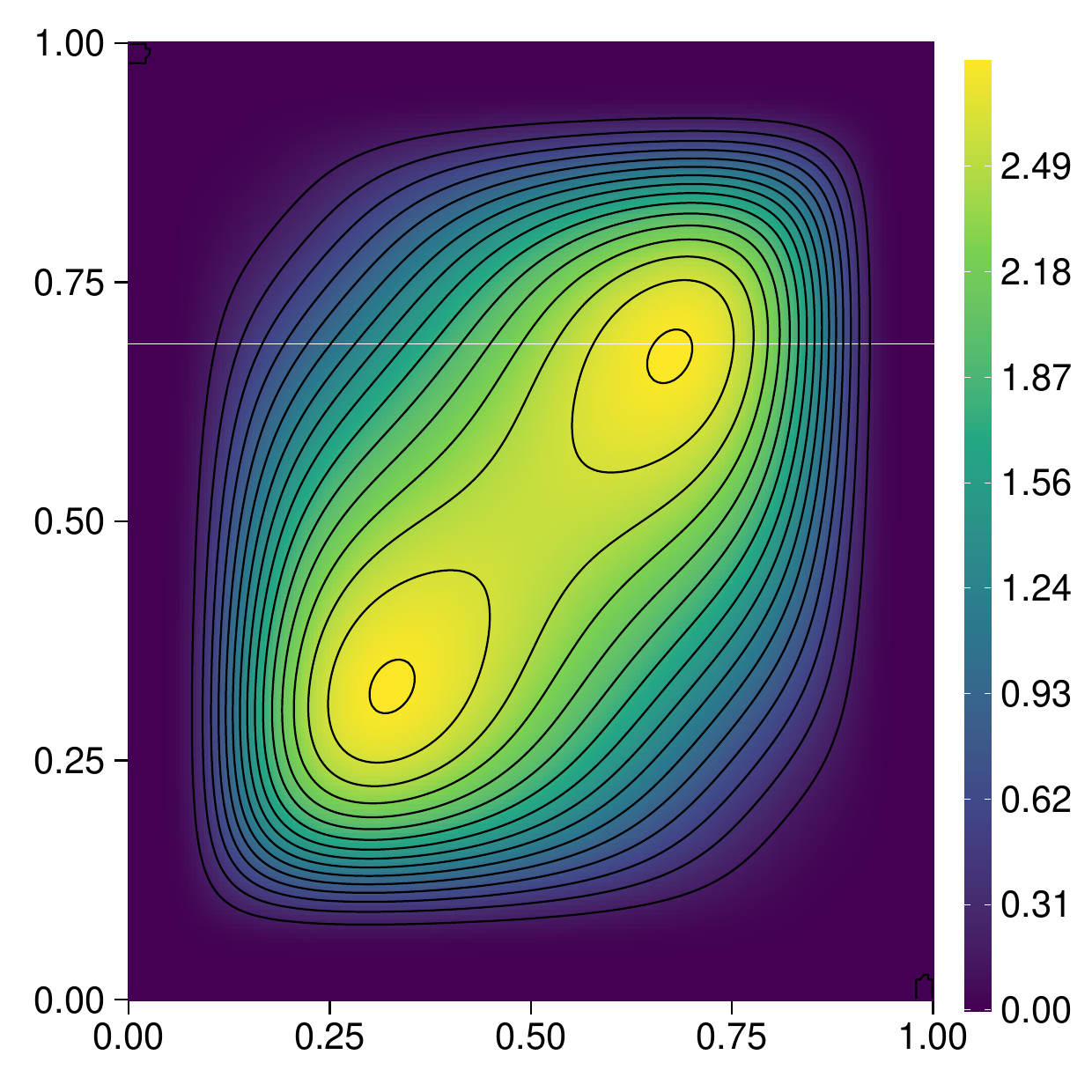}\\
    \includegraphics[width = 6.5cm]{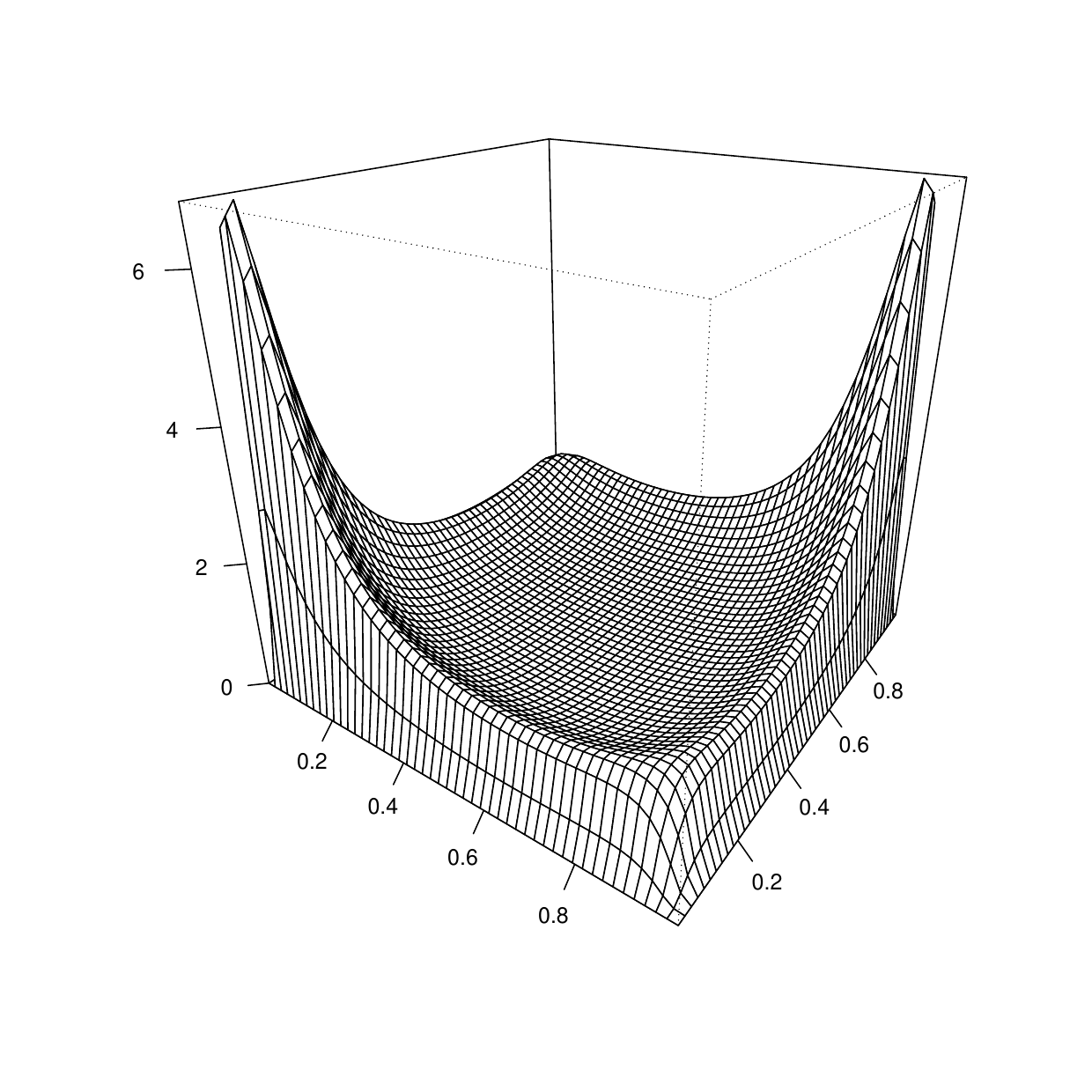}~
    \includegraphics[width = 6.5cm]{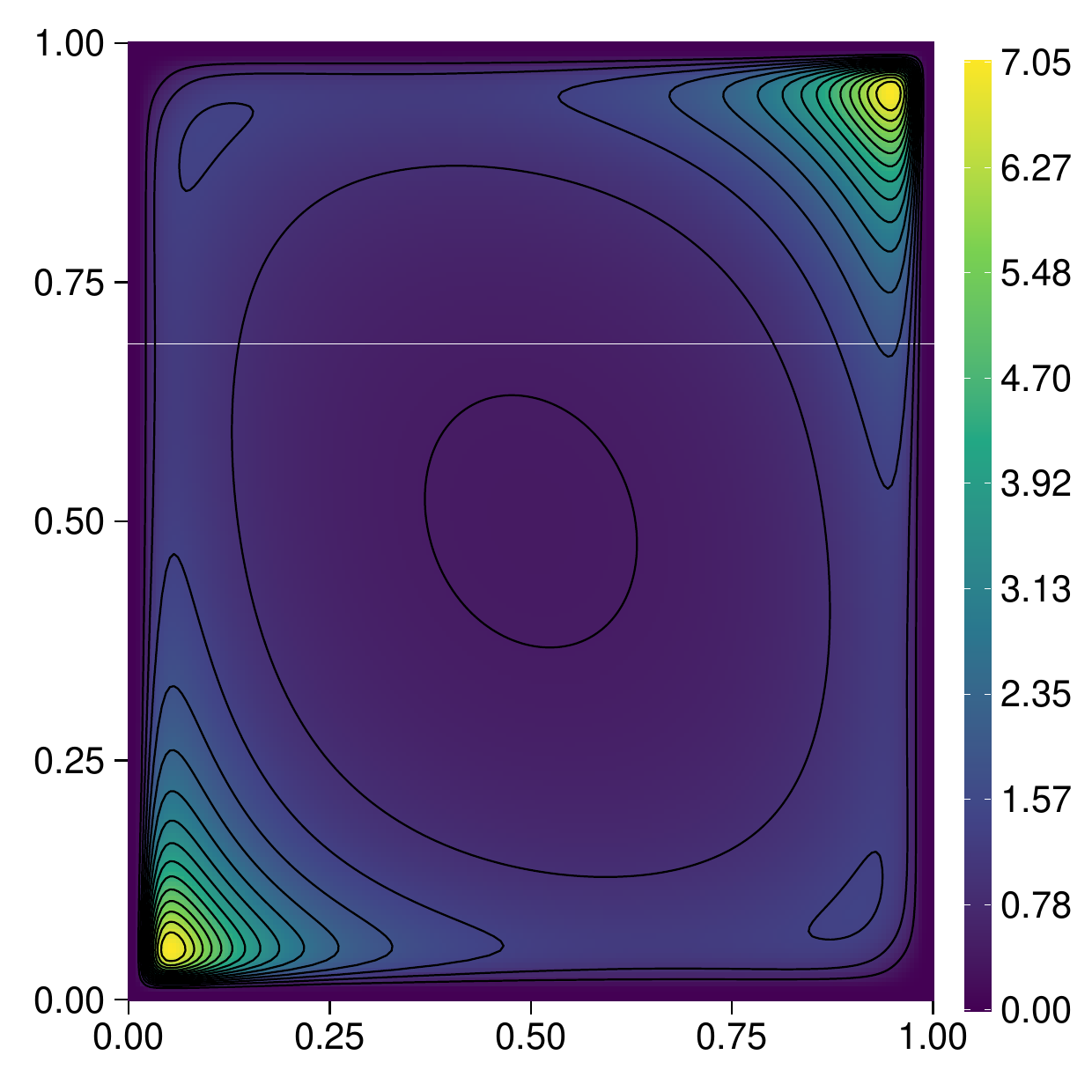}\\
    \includegraphics[width = 6.5cm]{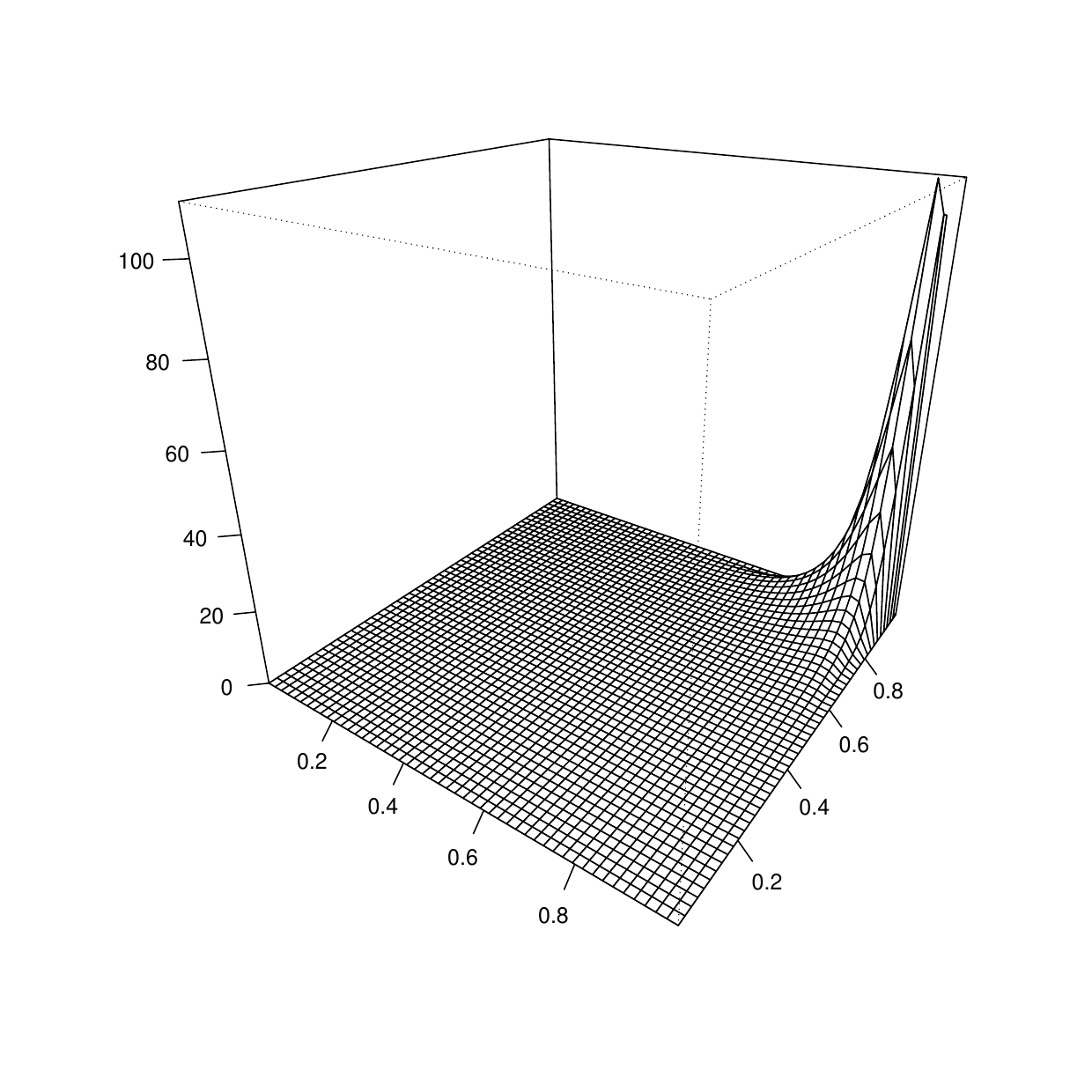}~
    \includegraphics[width = 6.5cm]{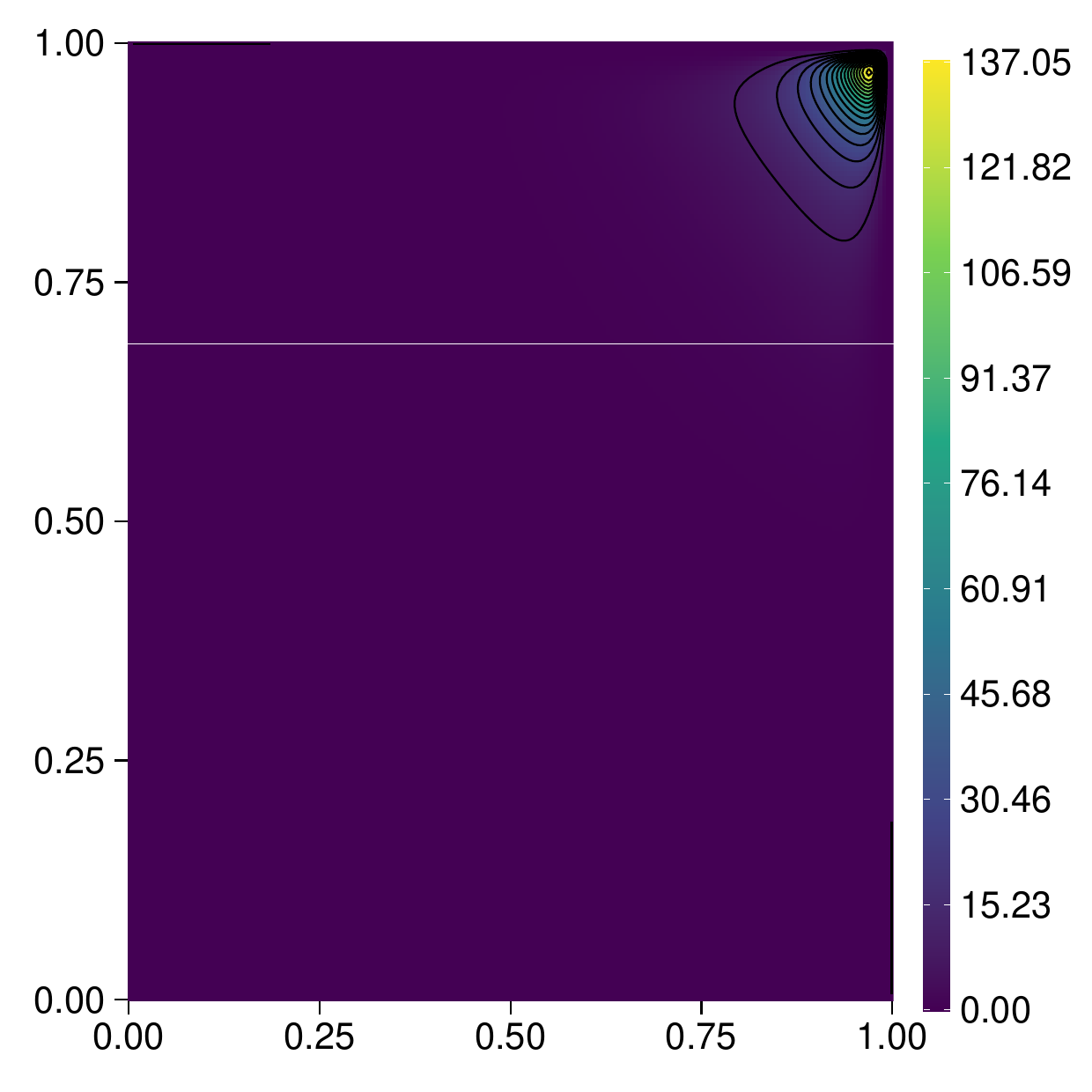}
    \caption{Surface and contour graphs for
    $\boldsymbol{\theta_1}$ = (0.5, 0.5, 2, 2, 1), 
    $\boldsymbol{\theta_2}$ = (0.5, 0.5, 5, 5, 1) and 
    $\boldsymbol{\theta_3}$ = (0.9, 0.9, $\sqrt{11}$, $\sqrt{11}$, 1), scenario 1.}
    \label{fig:superficie-contono1}
\end{figure}
\begin{table}[h!]
\centering{
    \caption{Mean, Bias, RMSE and Coverage of $95\%$ confidence intervals ($\%$) for $\bm{\theta} = (0.5, 0.5, 2, 2, 1)^{\top}$, Scenario 1.} 
\vspace{0.5cm}
{\small
\begin{tabular}{clRRRRR}
\hline
$n$  &   Medida   &  \mu_1  & \mu_2 & \sigma_{1}^2 & \sigma_{2}^2 & \lambda \\
\hline
\multicolumn{1}{c}{\multirow{4}{*}{50}}     
    &   Mean   &   0.501   &   0.501   &   1.954   &   1.950   &   0.850   \\ 
    &   Bias    &   0.001   &   0.001   &  -0.046   &  -0.050   &  -0.150   \\ 
    &   RMSE    &   0.027   &   0.027   &   0.202   &   0.208   &   0.284   \\ 
    &   Coverage      &   93.6    &   93.4    &   92.0    &   91.2    &   98.0    \\ 
\hline
\multicolumn{1}{c}{\multirow{4}{*}{100}}     
    &   Mean   &   0.500   &   0.500   &   1.970   &   1.975   &   0.902   \\ 
    &   Bias    &   0.000   &  -0.000   &  -0.030   &  -0.025   &  -0.098   \\ 
    &   RMSE    &   0.019   &   0.019   &   0.138   &   0.143   &   0.185   \\ 
    &   Coverage      &   94.9    &   94.8    &   94.3    &   93.2    &   98.1    \\ 
\hline
\multicolumn{1}{c}{\multirow{4}{*}{150}}     
    &   Mean   &   0.499   &   0.500   &   1.982   &   1.980   &   0.925   \\ 
    &   Bias    &  -0.001   &  -0.000   &  -0.018   &  -0.020   &  -0.075   \\ 
    &   RMSE    &   0.015   &   0.015   &   0.116   &   0.114   &   0.142   \\ 
    &   Coverage      &   94.8    &   94.5    &   93.7    &   94.7    &   98.5    \\ 
\hline
\multicolumn{1}{c}{\multirow{4}{*}{200}}     
    &   Mean   &   0.500   &   0.500   &   1.990   &   1.986   &   0.929   \\ 
    &   Bias    &   0.000   &   0.000   &  -0.010   &  -0.014   &  -0.071   \\ 
    &   RMSE    &   0.013   &   0.013   &   0.093   &   0.100   &   0.133   \\ 
    &   Coverage      &   95.9    &   93.9    &   95.3    &   94.0    &   97.6    \\ 
\hline
\multicolumn{1}{c}{\multirow{4}{*}{1000}}     
    &   Mean   &   0.501   &   0.500   &   1.995   &   1.996   &   0.971   \\ 
    &   Bias    &   0.001   &   0.000   &  -0.005   &  -0.004   &  -0.029   \\ 
    &   RMSE    &   0.006   &   0.006   &   0.046   &   0.044   &   0.053   \\ 
    &   Coverage      &   95.2    &   95.1    &   93.3    &   94.5    &   98.3    \\ 
\hline
\end{tabular}
}
\label{tab:scenario-results-1.1}}
\end{table}
\begin{table}[h!]
\centering{
    \caption{Mean, Bias, RMSE and Coverage of $95\%$ confidence intervals ($\%$) for $\bm{\theta} = (0.5, 0.5, 5, 5, 1)^{\top}$, Scenario 1.} 
\vspace{0.5cm}
{\small
\begin{tabular}{clRRRRR}
\hline
$n$  &   Medida   &  \mu_1  & \mu_2 & \sigma_{1}^2 & \sigma_{2}^2 & \lambda \\
\hline
\multicolumn{1}{c}{\multirow{4}{*}{50}}     
    &   Mean   &   0.502   &   0.500   &   4.884   &   4.926   &   0.864   \\ 
    &   Bias    &   0.002   &  -0.000   &  -0.116   &  -0.074   &  -0.136   \\ 
    &   RMSE    &   0.037   &   0.036   &   0.530   &   0.504   &   0.261   \\ 
    &   Coverage      &   93.8    &   95.5    &   85.9    &   88.1    &   98.5    \\ 
\hline
\multicolumn{1}{c}{\multirow{4}{*}{100}}     
    &   Mean   &   0.500   &   0.500   &   4.950   &   4.970   &   0.907   \\ 
    &   Bias    &   0.000   &   0.000   &  -0.050   &  -0.030   &  -0.093   \\ 
    &   RMSE    &   0.026   &   0.025   &   0.353   &   0.351   &   0.173   \\ 
    &   Coverage      &   94.2    &   95.5    &   86.8    &   87.2    &   99.0    \\ 
\hline
\multicolumn{1}{c}{\multirow{4}{*}{150}}     
    &   Mean   &   0.499   &   0.499   &   4.959   &   4.970   &   0.924   \\ 
    &   Bias    &  -0.001   &  -0.001   &  -0.041   &  -0.030   &  -0.076   \\ 
    &   RMSE    &   0.021   &   0.021   &   0.286   &   0.292   &   0.139   \\ 
    &   Coverage      &   94.6    &   95.1    &   87.7    &   84.8    &   98.6    \\ 
\hline
\multicolumn{1}{c}{\multirow{4}{*}{200}}     
    &   Mean   &   0.500   &   0.499   &   4.971   &   4.976   &   0.932   \\ 
    &   Bias    &  -0.000   &  -0.001   &  -0.029   &  -0.024   &  -0.068   \\ 
    &   RMSE    &   0.019   &   0.019   &   0.254   &   0.246   &   0.125   \\ 
    &   Coverage      &   94.1    &   94.4    &   85.4    &   85.9    &   98.3    \\ 
\hline
\multicolumn{1}{c}{\multirow{4}{*}{1000}}     
    &   Mean   &   0.500   &   0.500   &   4.990   &   4.995   &   0.971   \\ 
    &   Bias    &  -0.000   &  -0.000   &  -0.010   &  -0.005   &  -0.029   \\ 
    &   RMSE    &   0.008   &   0.008   &   0.113   &   0.110   &   0.052   \\ 
    &   Coverage      &   95.6    &   95.2    &   86.7    &   88.6    &   98.8    \\ 
\hline
\end{tabular}
}
\label{tab:scenario-results-1.2}}
\end{table}
\begin{table}[h!]
\centering{
    \caption{Mean, Bias, RMSE and Coverage of $95\%$ confidence intervals ($\%$) for $\bm{\theta} = (0.9, 0.9, \sqrt(11), \sqrt(11), 1)^{\top}$, Scenario 1.} 
\vspace{0.5cm}
\begin{small}
\begin{tabular}{clRRRRR}
\hline
$n$  &   Medida   &  \mu_1  & \mu_2 & \sigma_{1}^2 & \sigma_{2}^2 & \lambda \\
\hline
\multicolumn{1}{c}{\multirow{4}{*}{50}}     
    &   Mean   &   0.901   &   0.900   &   3.248   &   3.247   &   0.858   \\ 
    &   Bias    &   0.001   &   0.000   &  -0.069   &  -0.070   &  -0.142   \\ 
    &   RMSE    &   0.011   &   0.011   &   0.335   &   0.342   &   0.274   \\ 
    &   Coverage      &   93.4    &   93.0    &   93.5    &   91.7    &   98.3    \\
\hline
\multicolumn{1}{c}{\multirow{4}{*}{100}}     
    &   Mean   &   0.901   &   0.900   &   3.278   &   3.283   &   0.896   \\ 
    &   Bias    &   0.001   &   0.000   &  -0.038   &  -0.034   &  -0.104   \\ 
    &   RMSE    &   0.008   &   0.008   &   0.237   &   0.235   &   0.189   \\ 
    &   Coverage      &   93.7    &   95.1    &   92.1    &   94.4    &   98.5    \\ 
\hline
\multicolumn{1}{c}{\multirow{4}{*}{150}}     
    &   Mean   &   0.900   &   0.900   &   3.291   &   3.290   &   0.925   \\ 
    &   Bias    &   0.000   &   0.000   &  -0.026   &  -0.027   &  -0.075   \\ 
    &   RMSE    &   0.006   &   0.006   &   0.192   &   0.185   &   0.141   \\ 
    &   Coverage      &   94.6    &   93.9    &   93.7    &   93.8    &   98.2    \\ 
\hline
\multicolumn{1}{c}{\multirow{4}{*}{200}}     
    &   Mean   &   0.900   &   0.900   &   3.292   &   3.290   &   0.937   \\ 
    &   Bias    &   0.000   &  -0.000   &  -0.025   &  -0.027   &  -0.063   \\ 
    &   RMSE    &   0.005   &   0.006   &   0.168   &   0.161   &   0.121   \\ 
    &   Coverage      &   94.5    &   95.4    &   92.8    &   95.2    &   98.4    \\ 
\hline
\multicolumn{1}{c}{\multirow{4}{*}{1000}}     
    &   Mean   &   0.900   &   0.900   &   3.310   &   3.306   &   0.972   \\ 
    &   Bias    &   0.000   &   0.000   &  -0.007   &  -0.010   &  -0.028   \\ 
    &   RMSE    &   0.002   &   0.002   &   0.072   &   0.072   &   0.051   \\ 
    &   Coverage      &   95.4    &   96.2    &   95.2    &   94.6    &   98.5    \\ 
\hline
\end{tabular}
\end{small}
\label{tab:scenario-results-1.3}}
\end{table}

\begin{figure}[h!]
\centering
        \includegraphics[width = 6.5cm]{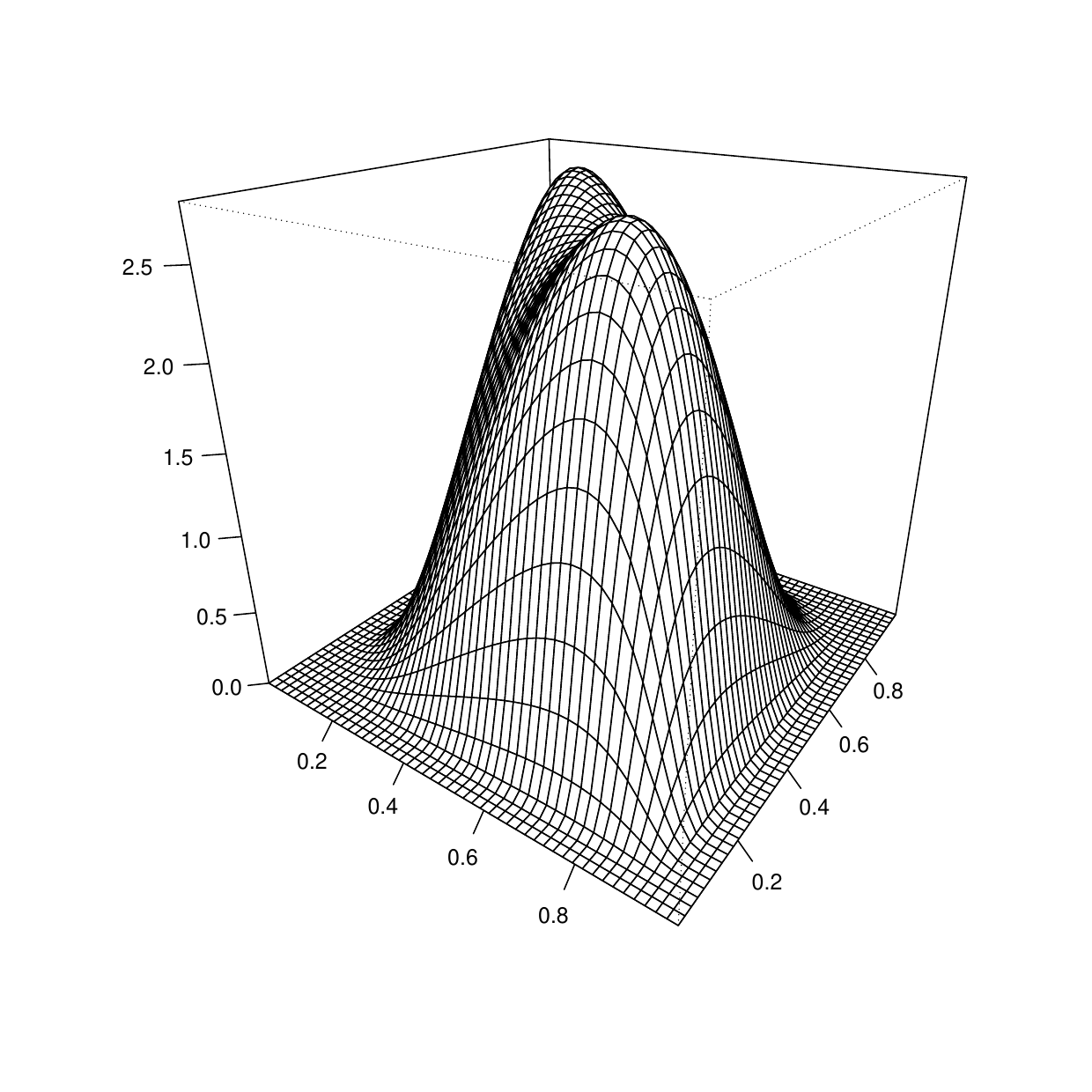}~
        \includegraphics[width = 6.5cm]{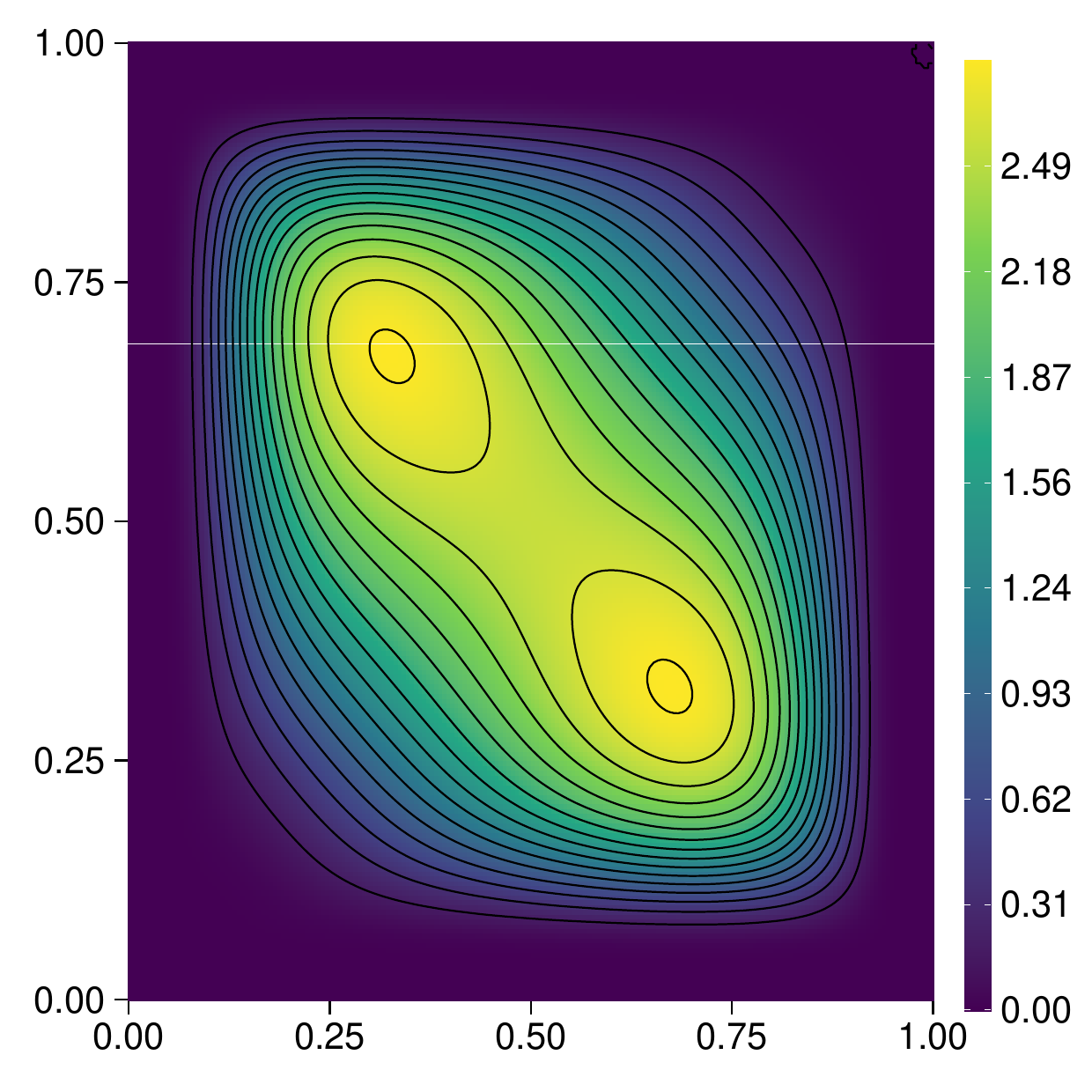}\\
        \includegraphics[width = 6.5cm]{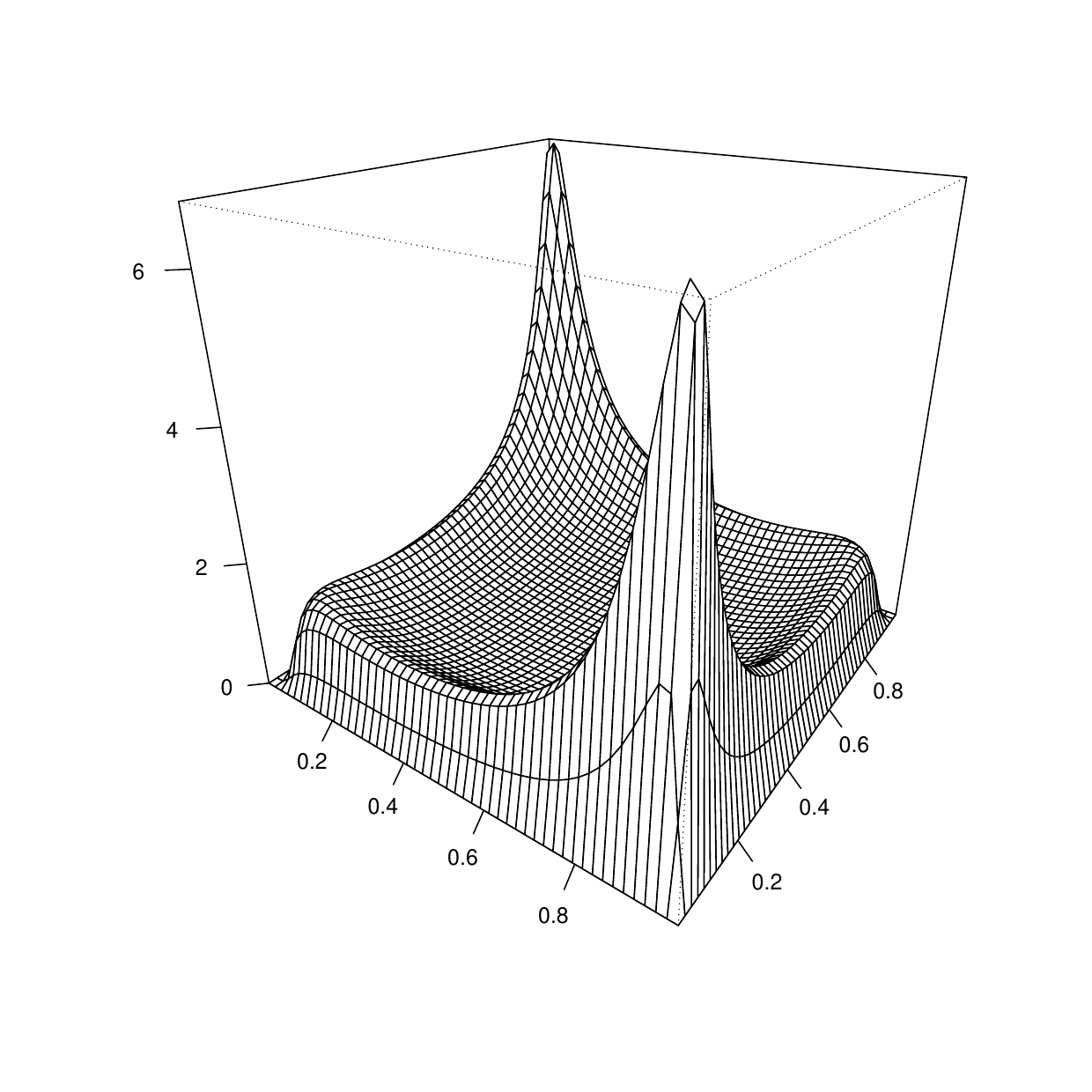}~
        \includegraphics[width = 6.5cm]{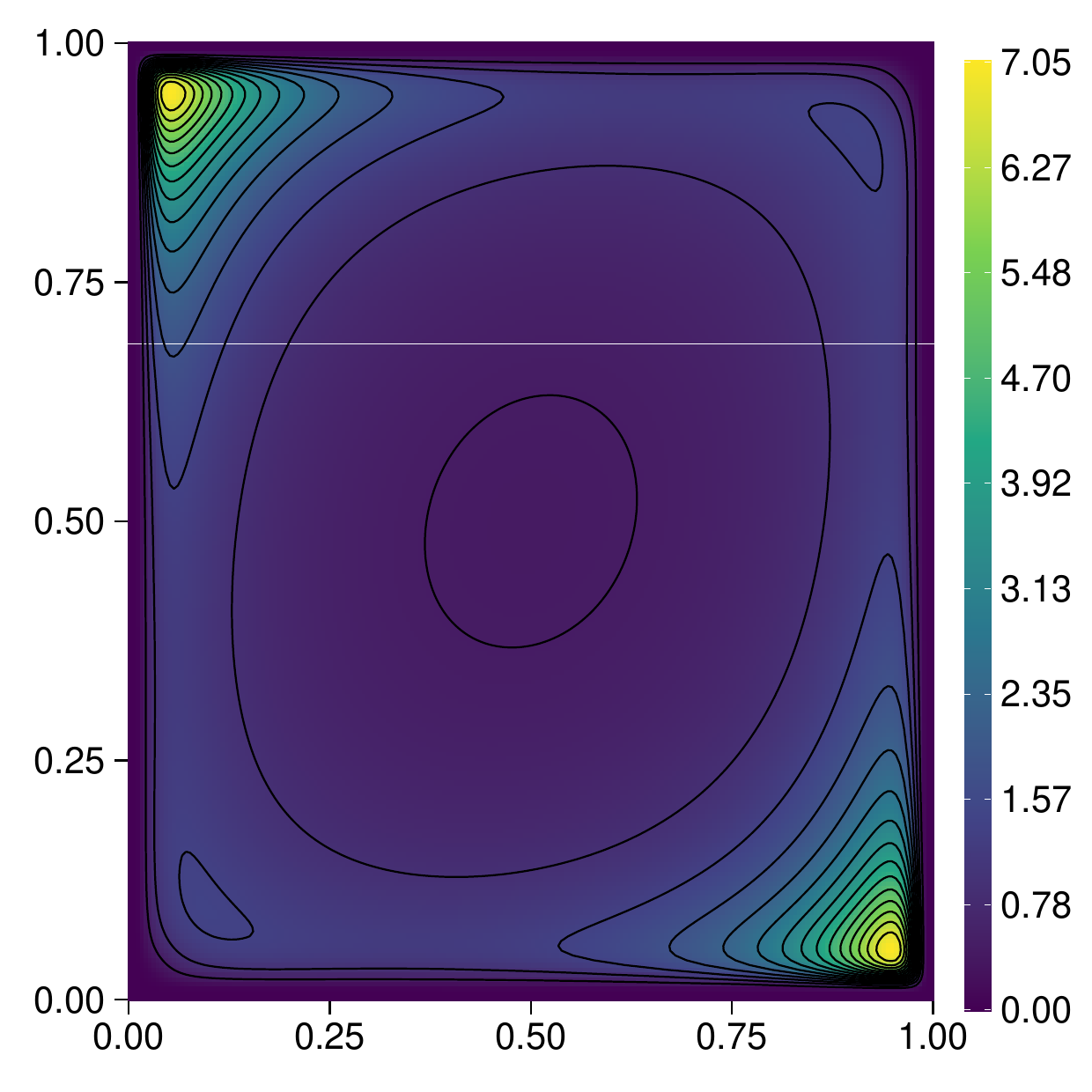}\\
        \includegraphics[width = 6.5cm]{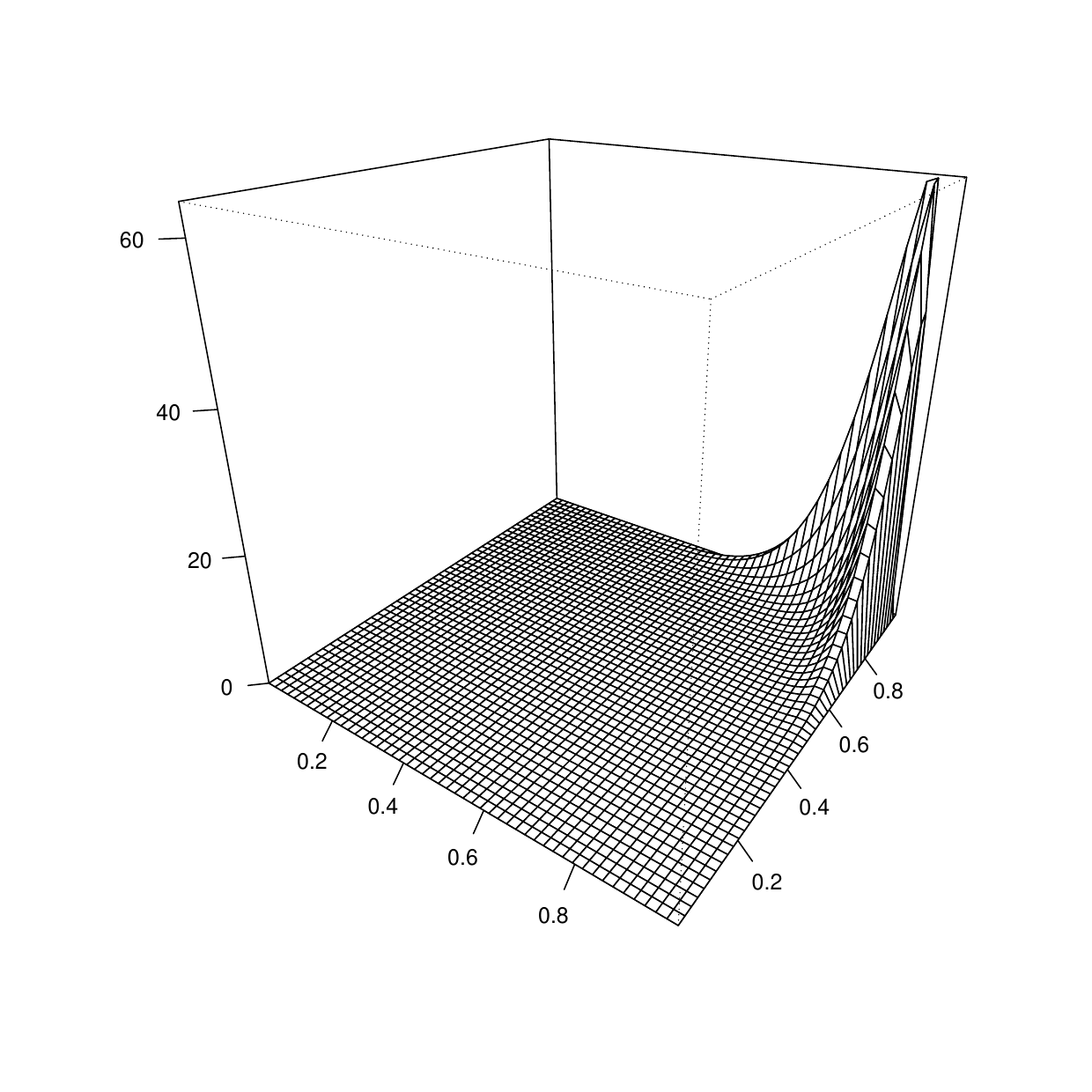}~
        \includegraphics[width = 6.5cm]{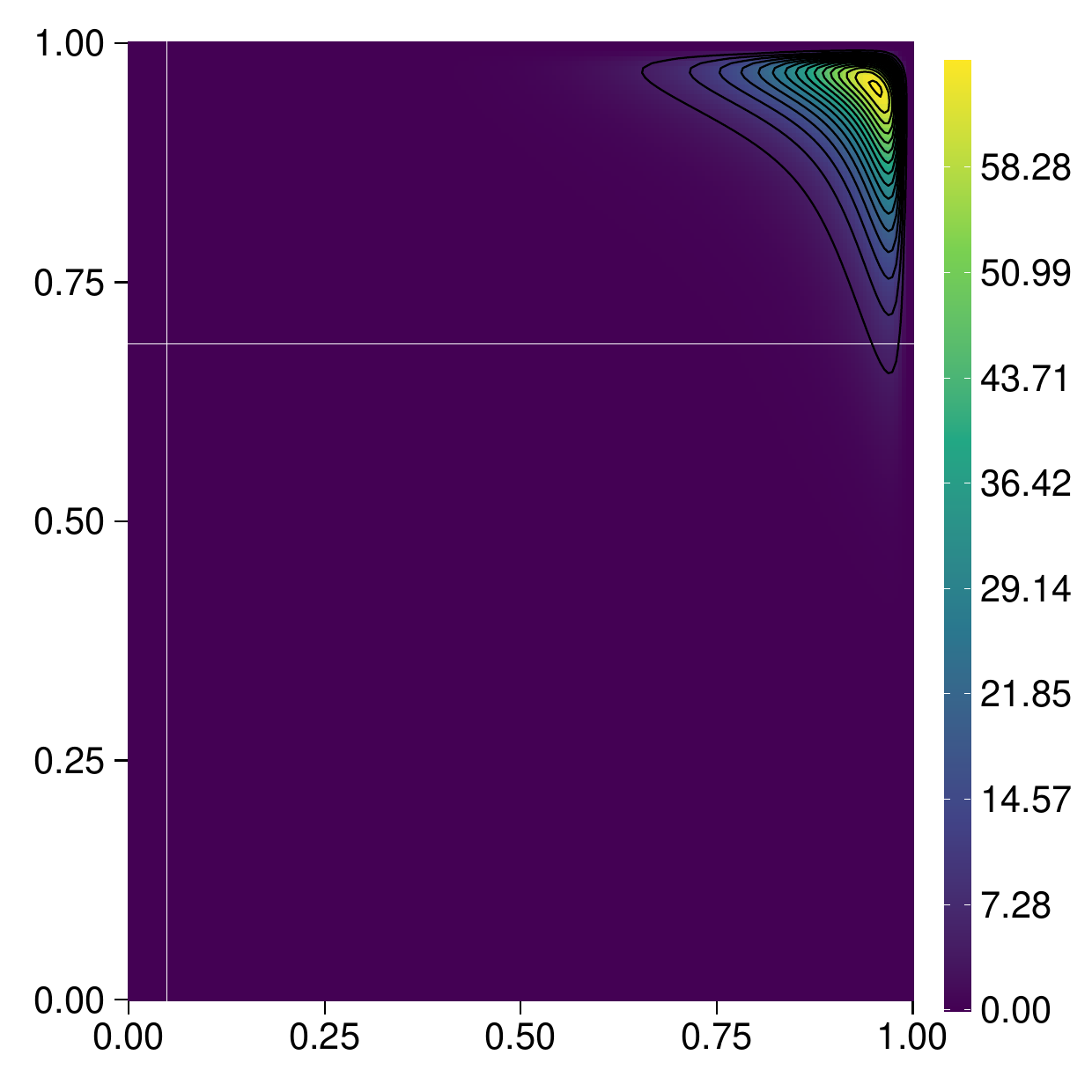}
    \caption{Surface and contour graphs for
    $\boldsymbol{\theta_1}$ = (0.5, 0.5, 2, 2, -1), 
    $\boldsymbol{\theta_2}$ = (0.5, 0.5, 5, 5, -1) and 
    $\boldsymbol{\theta_3}$ = (0.9, 0.9, $\sqrt{11}$, $\sqrt{11}$, -1), scenario 2.}
    \label{fig:superficie-e-contono-2}
\end{figure}
\begin{table}[h!]
\centering{
    \caption{Mean, Bias, RMSE and Coverage of $95\%$ confidence intervals ($\%$) for $\bm{\theta} = (0.5, 0.5, 2, 2,-1)^{\top}$, scenario 2.} 
\vspace{0.5cm}
\begin{small}
\begin{tabular}{clRRRRR}
\hline
$n$  &   Medida   &  \mu_1  & \mu_2 & \sigma_{1}^2 & \sigma_{2}^2 & \lambda \\
\hline
\multicolumn{1}{c}{\multirow{4}{*}{50}}     
    &   Mean   &   0.500   &   0.501   &   1.956   &   1.964   &  -0.855   \\ 
    &   Bias    &   0.000   &   0.001   &  -0.044   &  -0.036   &   0.145   \\ 
    &   RMSE    &   0.027   &   0.026   &   0.193   &   0.198   &   0.272   \\ 
    &   Coverage      &   94.6    &   94.5    &   92.8    &   93.0    &   98.4    \\ 
\hline
\multicolumn{1}{c}{\multirow{4}{*}{100}}     
    &   Mean   &   0.500   &   0.499   &   1.985   &   1.973   &  -0.910   \\ 
    &   Bias    &  -0.000   &  -0.001   &  -0.015   &  -0.027   &   0.090   \\ 
    &   RMSE    &   0.019   &   0.018   &   0.138   &   0.141   &   0.171   \\ 
    &   Coverage      &   94.2    &   94.6    &   95.1    &   93.6    &   98.6    \\ 
\hline
\multicolumn{1}{c}{\multirow{4}{*}{150}}     
    &   Mean   &   0.500   &   0.500   &   1.984   &   1.987   &  -0.924   \\ 
    &   Bias    &  -0.000   &  -0.000   &  -0.016   &  -0.013   &   0.076   \\ 
    &   RMSE    &   0.016   &   0.015   &   0.115   &   0.115   &   0.137   \\ 
    &   Coverage      &   95.1    &   94.9    &   93.9    &   93.4    &   98.8    \\ 
\hline
\multicolumn{1}{c}{\multirow{4}{*}{200}}     
    &   Mean   &   0.500   &   0.500   &   1.986   &   1.984   &  -0.940   \\ 
    &   Bias    &  -0.000   &   0.000   &  -0.014   &  -0.016   &   0.060   \\ 
    &   RMSE    &   0.012   &   0.013   &   0.102   &   0.099   &   0.112   \\ 
    &   Coverage      &   96.1    &   96.1    &   93.5    &   94.1    &   99.0    \\ 
\hline
\multicolumn{1}{c}{\multirow{4}{*}{1000}}     
    &   Mean   &   0.500   &   0.500   &   1.994   &   1.997   &  -0.975   \\ 
    &   Bias    &   0.000   &  -0.000   &  -0.006   &  -0.003   &   0.025   \\ 
    &   RMSE    &   0.006   &   0.006   &   0.046   &   0.046   &   0.047   \\ 
    &   Coverage      &   93.7    &   94.1    &   94.0    &   93.8    &   99.3    \\ 
\hline
\end{tabular}
\end{small}
\label{tab:scenario-results-2.1}}
\end{table}
\begin{table}[h!]
\centering{
    \caption{Mean, Bias, RMSE and Coverage of $95\%$ confidence intervals ($\%$) for $\bm{\theta} = (0.5, 0.5, 5, 5,-1)^{\top}$, scenario 2.} 
\vspace{0.5cm}
\begin{small}
\begin{tabular}{clRRRRR}
\hline
$n$  &   Medida   &  \mu_1  & \mu_2 & \sigma_{1}^2 & \sigma_{2}^2 & \lambda \\
\hline
\multicolumn{1}{c}{\multirow{4}{*}{50}}     
    &   Mean   &   0.500   &   0.499   &   4.870   &   4.888   &  -0.843   \\ 
    &   Bias    &   0.000   &  -0.001   &  -0.130   &  -0.112   &   0.157   \\ 
    &   RMSE    &   0.039   &   0.037   &   0.523   &   0.521   &   0.285   \\ 
    &   Coverage      &   93.3    &   94.1    &   86.8    &   86.0    &   97.6    \\ 
\hline
\multicolumn{1}{c}{\multirow{4}{*}{100}}     
    &   Mean   &   0.501   &   0.499   &   4.957   &   4.936   &  -0.903   \\ 
    &   Bias    &   0.001   &  -0.001   &  -0.043   &  -0.064   &   0.097   \\ 
    &   RMSE    &   0.027   &   0.025   &   0.343   &   0.355   &   0.186   \\ 
    &   Coverage      &   94.2    &   95.4    &   88.1    &   85.6    &   97.8    \\ 
\hline
\multicolumn{1}{c}{\multirow{4}{*}{150}}     
    &   Mean   &   0.499   &   0.500   &   4.968   &   4.944   &  -0.922   \\ 
    &   Bias    &  -0.001   &   0.000   &  -0.032   &  -0.056   &   0.078   \\ 
    &   RMSE    &   0.022   &   0.021   &   0.297   &   0.297   &   0.146   \\ 
    &   Coverage      &   93.6    &   94.8    &   85.4    &   83.8    &   98.4    \\ 
\hline
\multicolumn{1}{c}{\multirow{4}{*}{200}}     
    &   Mean   &   0.500   &   0.499   &   4.969   &   4.961   &  -0.935   \\ 
    &   Bias    &   0.000   &  -0.001   &  -0.031   &  -0.039   &   0.065   \\ 
    &   RMSE    &   0.019   &   0.018   &   0.252   &   0.255   &   0.118   \\ 
    &   Coverage      &   93.7    &   95.1    &   86.8    &   84.2    &   98.9    \\  
\hline
\multicolumn{1}{c}{\multirow{4}{*}{1000}}     
    &   Mean   &   0.500   &   0.500   &   4.990   &   4.982   &  -0.970   \\ 
    &   Bias    &   0.000   &   0.000   &  -0.010   &  -0.018   &   0.030   \\ 
    &   RMSE    &   0.008   &   0.008   &   0.111   &   0.112   &   0.053   \\ 
    &   Coverage      &   94.1    &   94.1    &   88.9    &   87.5    &   98.4    \\ 
\hline
\end{tabular}
\end{small}
\label{tab:scenario-results-2.2}}
\end{table}
\begin{table}[h!]
\centering{
    \caption{Mean, Bias, RMSE and Coverage of $95\%$ confidence intervals ($\%$) for $\bm{\theta} = (0.9, 0.9, \sqrt(11), \sqrt(11), -1)^{\top}$, scenario 2.} 
\vspace{0.5cm}
\begin{small}
\begin{tabular}{clRRRRR}
\hline
$n$  &   Medida   &  \mu_1  & \mu_2 & \sigma_{1}^2 & \sigma_{2}^2 & \lambda \\
\hline
\multicolumn{1}{c}{\multirow{4}{*}{50}}     
    &   Mean   &   0.900   &   0.901   &   3.238   &   3.252   &  -0.867   \\ 
    &   Bias    &   0.000   &   0.001   &  -0.079   &  -0.064   &   0.133   \\ 
    &   RMSE    &   0.011   &   0.011   &   0.344   &   0.336   &   0.251   \\ 
    &   Coverage      &   93.9    &   91.3    &   91.3    &   92.0    &   98.6    \\ 
\hline
\multicolumn{1}{c}{\multirow{4}{*}{100}}     
    &   Mean   &   0.900   &   0.900   &   3.269   &   3.275   &  -0.907   \\ 
    &   Bias    &   0.000   &   0.000   &  -0.047   &  -0.041   &   0.093   \\ 
    &   RMSE    &   0.008   &   0.008   &   0.234   &   0.238   &   0.173   \\ 
    &   Coverage      &   93.4    &   94.3    &   93.6    &   93.2    &   98.5    \\ 
\hline
\multicolumn{1}{c}{\multirow{4}{*}{150}}     
    &   Mean   &   0.900   &   0.900   &   3.298   &   3.278   &  -0.922   \\ 
    &   Bias    &   0.000   &   0.000   &  -0.019   &  -0.039   &   0.078   \\ 
    &   RMSE    &   0.006   &   0.006   &   0.188   &   0.196   &   0.143   \\ 
    &   Coverage      &   93.5    &   94.1    &   95.0    &   91.6    &   98.7    \\
\hline
\multicolumn{1}{c}{\multirow{4}{*}{200}}     
    &   Mean   &   0.900   &   0.900   &   3.303   &   3.304   &  -0.933   \\ 
    &   Bias    &   0.000   &   0.000   &  -0.014   &  -0.013   &   0.067   \\ 
    &   RMSE    &   0.006   &   0.005   &   0.164   &   0.163   &   0.121   \\ 
    &   Coverage      &   94.2    &   94.1    &   94.5    &   94.7    &   98.9    \\  
\hline
\multicolumn{1}{c}{\multirow{4}{*}{1000}}     
    &   Mean   &   0.900   &   0.900   &   3.309   &   3.310   &  -0.971   \\ 
    &   Bias    &   0.000   &   0.000   &  -0.008   &  -0.006   &   0.029   \\ 
    &   RMSE    &   0.003   &   0.002   &   0.074   &   0.074   &   0.053   \\ 
    &   Coverage      &   94.0    &   94.8    &   94.5    &   95.4    &   97.8    \\  
\hline
\end{tabular}
\end{small}
\label{tab:scenario-results-2.3}}
\end{table}

\begin{figure}[h!]
\centering
        \includegraphics[width = 6.5cm]{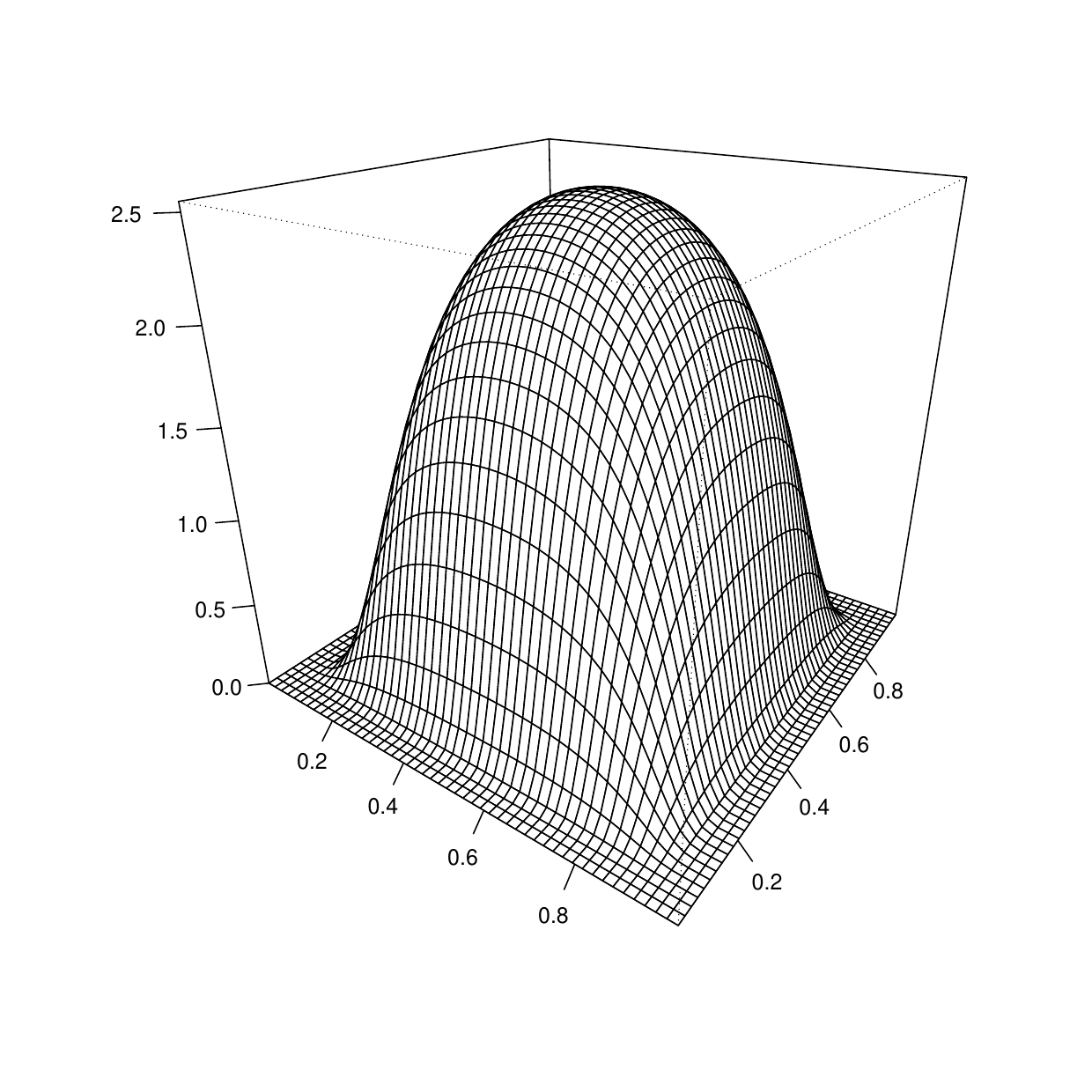}
        \includegraphics[width = 6.5cm]{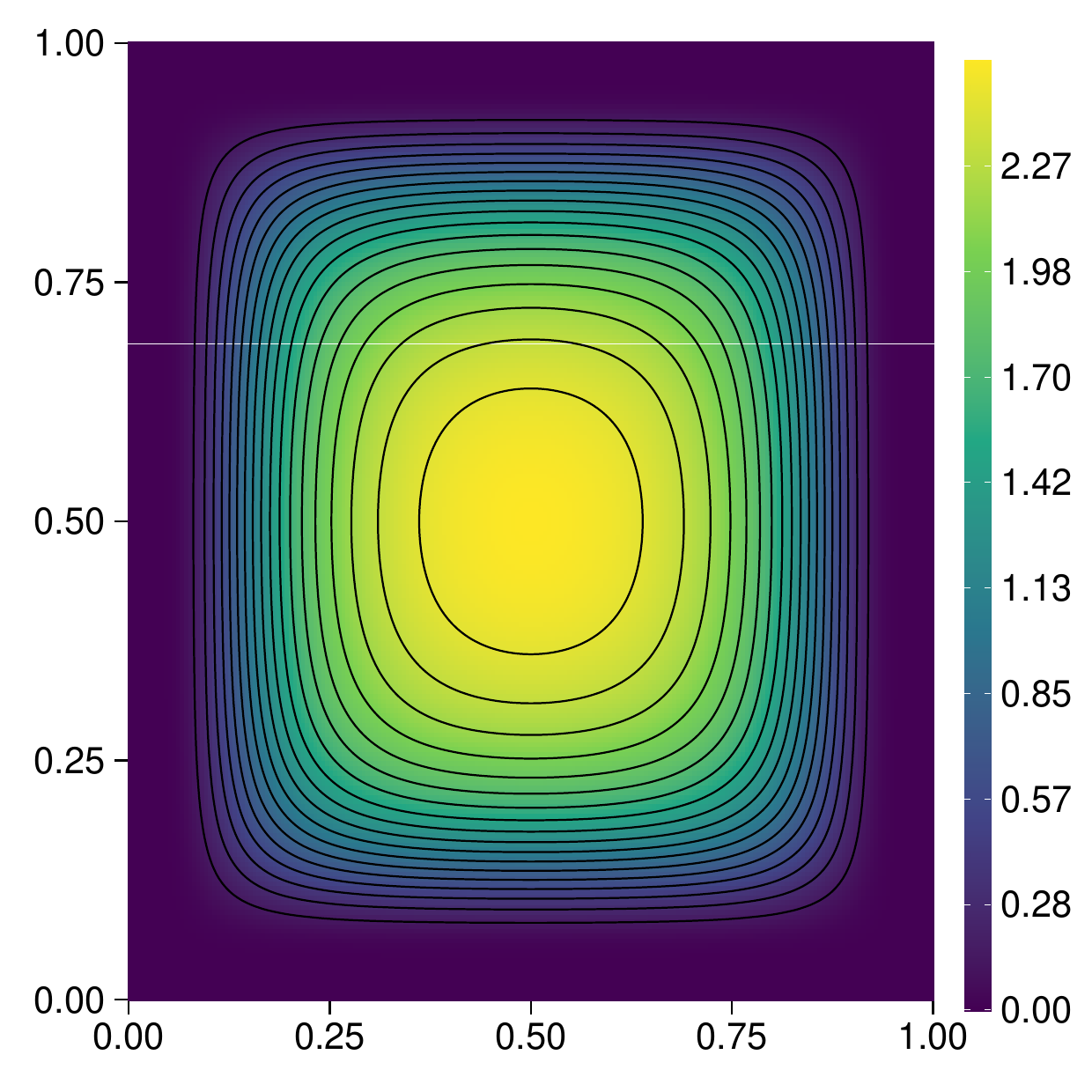}
        \includegraphics[width = 6.5cm]{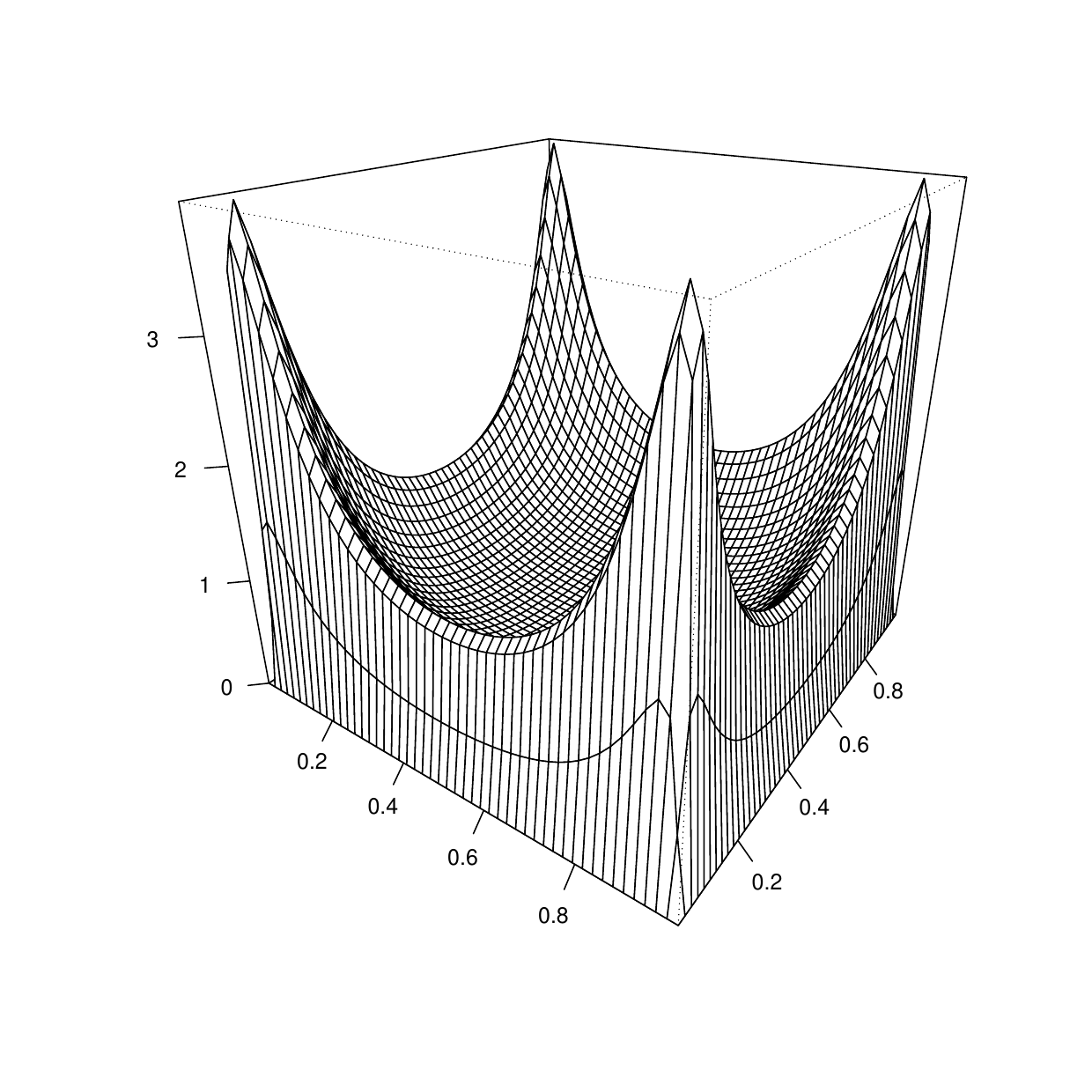}
        \includegraphics[width = 6.5cm]{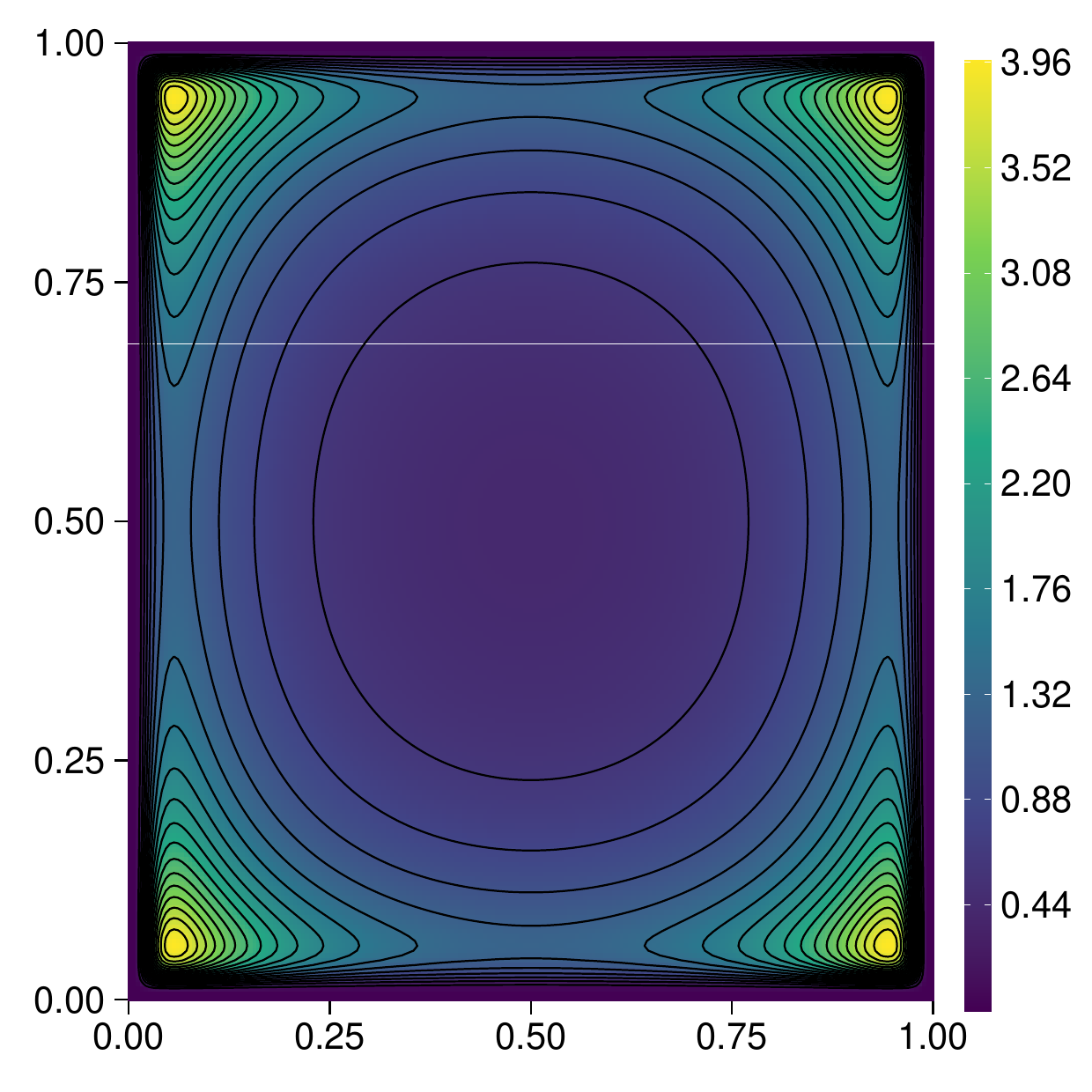}
        \includegraphics[width = 6.5cm]{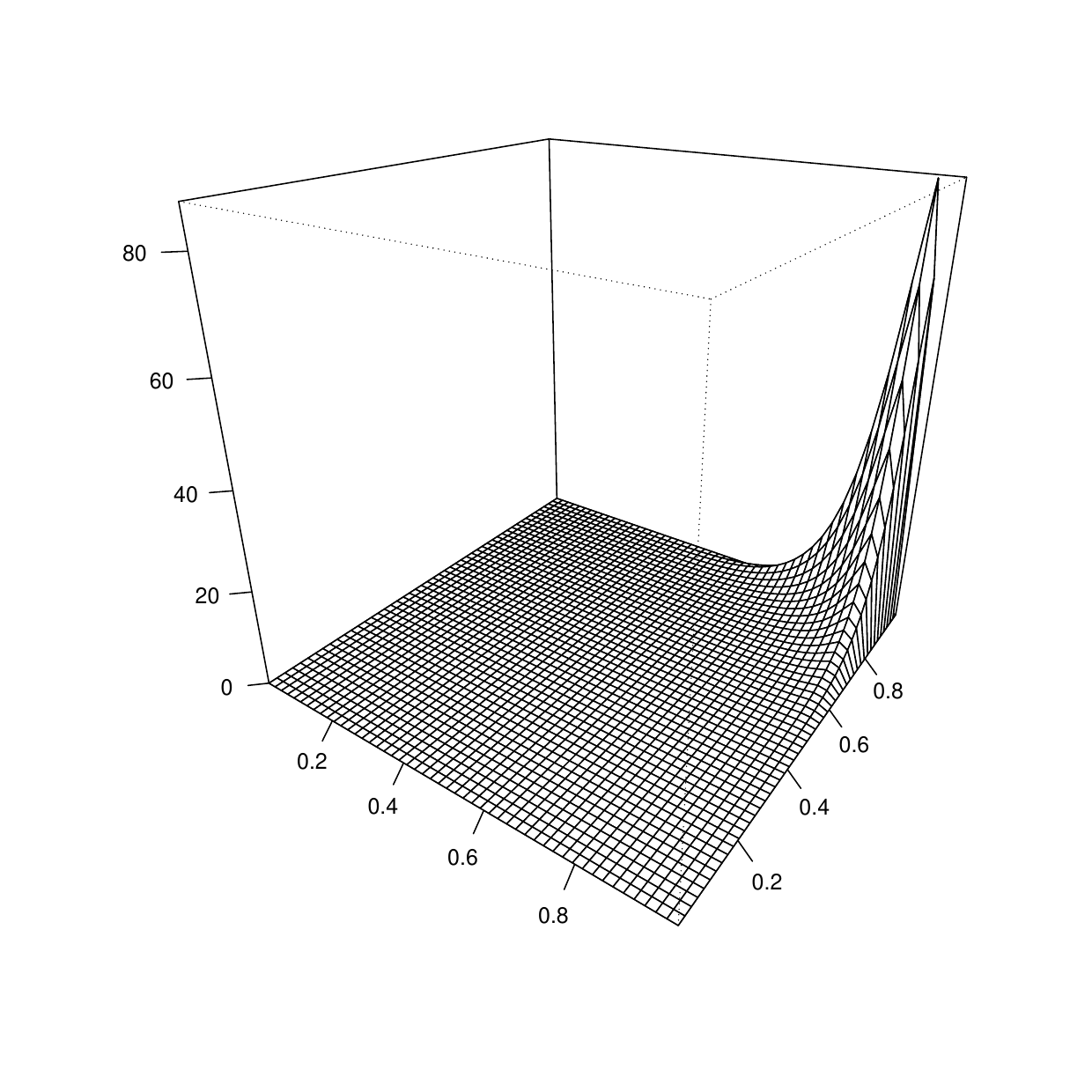}
        \includegraphics[width = 6.5cm]{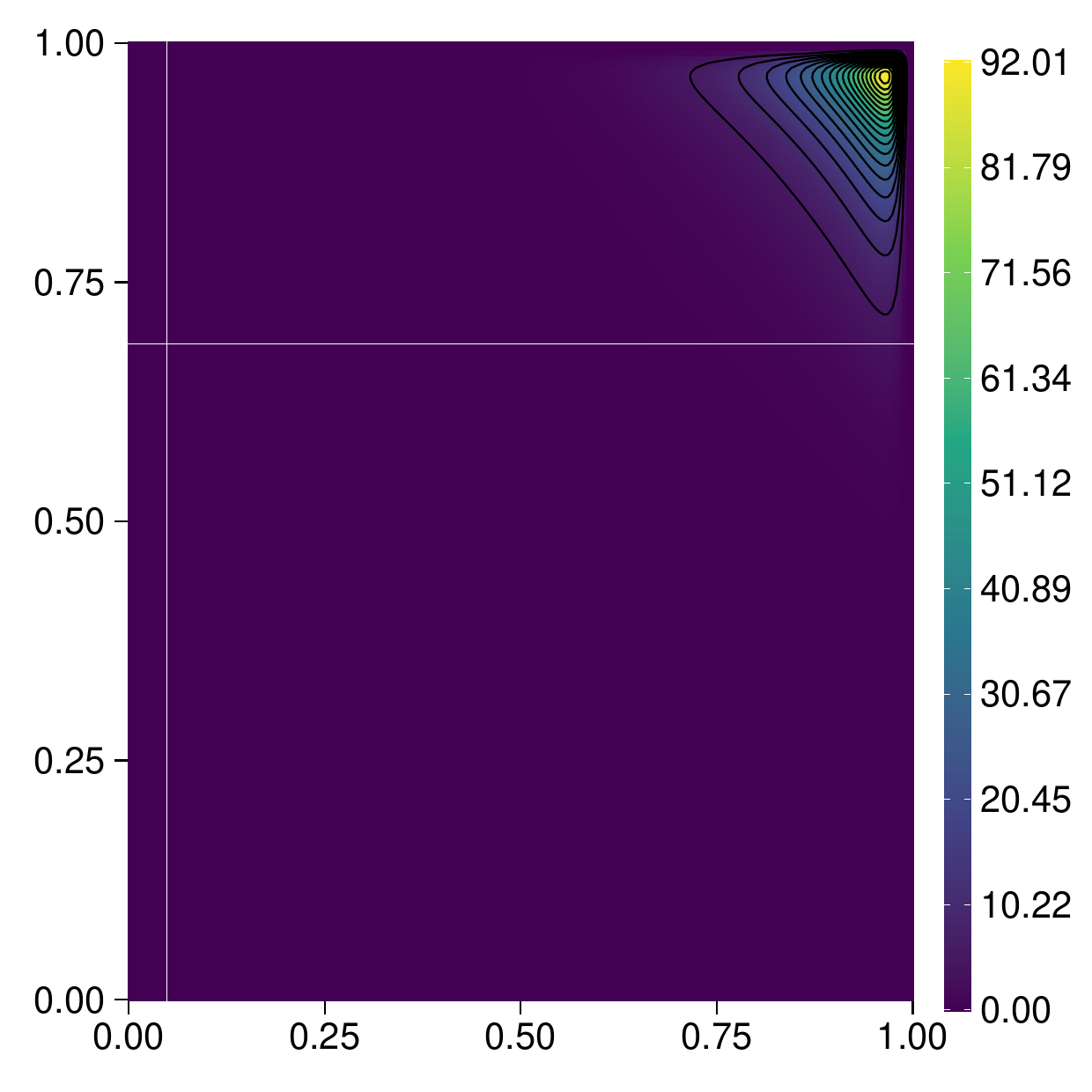}
    \caption{Surface and contour graphs for
    $\boldsymbol{\theta_1}$ = (0.5, 0.5, 2, 2, 0), 
    $\boldsymbol{\theta_2}$ = (0.5, 0.5, 5, 5, 0) and 
    $\boldsymbol{\theta_3}$ = (0.9, 0.9, $\sqrt{11}$, $\sqrt{11}$,0), scenario 3.}
    \label{fig:superficie-e-contono-3}
\end{figure}
\begin{table}[h!]
\centering{
    \caption{Mean, Bias, RMSE and Coverage of $95\%$ confidence intervals ($\%$) for $\bm{\theta} = (0.5, 0.5, 2, 2, 0)^{\top}$, scenario 3.} 
\vspace{0.5cm}
\begin{small}
\begin{tabular}{clRRRRR}
\hline
$n$  &   Medida   &  \mu_1  & \mu_2 & \sigma_{1}^2 & \sigma_{2}^2 & \lambda \\
\hline
\multicolumn{1}{c}{\multirow{4}{*}{50}}     
    &   Mean   &   0.501   &   0.500   &   1.963   &   1.984   &  -0.017   \\ 
    &   Bias    &   0.001   &   0.000   &  -0.037   &  -0.016   &  -0.017   \\ 
    &   RMSE    &   0.026   &   0.027   &   0.197   &   0.204   &   0.439   \\ 
    &   Coverage      &   94.3    &   94.0    &   92.9    &   93.1    &   93.0    \\ 
\hline
\multicolumn{1}{c}{\multirow{4}{*}{100}}     
    &   Mean   &   0.499   &   0.500   &   1.978   &   1.984   &   0.001   \\ 
    &   Bias    &  -0.001   &   0.000   &  -0.022   &  -0.016   &   0.001   \\ 
    &   RMSE    &   0.019   &   0.019   &   0.143   &   0.137   &   0.303   \\ 
    &   Coverage      &   94.9    &   95.6    &   93.8    &   94.7    &   93.2    \\ 
\hline
\multicolumn{1}{c}{\multirow{4}{*}{150}}     
    &   Mean   &   0.499   &   0.500   &   1.991   &   1.994   &   0.001   \\ 
    &   Bias    &  -0.001   &   0.000   &  -0.009   &  -0.006   &   0.001   \\ 
    &   RMSE    &   0.015   &   0.015   &   0.114   &   0.115   &   0.244   \\ 
    &   Coverage      &   95.3    &   95.4    &   93.5    &   95.8    &   93.0    \\
\hline
\multicolumn{1}{c}{\multirow{4}{*}{200}}     
    &   Mean   &   0.500   &   0.500   &   1.991   &   1.987   &  -0.009   \\ 
    &   Bias    &   0.000   &   0.000   &  -0.009   &  -0.013   &  -0.009   \\ 
    &   RMSE    &   0.014   &   0.013   &   0.098   &   0.097   &   0.222   \\ 
    &   Coverage      &   95.0    &   94.7    &   94.2    &   94.9    &   93.1    \\  
\hline
\multicolumn{1}{c}{\multirow{4}{*}{1000}}     
    &   Mean   &   0.500   &   0.500   &   1.997   &   1.999   &  -0.006   \\ 
    &   Bias    &   0.000   &   0.000   &  -0.003   &  -0.001   &  -0.006   \\ 
    &   RMSE    &   0.006   &   0.006   &   0.044   &   0.044   &   0.098   \\ 
    &   Coverage      &   94.8    &   95.6    &   95.2    &   94.6    &   93.6    \\  
\hline
\end{tabular}
\end{small}
\label{tab:scenario-results-3.1}}
\end{table}

\begin{table}[h!]
\centering{
    \caption{Mean, Bias, RMSE and Coverage of $95\%$ confidence intervals ($\%$) for $\bm{\theta} = (0.5, 0.5, 5, 5, 0)^{\top}$.} 
\vspace{0.5cm}
\begin{small}
\begin{tabular}{clRRRRR}
\hline
$n$  &   Medida   &  \mu_1  & \mu_2 & \sigma_{1}^2 & \sigma_{2}^2 & \lambda \\
\hline
\multicolumn{1}{c}{\multirow{4}{*}{50}}     
    &   Mean   &   0.501   &   0.500   &   4.923   &   4.936   &   0.018   \\ 
    &   Bias    &   0.001   &   0.000   &  -0.077   &  -0.064   &   0.018   \\ 
    &   RMSE    &   0.039   &   0.039   &   0.537   &   0.478   &   0.428   \\ 
    &   Coverage      &   93.3    &   92.5    &   88.7    &   91.9    &   93.2    \\ 
\hline
\multicolumn{1}{c}{\multirow{4}{*}{100}}     
    &   Mean   &   0.500   &   0.501   &   4.983   &   4.953   &  -0.005   \\ 
    &   Bias    &   0.000   &   0.001   &  -0.017   &  -0.047   &  -0.005   \\ 
    &   RMSE    &   0.026   &   0.027   &   0.358   &   0.355   &   0.302   \\ 
    &   Coverage      &   94.8    &   93.9    &   91.7    &   91.5    &   94.4    \\ 
\hline
\multicolumn{1}{c}{\multirow{4}{*}{150}}     
    &   Mean   &   0.501   &   0.499   &   4.974   &   4.978   &  -0.004   \\ 
    &   Bias    &   0.001   &  -0.001   &  -0.026   &  -0.022   &  -0.004   \\ 
    &   RMSE    &   0.023   &   0.022   &   0.295   &   0.298   &   0.247   \\ 
    &   Coverage      &   94.0    &   93.7    &   91.0    &   91.7    &   94.2    \\
\hline
\multicolumn{1}{c}{\multirow{4}{*}{200}}     
    &   Mean   &   0.500   &   0.500   &   4.986   &   4.986   &  -0.004   \\ 
    &   Bias    &   0.000   &   0.000   &  -0.014   &  -0.014   &  -0.004   \\ 
    &   RMSE    &   0.019   &   0.018   &   0.255   &   0.240   &   0.210   \\ 
    &   Coverage      &   94.9    &   95.5    &   93.1    &   93.9    &   94.9    \\  
\hline
\multicolumn{1}{c}{\multirow{4}{*}{1000}}     
    &   Mean   &   0.500   &   0.500   &   4.995   &   4.997   &  -0.006   \\ 
    &   Bias    &   0.000   &   0.000   &  -0.005   &  -0.003   &  -0.006   \\ 
    &   RMSE    &   0.008   &   0.008   &   0.112   &   0.113   &   0.093   \\ 
    &   Coverage      &   95.8    &   94.1    &   94.5    &   93.8    &   94.9    \\  
\hline
\end{tabular}
\end{small}
\label{tab:scenario-results-3.2}}
\end{table}

\begin{table}[h!]
\centering{
    \caption{Mean, Bias, RMSE and Coverage of $95\%$ confidence intervals ($\%$) for $\bm{\theta} = (0.9, 0.9, \sqrt{11}, \sqrt{11}, 0)^{\top}$.} 
\vspace{0.5cm}
\begin{small}
\begin{tabular}{clRRRRR}
\hline
$n$  &   Medida   &  \mu_1  & \mu_2 & \sigma_{1}^2 & \sigma_{2}^2 & \lambda \\
\hline
\multicolumn{1}{c}{\multirow{4}{*}{50}}     
    &   Mean   &   0.901   &   0.900   &   3.270   &   3.274   &  -0.008   \\ 
    &   Bias    &   0.001   &   0.000   &  -0.046   &  -0.043   &  -0.008   \\ 
    &   RMSE    &   0.011   &   0.012   &   0.330   &   0.348   &   0.423   \\ 
    &   Coverage      &   93.5    &   93.1    &   94.1    &   91.6    &   92.7    \\ 
\hline
\multicolumn{1}{c}{\multirow{4}{*}{100}}     
    &   Mean   &   0.900   &   0.900   &   3.297   &   3.290   &  -0.011   \\ 
    &   Bias    &   0.000   &   0.000   &  -0.020   &  -0.027   &  -0.011   \\ 
    &   RMSE    &   0.008   &   0.008   &   0.237   &   0.240   &   0.298   \\ 
    &   Coverage      &   93.7    &   93.9    &   93.2    &   93.6    &   94.2    \\ 
\hline
\multicolumn{1}{c}{\multirow{4}{*}{150}}     
    &   Mean   &   0.900   &   0.900   &   3.311   &   3.300   &  -0.016   \\ 
    &   Bias    &   0.000   &   0.000   &  -0.006   &  -0.016   &  -0.016   \\ 
    &   RMSE    &   0.007   &   0.006   &   0.192   &   0.186   &   0.253   \\ 
    &   Coverage      &   94.6    &   94.7    &   94.2    &   95.2    &   93.0    \\
\hline
\multicolumn{1}{c}{\multirow{4}{*}{200}}     
    &   Mean   &   0.900   &   0.900   &   3.301   &   3.312   &  -0.013   \\ 
    &   Bias    &   0.000   &   0.000   &  -0.015   &  -0.005   &  -0.013   \\ 
    &   RMSE    &   0.005   &   0.006   &   0.168   &   0.164   &   0.207   \\ 
    &   Coverage      &   95.0    &   93.4    &   94.7    &   95.4    &   94.4    \\  
\hline
\multicolumn{1}{c}{\multirow{4}{*}{1000}}     
    &   Mean   &   0.900   &   0.900   &   3.312   &   3.310   &  -0.007   \\ 
    &   Bias    &   0.000   &   0.000   &  -0.004   &  -0.007   &  -0.007   \\ 
    &   RMSE    &   0.003   &   0.003   &   0.072   &   0.072   &   0.096   \\ 
    &   Coverage      &   94.7    &   94.6    &   96.1    &   95.6    &   94.3    \\  
\hline
\end{tabular}
\end{small}
\label{tab:scenario-results-3.3}}
\end{table}
\end{document}